\documentclass[12pt]{article}

\usepackage[margin=0.75in]{geometry}
\usepackage{amsmath,amsthm,amssymb,amsfonts,float,tikz-cd,microtype,hyperref}

\newtheorem{theorem}{Theorem}[section]
\newtheorem{lemma}[theorem]{Lemma}
\newtheorem{prop}[theorem]{Proposition}

\newcommand{\R}{\mathbb{R}}
\newcommand{\C}{\mathbb{C}}
\newcommand{\N}{\mathbb{N}}
\newcommand{\X}{\mathcal{X}}
\newcommand{\Y}{\mathcal{Y}}
\newcommand{\g}{\mathbf{g}}
\newcommand{\h}{\mathbf{h}}
\newcommand{\V}[1]{\mathcal{V}_{\mathrm{#1}}(M)}
\newcommand{\ran}{\operatorname{Ran}}
\newcommand{\sol}{\operatorname{Sol}_{P}}
\newcommand{\sols}{\operatorname{Sol}_{P,\dot{C}^{\infty}}}
\newcommand{\solf}{\operatorname{Sol}_{P,f}}
\newcommand{\soly}{\operatorname{Sol}_{P,\Y^{-\infty}}}
\newcommand{\wf}{\mathrm{WF}_{\mathrm{sc}}}
\newcommand{\wfs}{\mathrm{WF}}
\newcommand{\lc}{\mathrm{LC}}
\newcommand{\supp}{\mathrm{supp}}
\newcommand{\diff}[2][]{\mathrm{Diff}^{#1}_{\mathrm{#2}}(M)}
\newcommand{\dv}{d\mathrm{vol}_{\g}}
\newcommand{\psc}[1]{\Psi_{\mathrm{sc}}^{#1}(M)}
\newcommand{\psco}{\Psi_{\mathrm{sc}}^{0,0}(M)}
\newcommand{\hsc}[1]{H_{\mathrm{sc}}^{#1}(M)}
\newcommand{\hb}[1]{H_{\mathrm{b}}^{#1}(M)}
\newcommand{\tsc}{{}^{\mathrm{sc}}TM}
\newcommand{\tscc}{{}^{\mathrm{sc}}T^*M}
\newcommand{\tsccc}{{}^{\mathrm{sc}}\overline{T^*M}}

\usepackage[backend=biber,sorting=nyt,giveninits=true]{biblatex}
\renewbibmacro{in:}{\ifentrytype{article}{}{\printtext{\bibstring{in}\intitlepunct}}}
\AtEveryBibitem{\clearfield{number}}
\DeclareFieldFormat{pages}{#1}
\usepackage[width=.85\textwidth]{caption}

\begin{document}

\title{An analogue of non-interacting quantum field theory in Riemannian signature}
\author{
    {Mikhail Molodyk\thanks{Stanford University, Department of Physics and Stanford Institute for Theoretical Physics, Stanford, CA 94305, USA. E-mail: mam765@stanford.edu} }
    and 
    {Andr\'as Vasy\thanks{Stanford University, Department of Mathematics, Stanford, CA 94305, USA. E-mail: andras@math.stanford.edu}}    
    }
    \date{}

\maketitle

    \begin{abstract}
 In this paper, we define an analogue of non-interacting quantum fields satisfying $(\Delta_{\g}-\lambda^2)\phi=0$ on a Riemannian scattering space $(M,\g)$ with two boundary components, i.e. a manifold with two asymptotically conic ends (meaning asymptotic to the ``large end" of a cone). Thus, the Lorentzian spacetime of usual QFT constructions is replaced by a Riemannian manifold with two boundary components, which play a role analogous to the two components of the mass shell for the Klein-Gordon field on Lorentzian spacetimes. Our main result describes a canonical construction of two-point functions satisfying a version of the Hadamard condition.
      \end{abstract}


\section{Introduction}
    In this paper, we define an analogue of non-interacting quantum fields satisfying $(\Delta_{\g}-\lambda^2)\phi=0$ on a Riemannian scattering space $(M,\g)$ with two boundary components, i.e. a manifold with two asymptotically conic ends (meaning asymptotic to the ``large end'' of a cone). The metric on such spaces has the form
    \begin{equation}
    \g=\frac{dx^2}{x^4}+\frac{\h}{x^2},
        \label{eq:metric-0}
    \end{equation}
    where $x$ is a boundary-defining function on $M$ (whose inverse is the radius of the asymptotic cone) and $\h$ is a symmetric smooth two-tensor on $M$ with nondegenerate and positive-definite restriction to the boundary. We use incoming and outgoing asymptotic data at infinity to construct two-point functions of quantum states satisfying a wavefront mapping property analogous to that of Hadamard states of the Klein--Gordon field. Our construction is a Euclidean analogue of the work of Vasy and Wrochna \cite{V-W} on the wave equation on asymptotically Minkowski spaces, though we only consider massive fields ($\lambda>0$) here and correspondingly our main analytic tool is the algebra of scattering pseudodifferential operators (whereas the massless case requires b-operators). While it is not an example of what is usually called Euclidean QFT, the algebraic structure being instead that of Lorentzian QFT, our model is an attempt to elucidate the nature of common (Lorentzian) QFT constructions by placing them in an unfamiliar Riemannian setting. For those more familiar with scattering theory, on the other hand, this model can serve as an introduction to some ideas of QFT from an unusual perspective.

    The characteristic set $\Sigma$ of $P=\Delta_{\g}-\lambda^2$ as a scattering operator lies over the boundary of $M$; consequently, when there are two boundary components, $\partial M=S=S_1\sqcup S_2$, the characteristic set has two connected components, $\Sigma=\Sigma_1\sqcup \Sigma_2$. They play a similar role to the positive and negative energy shells of the Klein--Gordon equation, but without a time orientation to distinguish them. The Hamilton vector field within $\Sigma$ governing the propagation of singularities of solutions to $Pu=0$ has a radial (critical) set $R$, leading to a source-sink structure of the flow. Then the Fredholm theory developed in \cite{Vasy-AH-KdS} and used in \cite{V-W} applies in essentially the same way, providing generalized inverses of $P$ acting between suitable weighted Sobolev spaces which are analogous to the retarded/advanced ($P_{\pm\mp}^{-1}$) and Feynman/anti-Feynman ($P_{\pm\pm}^{-1}$) propagators. While in the Lorentzian setting the retarded/advanced ones are often considered more fundamental due to causality considerations, these are absent in the Riemannian setting, and in fact the Feynman inverses are more natural from an analytic point of view (being generalizations of the limiting-absorption-principle resolvents). It can be shown that acting between the Feynman/anti-Feynman function spaces, $P$ is in fact always invertible and not just Fredholm, while in the retarded/advanced cases this is true for small enough values of $\lambda$ when $\dim M\geqslant 3$.
    
    In the algebraic approach to QFT, an algebra of quantum fields is built based on an appropriate space of test functions, the essential input being a choice of Hermitian form on this function space which prescribes the canonical commutation/anti-commutation relations (CCR/CAR) for fields; see Section~\ref{sec:alg-QFT} for details. In the standard construction for Klein--Gordon fields, this Hermitian form is given by the difference of the retarded and advanced propagators acting on the space of compactly supported smooth functions. In our setting, as in \cite{V-W} for the wave equation on asymptotically Minkowski spaces, we show that the difference $G_{+\pm}=P_{+\pm}^{-1}-P_{-\mp}^{-1}$ of \textit{either} pair of propagators, acting on the space $\dot{C}^{\infty}(M)$ of smooth functions on $M$ vanishing to infinite order at the boundary (the geometric generalization of Schwartz functions), can be used to define a field algebra in this way. Concretely, extending \cite{Vasy-positivity}, we show the following.

    \begin{theorem}
        $D_{+\pm}([f],[g])=i\langle f, G_{+\pm}g\rangle_{L^2_{\g}}$ are nondegenerate Hermitian sesquilinear forms on $\dot{C}^{\infty}(M)/P\dot{C}^{\infty}(M)$. In the Feynman $(++)$ case, the Hermitian form is positive.
    \end{theorem}
    
    Using either of these forms to define a field algebra, one obtains a quantum theory of fields on $M$ solving $Pu=0$, indexed by (or ``smeared with'') test functions in $\dot{C}^{\infty}(M)$.

    We proceed to define two-point functions of quantum states which are analogous to the in/out vacuum states on asymptotically Minkowski spacetimes. To do this, we use a parametrization of solutions to $Pu=0$ by incoming and outgoing asymptotic data at infinity (or, in the compactified picture, on the boundary of $M$), due in this geometric setting to Melrose \cite{Melrose-AES}. Given a boundary-defining function $x$ on $M$ with respect to which $\g$ has the form Eq.~(\ref{eq:metric-0}), any tempered distributional solution $u$ such that $\wf(u)\subset R$ (the space of which solutions we denote $\sol$) has the form
    \[u=x^{\frac{n-1}{2}}(e^{-\frac{i\lambda}{x}}a_M^++e^{\frac{i\lambda}{x}}a_M^-)\]
    for some $a_M^{\pm}\in C^{\infty}(M)$. Conversely, for any $a^+\in C^{\infty}(S)$ there exists a solution of this form with $a^+_M|_S=a^+$ (and similarly for $a^-_M$). This allows one to define boundary data maps $\rho^{\pm\pm}:u\mapsto a^{\pm}=a^{\pm}_M|_S$ and Poisson operators $\mathcal{U}_0^{\pm\pm}=(\rho^{\pm\pm})^{-1}$. Moreover, in the case when there are two boundary components, it is possible to consider ``mixed'' boundary data, that is the incoming data at one component and the outgoing data at the other. This results in maps $\rho^{\pm\mp}:u\mapsto \pi_1 a^{\pm}+\pi_2 a^{\mp}$ and $\mathcal{U}_0^{\pm\mp}=(\rho^{\pm\mp})^{-1}$, where $\pi_1,\pi_2$ are projections onto the corresponding boundary component.

    A definition of two-point functions amounts to a decomposition of the Hermitian form $iG_{+\pm}$ into a difference (in the bosonic/CCR case) or sum (in the fermionic/CAR case) of positive Hermitian forms $\Lambda_1,\Lambda_2$ which satisfy $P\Lambda_j=\Lambda_j P=0$. We accomplish this on the level of boundary data by decomposing a function on $S$ into its projections onto the two components. We take advantage of the following sequence of continuous linear maps:
    \[\begin{tikzcd}
   {\dot{C}^{\infty}(M)} & {\sol} & {C^{\infty}(S)} & {C^{\infty}(S)} & {C^{-\infty}(M),}
   	\arrow["G_{+\pm}", from=1-1, to=1-2]
   	\arrow["\rho^{\alpha\beta}", from=1-2, to=1-3]
    \arrow["\pi_j", from=1-3, to=1-4]
    \arrow["\mathcal{U}_0^{\alpha\beta}", from=1-4, to=1-5]
   \end{tikzcd}\]
   where $\alpha,\beta\in\{+,-\}$ and $j=1,2$, and $C^{-\infty}(M)$ is the space of distributions dual to $\dot{C}^{\infty}(M)$. The first two maps are isomorphisms of Hermitian inner product spaces, all three of which provide equivalent descriptions of the quantum theory: in terms of Schwartz functions on $M$ (as originally defined), in terms of sufficiently regular solutions to the field equation on $M$, and in terms of asymptotic/boundary data on $S$. This is an analogue of the more familiar isomorphisms between the space of test functions, the space of space-compact smooth solutions to the wave/Klein--Gordon equation, and the space of compactly supported smooth Cauchy data on any Cauchy surface, commonly used to give equivalent formulations of QFT on globally hyperbolic spacetimes (see e.g. \cite{Gerard-book}).
   
   Our construction of two-point functions requires the algebra to be bosonic when it is based on the retarded/advanced propagators (which is the usual case) and fermionic when based on the (anti-)Feynman propagators (for which there is no obvious physical interpretation). Our main result is the following.
    
    \begin{theorem}
        We define $\Lambda^{\alpha\beta}_{j}:\dot{C}^{\infty}(M)\to C^{-\infty}(M)$ for $\alpha,\beta\in \{+,-\}$ and $j=1,2$ by  
        \begin{equation}
            \Lambda_1^{\pm\pm} f=+i\mathcal{U}_0^{\pm\pm}\pi_1\rho^{\pm\pm}G_{++}f,
            \hskip 50pt
            \Lambda_2^{\pm\pm} f=+i\mathcal{U}_0^{\pm\pm}\pi_2\rho^{\pm\pm}G_{++}f,
        \end{equation}
        \begin{equation}
            \Lambda_1^{\pm\mp} f=+i\mathcal{U}_0^{\pm\mp}\pi_1\rho^{\pm\mp}G_{+-}f,
            \hskip 50pt
            \Lambda_2^{\pm\mp} f=-i\mathcal{U}_0^{\pm\mp}\pi_2\rho^{\pm\mp}G_{+-}f.
        \end{equation}
        The pairs $(\Lambda_1^{++},\Lambda_2^{++})$ and $(\Lambda_1^{--},\Lambda_2^{--})$ are two-point functions of states on the CAR algebra defined by the Hermitian form $D_{++}$, while the pairs $(\Lambda_1^{+-},\Lambda_2^{+-})$ and $(\Lambda_1^{-+},\Lambda_2^{-+})$ are two-point functions of states on the CCR algebra defined by the Hermitian form $D_{+-}$. The two-point functions extend continuously to $f\in C^{-\infty}(M)$ such that $\wf^{\infty,r}(f)\cap R=\varnothing$ for some $r>\frac{1}{2}$ and satisfy the wavefront mapping property
        \begin{equation}
            \wf(\Lambda^{\alpha\beta}_jf) \backslash R\subset \lc(\wf(f)\cap \Sigma_j).
        \end{equation}
    \end{theorem}

    Here $\lc(U)$, for $U\subset\Sigma$, is the union of bicharacteristics (integral curves of the Hamilton flow within $\Sigma$) whose closures intersect $U$. This wavefront mapping property is a scattering-theoretic version of the one obeyed by Hadamard two-point functions for the Klein--Gordon field (see Section \ref{sec:QFT}). The Hadamard condition is commonly used as a physical requirement on states in QFT.

    To summarize, the content of the theory is specified on the same function space by a choice of one of two Hermitian forms, depending on a choice of either the advanced/retarded or Feynman/anti-Feynman propagators as fundamental, corresponding to bosonic or fermionic theories. Once this is chosen, two-point functions can be constructed for either theory. We also briefly investigate the relationship between the global propagators defined via Fredholm theory and the ``mock propagators'' one can construct by mimicking the standard procedure for defining Feynman propagators associated to Hadamard states (in our case using the states just described). We find that the ``mock propagators'' differ from the corresponding global propagators by operators which are completely regularizing away from the radial set, but in general they do not coincide.

    The paper is organized as follows. In Section \ref{sec:QFT}, we summarize relevant notions from quantum field theory and their relationship with microlocal analysis. In Section \ref{sec:AES}, we introduce the geometric setting and review results from scattering theory which are relevant to our construction. In Section \ref{sec:fields-states}, we define the Hermitian structure that allows us to talk about CCR/CAR algebras and construct the two-point functions $\Lambda^{\alpha\beta}_j$. Finally, in Section \ref{sec:extension} we extend the two-point functions to distributions and prove the wavefront mapping property as well as define and discuss the ``mock propagators''.
    
    The authors gratefully acknowledge support from the National Science Foundation under grant numbers DMS-1953987, DMS-2247004 (AV) and PHY-2014215, PHY-2310429 (MM). The authors are also very grateful to Micha{\l} Wrochna for fruitful discussions and for comments on the manuscript, to Jan Derezi{\'n}ski, Albert Law, and Robert Wald for helpful discussions, and to the anonymous referee for comments that helped to improve the presentation of the results.

\section{QFT background}
\label{sec:QFT}
    
    We begin by summarizing relevant notions from quantum field theory in curved spacetime, drawing on the perspectives of the expositions \cite{D-G,Gerard-book,F-R}. A standard physicist-oriented introduction is \cite{Wald-book}, and important early rigorous results on the subject are due to Wald \cite{Wald-BH} and Kay \cite{Kay}.
    
    There is no accepted universal definition of QFT, with different aspects emphasized depending on the particular application; the structure important to us is expressed in a set of minimal algebraic requirements. The Haag--Kastler algebraic approach \cite{H-K}, first adapted to this purpose by Dimock \cite{Dimock}, has been especially useful for formulating QFT in curved spacetimes, where the Fourier transform, used extensively in QFT in Minkowski space, is in general no longer as helpful. Microlocal analysis, which provides tools for working ``in momentum space'', to an extent, even in curved spacetimes, is useful in this context for the same reason. Since we define our model on a curved manifold, this is the set of tools most relevant to us. However, since we work in Riemannian signature from the outset, we will not consider conditions related to causality, which are often also built into QFT constructions.

\subsection{Algebraic structure of QFT}
\label{sec:alg-QFT}
    In quantum field theories describing particle physics, the fields propagate on a Lorentzian spacetime, which in the framework of QFT on curved spacetime is considered to be fixed beforehand. However, many properties of quantum field theories can be described on the level of the algebra of local observables without reference to the spacetime. We will construct a model that is defined on a Riemannian (rather than Lorentzian) manifold $M$ but has the algebraic structure of a \textbf{free charged scalar field theory.} In this section, we describe this algebraic structure.

    We will need to fix a test function space $\mathcal{S}$ and introduce objects called $\phi(f),\phi^*(f)$ for all $f\in\mathcal{S}$. These objects represent what one usually calls ``smeared fields'' and formally writes as
    \[\phi(f)=\int_M \overline{f(x)}\phi(x)\ dx,\hskip 20pt \phi^*(f)=\int_M f(x)\phi^*(x)\ dx.\]
    The quantum field $\phi(x)$ represents a distribution on $M$ valued in operators (acting on quantum states), with the distribution itself satisfying a certain equation of motion $P\phi=0$. We will consider $P=\Delta_{\g}-\lambda^2$ for a metric $\g$ on $M$ and constant $\lambda>0$.
    
    We make the set of formal polynomials in the smeared fields into a unital *-algebra over $\C$ called $\text{CCR}^{\text{pol}}(\mathcal{S}, D)$ or $\text{CAR}^{\text{pol}}(\mathcal{S}, D)$ by prescribing relations:
    \begin{itemize}
        \item 
            To respect the above interpretation, we require that $f\mapsto \phi(f)$ is conjugate-linear and $f\mapsto\phi^*(f)$ linear; $\phi(f)^*=\phi^*(f)$; and $\phi$ ``satisfies the equation of motion in the distributional sense'', which in our case means $\phi(Pf)=0$ for all $f\in\mathcal{S}$ since $P$ is formally self-adjoint.
        \item 
            If we want \textbf{bosonic charged fields} ($\text{CCR}^{\text{pol}}(\mathcal{S}, D)$), we set $[\phi(f),\phi(g)]=[\phi^*(f),\phi^*(g)]=0$ and $[\phi(f),\phi^*(g)]=D(f,g)I$;
        \item 
            If we want \textbf{fermionic charged fields} ($\text{CAR}^{\text{pol}}(\mathcal{S}, D)$), we set $\{\phi(f),\phi(g)\}=\{\phi^*(f),\phi^*(g)\}=0$ and $\{\phi(f),\phi^*(g)\}=D(f,g)I$.
    \end{itemize}
    Since $[\phi(f),\phi(g)^*]=[\phi(g),\phi(f)^*]^*$ and $\{\phi(f),\phi(g)^*\}=\{\phi(g),\phi(f)^*\}^*$, we see that in either case $D(f,g)$ must be a Hermitian sesquilinear form on $\mathcal{S}$. In addition, consistency with the equation of motion requires that $D(Pf,g)=D(f,Pg)=0$ for any $f,g\in \mathcal{S}$, which means that $D$ in fact defines a Hermitian form on the quotient space $\mathcal{S}/P\mathcal{S}$. Thus, a QFT of this type is specified by a choice of manifold $M$, function space $\mathcal{S}$, and Hermitian form $D$ on $\mathcal{S}/P\mathcal{S}$.

    There is a degree of flexibility, within limits, in the choice of the space of ``smearing functions'' $\mathcal{S}$: the space should be large enough that the QFT can be used to model interesting situations, but the functions should be regular enough that any required additional structures are well-defined. On the other hand, the Hermitian form $D$ which appears in the CCR/CAR is the source of a theory's physical content. In standard models, $D$ is related to a symplectic (in the bosonic case) or Euclidean (in the fermionic case) inner product structure on the space of ``classical'' solutions to the equations of motion, but we will work directly with the Hermitian structure.
    
\subsection{States and two-point functions}
    
    A \textbf{state} of the QFT is any complex linear functional $\omega$ on the *-algebra satisfying $\omega(I)=1$ and $\omega(A^*A)\geqslant 0$ for all elements $A$. The values of $\omega$ on algebra elements which represent observable quantities are interpreted as expectation values for measurements of those quantities in that state.
    
    The two terms of the (anti-)commutators discussed above are not required to be multiples of the identity. In a state $\omega$, the \textbf{two-point functions} are the expectation values of those two terms:
    \begin{equation}
    \lambda_1(f,g)=\omega(\phi(f)\phi^*(g)),\hskip 20pt \lambda_2(f,g)=\omega(\phi^*(g)\phi(f)).
        \label{eq:2pt-fcn}
    \end{equation}
    Positivity of $\omega$ can be used to show that $\lambda_1,\lambda_2$ are positive Hermitian sesquilinear forms on $\mathcal{S}/P\mathcal{S}$. By construction, $\lambda_1\mp\lambda_2=D$ for bosonic ($-$) and fermionic ($+$) fields respectively.
    
    Conversely, given a Hermitian sesquilinear form $D$ and a pair of positive Hermitian sesquilinear forms $(\lambda_1,\lambda_2)$ satisfying $\lambda_1\mp\lambda_2=D$, there is a unique state $\omega$ satisfying Eq.~(\ref{eq:2pt-fcn}) in the class of \textbf{gauge-invariant quasi-free states} on $\text{CCR}^{\text{pol}}(\mathcal{S}, D)$ or $\text{CAR}^{\text{pol}}(\mathcal{S}, D)$; here gauge invariance refers to invariance under the global $U(1)$ symmetry transformation implemented by multiplying $\phi$ by an arbitrary global phase and $\phi^*$ by its conjugate -- see \cite{Gerard-book} for details. Therefore, the problem of defining such a state is equivalent to the problem of defining the two-point functions.

    In practice, one often deals with the integral kernels of the two-point functions. For concreteness, let $(M,\g)$ be a globally hyperbolic, oriented, and time-oriented Lorentzian manifold (without boundary) and take $\mathcal{S}=C_{\mathrm{c}}^{\infty}(M)$. This is the most standard setting for QFT and provides context and motivation for the paper's constructions. Nevertheless, the reader should keep in mind that our setting in the later sections, where the metric is positive definite and a boundary is present, is not completely analogous. If $\lambda_1,\lambda_2:C_{\mathrm{c}}^{\infty}(M)\times C_{\mathrm{c}}^{\infty}(M)\to \C$ are a pair of two-point functions which are continuous in addition to the required algebraic properties, then they are associated to continuous linear operators $\Lambda_1,\Lambda_2:C_{\mathrm{c}}^{\infty}(M)\to\mathcal{D}'(M)$ and Schwartz kernels $K_1,K_2\in\mathcal{D}'(M\times M)$ defined by $\lambda_j(f,g)=\Lambda_j(g)[\bar{f}]=K_j(\bar{f}\otimes g)$. The property that $\lambda_j(Pf,g) =\lambda_j(f,Pg)=0$ for all $f,g\in C_{\mathrm{c}}^{\infty}(M)$ translates to $\Lambda_jP=P\Lambda_j=0$ and $P_xK_j(x,y)=P_yK_j(x,y)=0$. By microlocal elliptic regularity, the latter property implies that $\wfs(K_1),\wfs(K_2)\subset \overline{\Sigma}\times \overline{\Sigma},$ where $\Sigma$ is the characteristic set of $P$.

    Typically, $P$ is a hyperbolic differential operator such as the wave, Klein--Gordon, or Dirac operator. In this case, $\Sigma$ has two connected components $\Sigma^{\pm}$ (the future and past dual light cones). The Hamilton vector field associated to the principal symbol of $P$ defines a flow on $\Sigma$, and the bicharacteristic relation $C\subset \Sigma\times\Sigma$ is defined as the set of pairs of points of $\Sigma$ which are connected by a bicharacteristic (an integral curve of the flow).

    The Hadamard condition, which specifies a class of gauge-invariant quasi-free states for which one can define the momentum-energy tensor, is commonly used as a physical requirement on states of QFT in curved spacetime. It was originally precisely formulated by Kay and Wald \cite{K-W} in terms of asymptotics of the two-point function and reformulated by Radzikowski \cite{Radzikowski} as a microlocal spectrum condition. For the charged scalar field, a state is called an \textbf{Hadamard state} if its two-point functions satisfy
    \begin{equation}
    \wfs(K_1)'\subset (\Sigma^+\times\Sigma^+)\cap C,
    \hskip 50pt
    \wfs(K_2)'\subset (\Sigma^-\times\Sigma^-)\cap C,
        \label{eq:Hadamard-kernel}
    \end{equation}
    where we use the notation $\wfs(K)'=\{((x,\xi),(y,-\eta))\in T^*M\times T^*M\ \mid\ ((x,\xi),(y,\eta))\in\wfs(K)\}$. The relationship of this form of the condition for charged (complex) fields to Radzikowski's original condition for neutral (real) fields is explained in \cite{Gerard-book}. $\Lambda_1,\Lambda_2$ can be continuously extended to $\mathcal{E}'(M)$, and Eq.~(\ref{eq:Hadamard-kernel}) implies that the extensions satisfy the wavefront mapping property
    \begin{equation}
    \wfs(\Lambda_1u)\subset \lc(\wfs(u)\cap\Sigma^+),
    \hskip 50pt
    \wfs(\Lambda_2u)\subset \lc(\wfs(u)\cap\Sigma^-),
    \label{eq:Hadamard-mapping}
    \end{equation}
    where for $U\subset \Sigma$ we define $\mathrm{LC}(U)$ to be the union of bicharacteristics passing through $U$. This, in turn, means that the two-point functions $\lambda_1,\lambda_2$ are well-defined for pairs of compactly supported distributions whose wavefront sets are not connected by any bicharacteristic in $\Sigma^+$ (for $\lambda_1$)  or $\Sigma^-$ (for $\lambda_2$). However, the condition Eq.~(\ref{eq:Hadamard-mapping}) is weaker than Eq.~(\ref{eq:Hadamard-kernel}), in particular because it loses information about rescaling of the singular momenta under the action of $\Lambda_j$: Schwartz kernels distinguish between points of the form $((x,a\xi),(y,\eta))$ for different $a>0$, whereas the mapping property does not. See \cite{D-Z,D-W} for a treatment of related issues.

\subsection{Propagators and distinguished parametrices}
    
    In this section, we recall the notions of propagators and distinguished parametrices in the setting of typical QFT constructions. Let $(M,\g)$ and $\mathcal{S}$ be as above. For any subset $U\subset M$, we write $J^{\pm}(U)$ for the set of points in the causal future ($+$) or past $(-)$ of $U$. By propagators, we will mean operators which are special inverses of $P$ (as opposed to these operators' Schwartz kernels, which, up to conventional factors of $i$, is the common usage in the physics literature).
    
    For hyperbolic differential operators $P$ such as the Klein--Gordon or Dirac operators, the retarded and advanced propagators $G_{\mathrm{R}}, G_{\mathrm{A}}:C_{\mathrm{c}}^{\infty}(M)\to C^{\infty}(M)$ are uniquely defined by the properties that $PG_{\mathrm{R}}(f)=PG_{\mathrm{A}}(f)=f$ and $\supp(G_{\mathrm{R}} f)\subset J^+(\supp(f))$, $\supp(G_{\mathrm{A}}f)\subset J^-(\supp(f))$ for any $f\in C_{\mathrm{c}}^{\infty}(M)$. We denote their kernels $K_{\mathrm{R}}$, $K_{\mathrm{A}}$. We refer to \cite{B-G-P} for a detailed treatment of wave equations on globally hyperbolic spacetimes.

    At this point the formalisms diverge for physical bosonic and fermionic theories, illustrated by the Klein--Gordon and Dirac cases. For the Klein--Gordon equation, $i(G_{\mathrm{R}}-G_{\mathrm{A}})$ defines a Hermitian form, while for the Dirac equation it is anti-Hermitian. This reflects the fact that the spaces of solutions to these equations have a natural symplectic (Klein--Gordon) or Euclidean (Dirac) structure. The following discussion is based on the Klein--Gordon case, since our model will have the same algebraic structure. Thus, we assume that the sesquilinear form $D$ corresponding to $i(G_{\mathrm{R}}-G_{\mathrm{A}})$ is Hermitian; it necessarily vanishes on $PC_{\mathrm{c}}^{\infty}(M)$, and we take it to define a CCR algebra.
    
    Consider again a pair of continuous two-point functions with kernels $K_1,K_2$ and associated operators $\Lambda_1,\Lambda_2$. We have $K_1-K_2=i(K_{\mathrm{R}}-K_{\mathrm{A}})$ by definition, so the support properties of $G_{\mathrm{R}}$ and $G_{\mathrm{A}}$ imply the following identities for any valid choice of time function $t$:
    \[K_{\mathrm{R}}(x,y)=-i\theta\big(t(x)-t(y)\big)(K_1-K_2),\hskip 20pt K_{\mathrm{A}}(x,y)=i\theta\big(t(y)-t(x)\big)(K_1-K_2).\]
    One often needs time-ordered two-point functions, whose Schwartz kernels are defined by
    \[
    K_{\mathrm{F}}(x,y)
    =
    -i\Big(
    \theta\big(t(x)-t(y)\big)K_1
    +\theta\big(t(y)-t(x)\big)K_2
    \Big),
    \]
    \[
    K_{\bar{\mathrm{F}}}(x,y)
    =
    i\Big(
    \theta\big(t(y)-t(x)\big)K_1
    +\theta\big(t(x)-t(y)\big)K_2
    \Big).
    \]
    These definitions imply the following relations, where the Feynman and anti-Feynman propagators $G_{\mathrm{F}},G_{\bar{\mathrm{F}}}:C_{\mathrm{c}}^{\infty}(M)\to \mathcal{D}'(M)$ are the operators with kernels $K_{\mathrm{F}},K_{\bar{\mathrm{F}}}$: 
    \begin{equation}
        i(G_{\mathrm{R}}-G_{\mathrm{A}})=\Lambda_1-\Lambda_2,\hskip 20pt
        i(G_{\mathrm{F}}-G_{\bar{\mathrm{F}}})=\Lambda_1+\Lambda_2,\hskip 20pt
        G_{\mathrm{R}}+G_{\mathrm{A}}=G_{\mathrm{F}}+G_{\bar{\mathrm{F}}},
        \label{eq:props-bose}
    \end{equation}
    \begin{equation}
        G_{\mathrm{F}}=G_{\mathrm{R}}-i\Lambda_2=G_{\mathrm{A}}-i\Lambda_1,\hskip 20pt G_{\bar{\mathrm{F}}}=G_{\mathrm{R}}+i\Lambda_1=G_{\mathrm{A}}+i\Lambda_2,
        \label{eq:props-bose-c-to-f}
    \end{equation}
    \begin{equation}
        G_{\mathrm{R}}=G_{\mathrm{F}}+i\Lambda_2=G_{\bar{\mathrm{F}}}-i\Lambda_1,\hskip 20pt G_{\mathrm{A}}=G_{\mathrm{F}}+i\Lambda_1=G_{\bar{\mathrm{F}}}-i\Lambda_2.
        \label{eq:props-bose-f-to-c}
    \end{equation}

    From Eq.~(\ref{eq:props-bose}) we see that $i(G_{\mathrm{F}}-G_{\bar{\mathrm{F}}})$ is also a Hermitian form vanishing on $PC_{\mathrm{c}}^{\infty}(M)$, and that it is in fact positive. Moreover, we see that $(\Lambda_1,\Lambda_2)$ is a valid pair of two-point functions, now fermionic, on $\mathrm{CAR}^{\mathrm{pol}}(C_{\mathrm{c}}^{\infty}(M),i(G_{\mathrm{F}}-G_{\bar{\mathrm{F}}}))$. The physical interpretation of this fact is unclear: first of all, this is a fermionic algebra of spin-zero fields whose classical phase space is symplectic; also, the difference $G_{\mathrm{F}}-G_{\bar{\mathrm{F}}}$ does not vanish at spacelike separation, so fields whose anticommutators are given by this difference do not satisfy usual causality restrictions. Anticommuting spin-zero fields can exist in nonrelativistic theories (in which case Einstein causality also does not apply) \cite{Levy-Leblond}; anticommuting scalar ghost fields are also introduced in BRST quantization of gauge theories \cite{K-O,Krahe}. In any case, as far as this level of structure is concerned, we can use any Hermitian form to define either bosonic or fermionic fields (formally defined by the choice of CCR vs.\ CAR algebra), regardless of which equation and which propagators the form came from.

    As a pair, we call $(G_{\mathrm{R}},G_{\mathrm{A}})$ the \textbf{causal propagators} and $(G_{\mathrm{F}},G_{\bar{\mathrm{F}}})$ the \textbf{Feynman propagators}. Given either pair of propagators and the two-point functions, we can reconstruct the other pair using Eqs.~(\ref{eq:props-bose-c-to-f})-(\ref{eq:props-bose-f-to-c}) without reference to kernels or time-ordering (which will be important in the Riemannian setting). Note that the Feynman propagators are defined for a particular state; however, in Section~\ref{sec:dist-inv} we will describe an alternative notion of propagators under which all four are state-independent, and in the body of the paper we will consider the possibility that the Feynman propagators are fundamental as opposed to derived from the causal ones.

    The propagators' place in the greater mathematical picture is further clarified by the notion of distinguished parametrices introduced by Duistermaat and H\"ormander \cite{D-H}. A parametrix for a differential operator $P$ is a continuous linear operator $Q:C_{\mathrm{c}}^{\infty}(M)\to C^{\infty}(M)$ such that $PQ-I$ and $QP-I$ are smoothing operators (equivalently, they have smooth Schwartz kernels). Corresponding to two ways to choose the direction of propagation of singularities in each component of $\Sigma$, there are four distinguished classes of parametrices for $P$. We let $o\subset T^*M$ denote the zero section and $\Delta=\{((x,\xi),(x,\xi))\in (T^*M\backslash o)\times(T^*M\backslash o)\}$. Then a parametrix with kernel $K_Q$ is called
    \begin{itemize}
        \item retarded, if $\wfs(K_Q)'\subset \Delta\cup \{((x,\xi),(y,\eta))\in C\ \mid\ x\in J^+(y)\}$;
        \item advanced, if $\wfs(K_Q)'\subset \Delta\cup \{((x,\xi),(y,\eta))\in C\ \mid\ x\in J^-(y)\}$;
        \item Feynman, if $\wfs(K_Q)'\subset \Delta\cup \{((x,\xi),(y,\eta))\in C\ \mid\ x\in J^{\pm}(y)\text{ if }(x,\xi),(y,\eta)\in\Sigma^{\pm}\}$;
        \item anti-Feynman, if $\wfs(K_Q)'\subset \Delta\cup \{((x,\xi),(y,\eta))\in C\ \mid\ x\in J^{\mp}(y)\text{ if }(x,\xi),(y,\eta)\in\Sigma^{\pm}\}$.
    \end{itemize}
    Parametrices of each type exist for the Klein--Gordon operator, and any two parametrices of the same type differ only by a smoothing operator. The retarded and advanced propagators are always retarded and advanced parametrices. On the other hand, Radzikowski \cite{Radzikowski} proved that the (anti-)Feynman propagator associated to a state $\omega$ is a (anti-)Feynman parametrix if and only if $\omega$ is a Hadamard state. These wavefront set properties are the precise formulation of the idea that the Feynman propagator propagates positive frequencies forward in time and negative frequencies backward. 
    
    From the perspective of microlocal analysis, it is more natural to order points on a bicharacteristic not with respect to time but with respect to the Hamilton flow (which flows forward in time in $\Sigma^+$ and backward in $\Sigma^-$), which allows one to define the distinguished parametrices more generally. For a Feynman/anti-Feynman parametrix, $\wfs(K_Q)'\backslash\Delta$ consists of pairs of points of $\Sigma$ where the first point is downstream/upstream from the second along the flow. For a retarded/advanced parametrix, $\wfs(K_Q)'\backslash\Delta$ consists of pairs of points of $\Sigma^+$ where the first point is downstream/upstream from the second, and vice versa for points of $\Sigma^-$.

\subsection{Globally defined distinguished inverses}
\label{sec:dist-inv}

    In the setting of \cite{D-H}, one only gets distinguished parametrices, which are necessarily unique only modulo smooth kernels. Therefore, in QFT on globally hyperbolic spacetimes, the usual approach is the one described above: namely, the causal propagators (which are defined uniquely) are taken as a starting point, and for a given Hadamard state, the Feynman propagators are defined via Eq.~(\ref{eq:props-bose-c-to-f}). In spacetimes possessing exact or asymptotic symmetries, there often exist states which are of particular physical interest, the basic example being the Minkowski vacuum. Nevertheless, a globally hyperbolic spacetime admits many Hadamard states \cite{F-N-W,G-W-Hadamard}, and correspondingly there are many choices of Feynman propagators, none of which are a priori preferred. 
    
    However, in the presence of additional structure it can be possible to single out parametrices of all four classes which indeed define Fredholm generalized inverses, or in some cases even true inverses, of $P$ acting between naturally-defined function spaces. These results are part of a general Fredholm theory for non-elliptic problems, based on microlocal radial point estimates which combine with propagation-of-singularities theorems to yield global information about regularity and decay of solutions, introduced by Vasy \cite{Vasy-AH-KdS} and further developed by Baskin, Vasy, and Wunsch \cite{B-V-W} and Hintz and Vasy \cite{Hintz-V}. A key difference from the framework of \cite{D-H} is analysis on compactified spacetimes, which allows one to track behavior at infinity and effectively impose boundary conditions, which is necessary to define inverses. Feynman propagators were explicitly considered by Gell-Redman, Haber, and Vasy in \cite{GR-H-V}, and Vasy in \cite{Vasy-positivity} proved positivity properties relevant for QFT for propagators in a range of geometries. Another approach to defining distinguished Feynman inverses for Klein--Gordon operators on asymptotically Minkowski spacetimes was proposed by G\'erard and Wrochna \cite{G-W-Feynman-1,G-W-Feynman-2} and is related to the work of B\"ar and Strohmaier \cite{B-S-index,B-S-chiral} on the Dirac operator on a class of compact spacetimes. The problem of Feynman inverses is also discussed by Derezi\'nski and Siemssen \cite{D-S-static,D-S-evolution,D-S-QFT}.

    When distinguished inverses of this kind exist, an alternative perspective on propagators is possible. One can choose whether to use causal or Feynman propagators as the starting point, \textit{define} those propagators to be the corresponding inverses, and use their difference to define the Hermitian form on a CCR/CAR algebra. Given a Hadamard state on the algebra, then, the other pair of propagators can be defined using the algebraic relations in Eqs.~(\ref{eq:props-bose-c-to-f})-(\ref{eq:props-bose-f-to-c}). However, it is not guaranteed that the propagators thus defined coincide with the second pair of distinguished inverses. In fact, in a generic spacetime the four distinguished inverses taken together would \textit{not} be expected to satisfy Eq.~(\ref{eq:props-bose}) for any choice of two-point functions, meaning that there is no state to which they are simultaneously associated as causal and Feynman propagators in the sense of the previous section. Concretely, the global inverses can often be interpreted as ``in-out'' propagators, that is matrix elements of expressions involving products of fields with respect to two different asymptotic vacuum states, as opposed to the corresponding expectation values in any given state -- for instance, the ``in-in'' or ``out-out'' correlators defined using either asymptotic vacuum state (though it is interesting to precisely characterize the difference). Therefore, the notion of distinguished inverses is distinct from that of propagators associated to a state; it however has the advantage of being canonical when available. 
    
    In the remainder of the paper, we often refer to the distinguished inverses simply as propagators when there is no risk of ambiguity. Compared to these objects of Lorentzian QFT, in this paper the distinguished inverses described in Section~\ref{sec:dist-inv-here} play a role similar to ``in-out'' propagators, the states we construct are analogous to in/out vacuum states, and the ``mock propagators'' we consider in Section~\ref{sec:mock-props} are similar to ``in-in'' or ``out-out'' correlators. We also note that in the construction just described, one has a choice of starting with either the causal or the Feynman propagators as fundamental; it is an interesting question how the resulting theories are related to each other.

    Vasy and Wrochna in \cite{V-W} implement this procedure, among other settings, for the wave equation on asymptotically Minkowski spacetimes, where the Fredholm setup uses weighted b-Sobolev spaces. They show that the difference of either pair of complementary propagators defines a QFT, and they construct Hadamard two-point functions (bosonic in the advanced-minus-retarded case, fermionic in the Feynman-minus-anti-Feynman case) using asymptotic data at infinity. Using asymptotic \cite{B-J,D-M-P-cosmological,D-M-P-Unruh,Moretti-BMS,Moretti-Hadamard,G-W-Hadamard-Cauchy,G-W-in/out} or Cauchy \cite{Junker,G-W-Hadamard,G-O-W} data to specify two-point functions is a common approach which connects states to quantities with clear physical interpretations.

    A very similar setup is available for the Laplacian on Riemannian scattering spaces, in the sense defined by Melrose in \cite{Melrose-AES}. While it does not have a readily apparent physical interpretation, the absence of a spacetime setting also lends itself to a sharper focus on the analytic structure, from whose perspective the Feynman propagators are just as natural as the causal ones, if not more so. For $P=\Delta-\lambda^2$, connected components of the (scattering) characteristic set correspond not to different signs of a frequency variable, but simply to connected components of the manifold's boundary. Therefore, the analysis of $P$ on a Riemannian scattering space with two boundary components is closely analogous to the analysis of the Klein--Gordon operator on an asymptotically Minkowski Lorentzian manifold, as was already noted in \cite{Vasy-positivity}. This model is the focus of our paper.
    
\subsection{Relationship with Euclidean QFT}
	Since we construct a quantum field model on a Riemannian manifold, we briefly comment on the relationship to common notions of Euclidean QFT.
	
	In particle physics, Euclidean QFTs arise as Wick-rotated versions of Lorentzian ones. Usual prescriptions for Wick rotation, based on the Osterwalder--Schrader theorem \cite{O-S,O-S-2}, require the Euclidean spacetime to have a distinguished imaginary time coordinate; they also rely on the choice of a state whose Euclidean correlators satisfy a reflection positivity property with respect to imaginary time. In our model, there is no distinguished time direction, and we construct states to have positivity properties analogous to states in Lorentzian QFT; therefore, at least for a generic Riemannian scattering space, our model seems to be incompatible with this Euclidean QFT framework.
	
	Euclidean fields also arise in statistical physics in two distinct ways (see e.g. \cite{Fradkin}). In systems described by a quantum-mechanical many-body Hamiltonian, thermodynamic quantities at finite temperature can be expressed in terms of correlators of a Euclidean QFT. To obtain the Euclidean spacetime, the spatial coordinates of the system are supplemented by an auxiliary imaginary time coordinate with compact range given by the inverse temperature \cite{Matsubara, A-G-D}. In fact, after Wick rotation one obtains thermal Lorentzian QFT, i.e. these models are Euclidean QFTs in the sense of the previous paragraph where the chosen state is a KMS state \cite{Ful-R}. Since this approach requires a distinguished time coordinate before introducing Euclidean fields, our model seems incompatible with it as well.
	
	On the other hand, in statistical field theory, a space-varying field is used to represent the time-independent local value of an order parameter. The underlying system may be classical, and the connection to Euclidean QFT is obtained by formally replacing the temperature with $\hbar$ in the partition function. While one would need to identify an imaginary time coordinate and verify reflection positivity to Wick-rotate this to a Lorentzian QFT, the original purely spatial field is already of physical interest, including in cases where there is no notion of reflection positivity. This suggests it may be possible to interpret our model as a statistical field theory, but the connection is far from clear: in such theories, the field itself is classical and the connection to QFT comes from the path integral formalism, while our approach is instead based on canonical quantization and the operator formalism.
	
	 We also note that, on a formal level, the parameter $\lambda^2$ comes into our operator with the opposite sign relative to the Laplacian than the squared mass in a (non-tachyonic) Wick-rotated Klein--Gordon operator, which would be globally elliptic as a scattering operator and therefore admit a single distinguished inverse.
	
	Based on all this, rather than a Euclidean QFT in the usual sense, our model should be considered a toy model of conventional QFT defined on a Riemannian manifold.

\section{Review of asymptotically Euclidean scattering theory}
\label{sec:AES}
    In this section, we introduce our geometric setting and discuss properties of the Laplacian. Most constructions and results in this section are due to Melrose \cite{Melrose-AES}, but we include many details in order to be self-contained and emphasize aspects most relevant to our construction.

\subsection{Riemannian scattering spaces}
    We consider a Riemannian scattering space $M$, that is a compact $n$-manifold with boundary whose interior is equipped with a scattering metric, i.e. a Riemannian metric of the form
    \begin{equation}
    \g=\frac{dx^2}{x^4}+\frac{\h}{x^2},
        \label{eq:metric}
    \end{equation}
    where $x$ is a boundary-defining function on $M$ (that is a function in $C^{\infty}(M)$ which vanishes simply on the boundary and is strictly positive elsewhere) and $\h$ is a symmetric smooth two-tensor on $M$ with nondegenerate and positive-definite restriction to the boundary. The boundary is an infinite distance away from any interior point. We consider the case when the boundary $\partial M=S=S_1\sqcup S_2$ has two connected components, though this will not be important until we begin to consider propagation of singularities.
    
    We often work in local coordinates $(x,y_1,\ldots,y_{n-1})$ on a collar neighborhood $U\simeq [0,\varepsilon)_x\times S$ of $S$ in $M$, where $(y_1,\ldots,y_{n-1})$ restrict to local coordinates on $S$. We denote by $g_{ij}$ and $g^{ij}$, where $i,j=0,1,\ldots,n-1$, the components of the metric and inverse metric respectively with respect to such a coordinate system. We also write $|g|=|\det \g|$ in these coordinates. Similarly, we write
    \[\h=h_{00}\ dx^2+\sum_{j=1}^{n-1}h_{0j}\ dx\ dy_j+\sum_{i,j=1}^{n-1}h_{ij}\ dy_i\ dy_j .\]  
     Taking $\varepsilon>0$ small enough ensures $\h$ is nondegenerate throughout $U$. For any $x\in [0,\varepsilon)$, $h_{ij}(x,y)$ for $i,j=1,\ldots,n-1$ are the components with respect to coordinates $(y_1,\ldots,y_{n-1})$ of the restriction $h_x(y)$ of $\h(x,y)$ to the tangent space of a constant-$x$ hypersurface. We also write $h^{ij}(x,y)$ for the components of the inverse of $h_x(y)$ and $|h|(x,y)=|\det h_x(y)|$. The form of the metric implies the following relations for any $i,j=1,\ldots,n-1$, with the $\mathcal{O}(\bullet)$ errors all in $C^{\infty}(M)$:
    \begin{equation}
        h_{ij},h^{ij},|h|\in C^{\infty}(U);
        \label{eq:metric-h}
    \end{equation}
    \begin{equation}
        g_{00}=\frac{1}{x^4}(1+x^2 h_{00}),\hskip 20pt
        g_{0j}=\frac{1}{x^2}h_{0j},\hskip 20pt
        g_{ij}=\frac{1}{x^2}h_{ij},\hskip 20pt
        |g|=\frac{1}{x^{2n+2}}(|h|+\mathcal{O}(x^2));
        \label{eq:metric-g}
    \end{equation}
    \begin{equation}
        g^{00}=x^4(1+\mathcal{O}(x^2)),\hskip 20pt
        g^{0j}=\mathcal{O}(x^4),\hskip 20pt
        g^{ij}=x^2(h^{ij}+\mathcal{O}(x^2)).
        \label{eq:metric-g-inv}
    \end{equation}
    
    We denote by $\dot{C}^{\infty}(M)$ the space of functions in $C^{\infty}(M)$ which vanish to infinite order at the boundary. A Fr{\'e}chet-space topology on $\dot{C}^{\infty}(M)$ is defined by the countable directed family of seminorms $\lVert x^{-l}u\rVert _{C^m(M)}$ for $m,l\in\N$. Its dual space of distributions is denoted $C^{-\infty}(M)$, which we consider equipped with the weak-* topology. We identify sufficiently regular functions on $M^{
    \circ}$ with elements of $C^{-\infty}(M)$ via the $L^2$ pairing based on the volume density of $\g$, that is,
    \[u[\varphi]=\langle \bar{u},\varphi\rangle_{L^2_{\g}}=\int_M u\varphi\ \dv\]
    for $\varphi\in\dot{C}^{\infty}(M)$ and e.g. $u$ locally integrable and such that $x^lu$ is bounded for some $l$. 
    
    We also write $C^{-\infty}(S)$ for the space of distributions on the boundary dual to $C^{\infty}(S)$ and identify sufficiently regular functions with distributions using the $L^2$ pairing based on the area density of $\h$, which we denote $L^2_{\h}$. Since $S$ is compact, we have $C^{\infty}(S)=\bigcap_{s\in\N}H^s(S)$, where $H^s(S)$ are $L^2$-based Sobolev spaces on $S$. 
    
    We often omit the subscripts $L^2_{\g}$ and $L^2_{\h}$ to lighten notation in calculations involving pairings of functions on $M$ and $S$ respectively. Note that we take sesquilinear forms to be linear in the second entry.
    
    The model example of a Riemannian scattering space is the radial compactification of $\R^n$, defined by considering inverse polar coordinates $x=\frac{1}{r}$, $\theta\in S^{n-1}$ and adding a surface of $x=0$, under which the Euclidean metric takes the form Eq.~(\ref{eq:metric}) and $\dot{C}^{\infty}(M)$ and $C^{-\infty}(M)$ are the images of the spaces of Schwartz functions and tempered distributions respectively. Scattering metrics exist on any smooth manifold with boundary, and we often refer to elements of $\dot{C}^{\infty}(M)$ and $C^{-\infty}(M)$ as Schwartz functions and tempered distributions in the general setting as well. An example with two boundary components can be built out of a connected sum of two copies of $\R^n$ radially compactified at infinity.

    We study the operator $P=\Delta_{\g}-\lambda^2$ for constant $\lambda>0$, where $\Delta_{\g}$ is the positive Laplacian associated to $\g$. From Eq.~(\ref{eq:metric-h})-(\ref{eq:metric-g-inv}), in coordinates it takes the form
    \begin{equation}
    \begin{split}
        {P
    =
    -
    (x^2\partial_x)^2
    -
    \lambda^2
    +
    (n-1)x(x^2\partial_x)
    +
    x^2\Delta_{\h}
    +\sum_{i,j=1}^{n-1}\Big(\mathcal{O}(x^2)(x^2\partial_x)^2+}
    \\
    {+\mathcal{O}(x^2)(x\partial_{y_i})(x\partial_{y_j})+\mathcal{O}(x)(x\partial_{y_i})(x^2\partial_x)+\mathcal{O}(x^2)(x^2\partial_x)+\mathcal{O}(x^2)(x\partial_{y_j})\Big),}
    \end{split}
    \label{eq:op}
    \end{equation}
    where
    \[
    \Delta_{\h}=-\frac{1}{\sqrt{|h|}}\sum_{i,j=1}^{n-1}\partial_{y_i}\Big(\sqrt{|h|}h^{ij}\partial_{y_j}\Big)
    \]
    restricts to the Laplacian of $h_x$ on any surface of constant $x$, and the $\mathcal{O}(\bullet)$ coefficients are in $C^{\infty}(M)$.

\subsection{The scattering calculus}
    The microlocal analysis of $P$ is based on the calculus of scattering pseudodifferential operators, introduced in this setting by Melrose in \cite{Melrose-AES}. We refer to \cite{Vasy-minicourse} for a systematic introduction, while here we state the main facts for convenience. They are direct generalizations of properties of the classical pseudodifferential calculus which allow one to track both smoothness and decay at infinity.
    
    We adopt the compactified phase space picture, in which characteristic/wavefront sets, etc. lie on the boundary ``at infinity'' of a compactification of $T^*M^{\circ}$, defined as follows. First, $\V{b}$ is defined as the space of smooth vector fields on $M$ tangent to the boundary, spanned in any local coordinates $(x,y_1,\ldots,y_{n-1})$ as above by $x\partial_x,\partial_{y_1},\ldots,\partial_{y_{n-1}}$. The space of scattering vector fields is defined by $\V{sc}=x\V{b}$. Then instead of the usual tangent bundle $TM$ one considers the scattering tangent bundle $\tsc$, spanned by $x^2\partial_x,x\partial_{y_1},\ldots,x\partial_{y_{n-1}}$, so $\V{sc}$ is its space of smooth sections. Next, the scattering cotangent bundle $\tscc$ is defined as the dual bundle to $\tsc$; it is spanned by $dx/x^2,dy_1/x,\ldots,dy_{n-1}/x$. Finally, the compactified phase space $\tsccc$ is defined by radial compactification in each fiber of $\tscc$. The resulting space is a manifold with corners; we distinguish the boundary face over $S$ (with defining function $x$) as ``base infinity'' and the boundary face at infinite sc-momentum over every point of $M$ (for which we can choose a defining function $\rho$) as ``fiber infinity''.

    The space $\psc{m,l}$ of (classical) scattering pseudodifferential operators of order $(m,l)$ on $M$ is constructed by quantization of symbols in $S_{\mathrm{cl}}^{m,l}(M)=x^{-l}\rho^{-m}C^{\infty}(\tsccc)$. The total space is $\psc{\infty,\infty}=\bigcup_{m,l\in\R}\psc{m,l}$. The order-$(m,l)$ scattering principal symbol $\sigma_{\mathrm{sc}}^{m,l}(A)$ is the equivalence class in $S_{\mathrm{cl}}^{m,l}/S_{\mathrm{cl}}^{m-1,l-1}$ of the full symbol of $A$. The subspace of differential operators lying in $\psc{m,0}$ for a non-negative integer $m$ is denoted $\diff[m]{sc}$; the space $\diff{sc}$ of scattering differential operators of any order is generated as an algebra over $C^{\infty}(M)$ by $\V{sc}$.

    Consider $A\in\psc{\infty,\infty}$ with symbol $a$. $A$ is called elliptic of order $(m,l)$ at a point $\zeta\in \partial\tsccc$ if there exists a neighborhood $U\subset \tsccc$ of $\zeta$ such that on $U\backslash \partial\tsccc$, we have $|a|\geqslant Cx^{-l}\rho^{-m}$. The set $\Sigma_{m,l}(A)$ of points of $\partial\tsccc$ where $A$ is not elliptic is called the characteristic set of $A$, the orders often being omitted when clear from context.
    
    On the other hand, the essential support $\wf'(A)\subset \partial\tsccc$ is defined by the fact that $\zeta\notin\wf'(A)$ if there exists a neighborhood $U\subset \tsccc$ of $\zeta$ such that on $U\backslash \partial\tsccc$, we have $|a|\leqslant C_{m,l}x^{l}\rho^{m}$ for every $m,l\in\R$. $\wf'(A)=\varnothing$ if and only if $A\in\psc{-\infty,-\infty}=\bigcap_{m,l\in\R}\psc{m,l}$.
    
    Let $A\in\psc{m_1,l_1},B\in\psc{m_2,l_2}$ with principal symbols $a,b$ respectively. Then
    \begin{itemize}
        \item 
            $\wf'(A+B)\subset \wf'(A)\cup\wf'(B)$.
        \item 
            $AB\in\psc{m_1+m_2,l_1+l_2}$, with principal symbol $ab$ and $\wf'(AB)\subset\wf'(A)\cap\wf'(B)$.
        \item 
            $[A,B]\in\psc{m_1+m_2-1,l_1+l_2-1}$, with principal symbol $-i$ times the Poisson bracket of $a$ and $b$.
        \item 
            $A^*\in\psc{m_1,l_1}$, with principal symbol $\bar{a}$ and $\wf'(A^*)=\wf'(A)$. Here $A^*$ is defined as a Fr\'echet-space adjoint, i.e. $(A^*u)[\overline{\varphi}]=u[\overline{A\varphi}]$ for any $u\in C^{-\infty}(M)$ and $\varphi\in\dot{C}^{\infty}(M)$.
    \end{itemize}

    The absolute scattering wavefront set $\wf(u)\subset \partial\tsccc$ for $u\in C^{-\infty}(M)$ is defined by the fact that $\zeta\notin\wf(u)$ if there exists $Q\in\psco$ which is elliptic at $\zeta$ such that $Qu\in\dot{C}^{\infty}(M)$. $\wf(u)=\varnothing$ if and only if $u\in\dot{C}^{\infty}(M)$.
    
    Regularity of distributions under the action of $\psc{\infty,\infty}$ is most naturally described using weighted Sobolev spaces, defined by
    \[\hsc{s,r}=\{u\in C^{-\infty}(M)\ \mid\  Au\in L^2_{\g} \text{ for all } A\in \psc{s,r}\}.\]
    The order-$(s,r)$ wavefront set $\wf^{s,r}(u)\subset\partial\tsccc$ of $u\in C^{-\infty}(M)$ is defined by the fact that $\zeta\notin\wf^{s,r}(u)$ if there exists $Q\in\psco$ which is elliptic at $\zeta$ such that $Qu\in \hsc{s,r}$. This captures the idea that $u$ is ``in $\hsc{s,r}$ in a neighborhood of $\zeta$'' if $\zeta\notin\wf^{s,r}(u)$. In the same spirit, one can define variable-order spaces $\hsc{s,r}$ where either or both of $s$ and $r$ are not constant but vary smoothly on $\partial\tsccc$ (see \cite{Vasy-minicourse} for details). One has $\wf^{s,r}(u)=\varnothing$ if and only if $u\in \hsc{s,r}$, and $\wf^{s,r}(u+v)\subset\wf^{s,r}(u)\cup\wf^{s,r}(v)$. Locally, for any point $q\in M$ we have $\pi^{-1}(q)\cap \wf^{s,r}(u)=\varnothing$ if and only if there exists $\varphi\in C_{\mathrm{c}}^{\infty}(M)$ which is nonzero at $q$ and such that $\varphi u\in \hsc{s,r}$.
    
    Importantly, $\dot{C}^{\infty}(M)=\bigcap_{s,r\in\R}\hsc{s,r}$ and $C^{-\infty}(M)=\bigcup_{s,r\in\R}\hsc{s,r}$. It follows that $\wf(u)=\bigcup_{s,r\in\R}\wf^{s,r}(u)$.

    Weighted Sobolev spaces of fixed order are Hilbert spaces with norm $\lVert u\rVert _{s,r}=\lVert \Lambda_{s,r}u\rVert _{L^2_{\g}}$, where $\Lambda_{s,r}$ is an invertible elliptic operator in $\psc{s,r}$, different choices leading to equivalent norms. The variable-order spaces are also Hilbert spaces with norm $\lVert u\rVert ^2_{s,r}=\lVert u\rVert _{s_0,r_0}^2+\lVert \Lambda_{s,r}u\rVert _{L^2_{\g}}^2$, where $s_0,r_0\in\R$ are lower bounds on $s,r$. The dual space of $\hsc{s,r}$ is $\hsc{-s,-r}$, acting by the $L^2_{\g}$ pairing. The topology of $\dot{C}^{\infty}(M)$ can be generated by the countable directed family of norms $\lVert u\rVert _{s,s}$ for $s\in\N$. Operators $A\in \psc{m,l}$ map continuously $\dot{C}^{\infty}(M)\to\dot{C}^{\infty}(M)$, $C^{-\infty}(M)\to C^{-\infty}(M)$, and $\hsc{s,r}\to \hsc{s-m,r-l}$ for any orders $s,r$. Additionally, for $u\in C^{-\infty}(M)$ and any $s,r$ we have
    \begin{itemize}
        \item \textit{(Microlocality)} $\wf^{s-m,r-l}(Au)\subset \wf'(A)\cap \wf^{s,r}(u)$;
        \item \textit{(Microlocal elliptic regularity)} $\wf^{s,r}(u)\subset \wf^{s-m,r-l}(Au)\cup\Sigma_{m,l}(A)$.
    \end{itemize}

    The $L^2_{\g}$ pairing on $\dot{C}^{\infty}(M)\times \dot{C}^{\infty}(M)$ extends to $\dot{C}^{\infty}(M)\times C^{-\infty}(M)$ and $C^{-\infty}(M)\times\dot{C}^{\infty}(M)$ according to $\langle\varphi, u\rangle=u[\bar{\varphi}]$ and $\langle u,\varphi\rangle=\overline{u[\bar{\varphi}]}$ respectively. It further extends to $\{(u,v)\in C^{-\infty}(M)\times C^{-\infty}(M)\ | \ \wf(u)\cap\wf(v)=\varnothing\}$ according to $\langle u,v\rangle=\langle Qu,v\rangle+\langle u,(I-Q^*)v\rangle$ for any $Q\in\psco$ such that $\wf'(Q)\cap \wf(u)=\varnothing $ and $\wf'(I-Q)\cap\wf(v)=\varnothing$, the result being independent of $Q$. The basic ``integration by parts'' formula $\langle u, Av\rangle=\langle A^*u,v\rangle$, valid when either $u$ or $v\in\dot{C}^{\infty}(M)$ (by definition of the adjoint) or $A\in\psc{-\infty,-\infty}$ (since such $A$ have rapidly decaying Schwartz kernels), admits the following generalization.
    \begin{lemma} 
    \label{thm:int-by-parts}
        Let $u,v\in C^{-\infty}(M)$, $A\in \psc{\infty,\infty}$. If $\wf(u)\cap\wf(v)\cap\wf'(A)=\varnothing$, then $\langle u, Av\rangle_{L^2_{\g}}=\langle A^*u,v\rangle_{L^2_{\g}}$.
    \end{lemma}
    \begin{proof} 
        Note that both pairings are well-defined under the hypotheses since the wavefront sets of the factors are disjoint.
    
        Let us first suppose that $\wf'(A)\cap\wf(v)=\varnothing$. Then we can choose $Q\in\psco$ such that $\wf'(Q)\cap\wf'(A)=\wf'(I-Q)\cap\wf(v)=\varnothing$. Using the facts that $AQ\in \psc{-\infty,-\infty}$ and $(I-Q)v\in\dot{C}^{\infty}(M)$, we get
        \begin{equation}
        \langle u, Av\rangle
        =
        \langle u,A Qv\rangle
        +
        \langle u, A(I-Q)v\rangle
        =
        \langle Q^*A^* u,v\rangle
        +
        \langle A^*u, (I-Q)v\rangle
        =
        \langle A^*u,v\rangle.
            \label{eq:int-by-parts}
        \end{equation}
        The proof in the case that $\wf'(A)\cap\wf(u)=\varnothing$ is similar, starting from the right-hand side and using the fact that $\wf'(A^*)=\wf'(A)$.

        Now consider the general case $\wf(u)\cap\wf(v)\cap\wf'(A)=\varnothing$. Choose $Q\in\psco$ such that $\wf'(Q)\cap (\wf'(A)\cap\wf(u))=\wf'(I-Q)\cap(\wf'(A)\cap\wf(v))=\varnothing$. Then using the previous result combined with the fact that $\wf'(AQ)\cap \wf(u)=\wf'(A)\cap\wf((I-Q)v)=\varnothing$, every step of Eq.~(\ref{eq:int-by-parts}) is justified again.
    \end{proof}

\subsection{The b-calculus and the Mellin transform}
\label{sec:b-Mellin}
    While we do not require the full pseudodifferential machinery of the b-calculus, we use some facts about b-differential operators and their associated weighted Sobolev spaces.

    $\diff[m]{b}$ is defined as the space of linear combinations (with coefficients in $C^{\infty}(M)$) of products of order up to $m$ of vector fields in $\V{b}$, and we write $\diff{b}$ for the space of b-differential operators of any order. The weighted b-Sobolev spaces for $r\in\R$ and non-negative integers $s$ are defined by
    \[
    \hb{s,r}=\{u\in C^{-\infty}(M)\ \mid\ Au\in L^2_{\mathrm{b}}(M)\text{ for all }A\in x^{-r}\diff[s]{b}\},
    \]
    where $L^2_{\mathrm{b}}(M)$ is the space of functions on $M$ square-integrable with respect to the density corresponding to any nonvanishing smooth section of the top exterior power of ${}^{\mathrm{b}}T^*M$, the dual bundle to the bundle ${}^{\mathrm{b}}TM$ whose space of smooth sections is $\V{b}$. In coordinates, such densities are smooth multiples of $(dx/x) dy_1\ldots dy_{n-1}$ near the boundary. An operator $A\in x^{-l}\diff[m]{b}$ maps $\hb{s,r}\to \hb{s-m,r-l}$ for any $s,r,m,l$.

    Near the boundary, 
    \[\dv=\sqrt{|g|}\ dx\ dy_1\ldots dy_{n-1}=x^{-n}\Big(\sqrt{|h|}+\mathcal{O}(x^2)\Big)(dx/x)\ dy_1\ldots dy_{n-1},\]
    so $L^2_{\g}=x^{\frac{n}{2}}L^2_{\mathrm{b}}$. Then, since $x^{m}\diff[m]{b}\subset \diff[m]{sc}\subset \diff[m]{b}$, the sc- and b- weighted Sobolev spaces are related by $\hb{s,r+\frac{n}{2}}\subset\hsc{s,r}\subset \hb{s,r+\frac{n}{2}-s}$ for $s\geqslant 0$.

    To study the behavior of a function $u_0$ on $M$ at infinity, we often consider $u=\varphi u_0$ for some cutoff $\varphi\in C^{\infty}(M)$ which is supported in a collar neighborhood of $S$ of the form $\{x<\varepsilon\}$ and is equal to 1 in a neighborhood of $S$. $u$ can be interpreted as a function on $[0,\varepsilon)_x\times S$, and by setting it to zero for $x>\varepsilon$, as a function on $[0,+\infty)_x\times S$. We say $u\in H_{\mathrm{b}}^{s,r}([0,+\infty)\times S)$ if $Qu\in x^rL^2_{\mathrm{b}}([0,+\infty)\times S)$ for all differential operators $Q$ which are products of order up to $s$ of $x\partial_x$ and elements of $\mathrm{Diff}(S)$. Here $u\in L^2_{\mathrm{b}}([0,+\infty)\times S)$ if $u$ is square-integrable with respect to $dA_{\h}\ dx/x$, where $dA_{\h}$ is the area density defined on $S$ by $h_0$.
    
    If $u_0$ (and therefore also $u$) is in a weighted b-Sobolev space, it can be useful to consider its Mellin transform, defined (initially for instance on $C^\infty_c((0,\infty)\times S$) for $\sigma\in\C$ and $y\in S$ by
    \[\hat{u}(\sigma,y)=\int_0^{\infty} x^{-i\sigma}u(x,y)\frac{dx}{x}.\]
    Some useful properties are (see Section 3 of \cite{Vasy-AH-KdS} for details):
    \begin{itemize}
        \item 
            For any $r\in\R$ and non-negative integer $s$, the map $u\mapsto (\sigma\mapsto\hat{u}(\sigma,\bullet))$ is an isomorphism
            \[
            H_{\mathrm{b}}^{s,r}([0,+\infty)\times S)\to \{v\in L^2(\R-ir;H^s(S))\ \mid\ (1+|\sigma|^2)^{\frac{s}{2}}v\in L^2(\R-ir;L^2(S))\}.
            \]
            Moreover, the Mellin inversion formula
            \[u(x,y)=\frac{1}{2\pi}\int_{\R-ir}x^{i\sigma}\hat{u}(\sigma,y)\ d\sigma\]
            holds (the inverse transform defined via continuous extension from a dense subspace).
        \item 
            If $u\in H_b^{s,r}([0,+\infty)\times S)$ and $\supp(u)\subset [0,\varepsilon)_x\times S$ for some $\varepsilon>0$, then the map $\sigma\mapsto \hat{u}(\sigma,\bullet)$ is holomorphic from the half-plane $\mathrm{Im}(\sigma)>-r$ to $H^s(S)$, continuous down to $\mathrm{Im}(\sigma)=-r$, with precise estimates based on the Paley--Wiener theorem.
        \item 
            As distributions, $\widehat{x\partial_xu}(\sigma,y)=i\sigma\hat{u}(\sigma,y)$.
    \end{itemize}
    
\subsection{Hamilton flow and propagation of singularities for $P=\Delta_{\g}-\lambda^2$}
    Now we can consider properties of $P$ as a scattering operator. Properties of solutions to the equation $Pu=f$ are constrained by elliptic regularity except at the characteristic set, where we instead have propagation of singularities of solutions along the Hamilton vector field associated to the principal symbol. In this section, we describe these structures for $P$.

    Let $(x,y_1,\ldots,y_{n-1})$ be local coordinates as above. Let $(\xi,\eta_1,\ldots,\eta_{n-1})$ be the dual variables in $\tscc$, related to the canonical dual variables $(\xi^0,\eta^0)$ in $T^*M$ by $\xi=x^2\xi^0$ and $\eta_j=x\eta_j^0$.

    Let $\lambda>0$. The operator $P=\Delta_{\g}-\lambda^2\in \diff[2]{sc}\subset\psc{2,0}$ has principal scattering symbol
    \[\mathbf{p}(q,p)=\lVert p\rVert _{\g}^2-\lambda^2=\xi^2+\lVert \eta\rVert _{\h}^2-\lambda^2+\mathcal{O}(x),\]
    where $\lVert \bullet\rVert _{\g}$ is the norm on fibers of $T^*M^{\circ}$ defined by the dual metric and $\lVert \eta\rVert _{\h}^2=\sum_{i,j=1}^{n-1}h^{ij}\eta_i\eta_j$. The characteristic set of $P$ is located at base infinity away from the corner and can be described in coordinates as
    \[\Sigma=\{(x,y,\xi,\eta)\in \tsccc\ \mid\ x=0,\ \xi^2+\lVert \eta\rVert _{\h}^2=\lambda^2\}.\]
    We denote the two connected components $\Sigma_1=\{(q,p)\in\Sigma\ \mid\ q\in S_1\}$, $\Sigma_2=\{(q,p)\in\Sigma\ \mid\ q\in S_2\}$. In particular, as a scattering operator $P$ is not globally elliptic (which it is in the usual (pseudo)differential operator sense).

    The Hamilton vector field corresponding to $\mathbf{p}$ is
        \begin{equation}
        H
        =
        2x \Bigg(
        \xi x\frac{\partial}{\partial x}
        -\lVert \eta\rVert _{\h}^2\frac{\partial}{\partial\xi}
        +\sum_{j,k=1}^{n-1}h^{jk}\eta_j\frac{\partial}{\partial y_k}
        +\sum_{k=1}^{n-1}\Bigg(\xi\eta_k-\frac{1}{2}\sum_{i,j=1}^{n-1}\frac{\partial h^{ij}}{\partial y_k}\eta_i\eta_j\Bigg)\frac{\partial}{\partial\eta_k}
        +xV
        \Bigg),
        \label{eq:Hamilton}
        \end{equation}
    where $V$ is a vector field tangent to $\partial\tsccc$. Dividing by $x$ and restricting to $x=0$, we get the rescaled Hamilton flow on the phase space boundary:
        \begin{equation}
        \Tilde{H}
        =
        -2\lVert \eta\rVert _{\h}^2\frac{\partial}{\partial\xi}
        +2\sum_{j,k=1}^{n-1}h^{jk}\eta_j\frac{\partial}{\partial y_k}
        +\sum_{k=1}^{n-1}\Bigg(2\xi\eta_k-\sum_{i,j=1}^{n-1}\frac{\partial h^{ij}}{\partial y_k}\eta_i\eta_j\Bigg)\frac{\partial}{\partial\eta_k}.
        \label{eq:Hamilton-rescaled}
        \end{equation}    
    $\Tilde{H}=0$ if and only if $\eta=0$. On $\Sigma$, this also implies $\xi=\pm\lambda$. 
    The set of radial (critical) points $R\subset \Sigma$ therefore has four connected components, written in fiber coordinates $(\xi,\eta_1,\ldots,\eta_{n-1})$ as
        \[
        R_1^{\pm}
        =
        S_1\times\{(\pm\lambda,0,\ldots,0)\}\text{ and }R_2^{\pm}=S_2\times\{(\pm\lambda,0,\ldots,0)\}.
        \]
    $R^+=R_1^+\sqcup R_2^+$ is a source for the Hamilton flow, and $R^-=R_1^-\sqcup R_2^-$ is a sink. Every bicharacteristic (integral curve of $\tilde{H}$ in $\Sigma$) is either a point in $R$ or tends to a point of $R^+$ in the backward direction and a point of $R^-$ in the forward direction.

    The following is a summary of relevant statements about propagation of singularities along the Hamilton flow. Because the characteristic set stays away from fiber infinity, the differential orders will largely be irrelevant in what follows, since for $\zeta\in\Sigma$ and fixed $r$ the statements that $\zeta\in\wf^{s,r}(u)$ are equivalent for all $s\in\R$.
    
    \begin{theorem}{\textbf{(Propagation of singularities for $P$).}}
    \label{thm:PoS}
        Fix $s\in\R$.
        \begin{enumerate}
            \item (Real principal type propagation).
                For any $r\in\R$, $\wf^{s,r}(u)\backslash\wf^{s-1,r+1}(Pu)$ is a union of maximally extended bicharacteristics in $\Sigma\backslash\wf^{s-1,r+1}(Pu)$.
            \item (Variable-order version).
                For any $r\in C^{\infty}(\tsccc)$ which is non-increasing (non-decreasing) along $H$, $\wf^{s,r}(u)\backslash\wf^{s-1,r+1}(Pu)$ is a union of maximally backward-extended (forward-extended) bicharacteristics in $\Sigma\backslash\wf^{s-1,r+1}(Pu)$.
            \item (Below-threshold radial point estimate).
                Fix a point $\zeta\in R^{\pm}$. For any constant $r< -\frac{1}{2}$, if $\zeta\notin \wf^{s-1,r+1}(Pu)$ and there exists a neighborhood $U$ of $\zeta$ in $\partial\tsccc$ such that $\wf^{s,r}(u)\cap (U\backslash R^{\pm})=\varnothing$, then we have $\zeta\notin \wf^{s,r}(u)$.
            \item (Above-threshold radial point estimate).
                Fix a point $\zeta\in R^{\pm}$. For any constant $r> -\frac{1}{2}$, if $\zeta\notin \wf^{s-1,r+1}(Pu)$ and there exists $r_0\in(-\frac{1}{2},r)$ such that $\zeta\notin \wf^{s,r_0}(u)$, then we have $\zeta\notin \wf^{s,r}(u)$.
        \end{enumerate}
    \end{theorem}
    The real principal type result is a scattering-theoretic version, due to Melrose \cite{Melrose-AES}, of H{\"o}rmander's theorem \cite{Hormander,D-H}. For the variable-order version, see \cite{Vasy-minicourse}. The radial point results are localized versions of the original global radial set estimates of Melrose \cite{Melrose-AES}, which did not resolve different points within the radial set. Localized estimates are due to Vasy \cite{Vasy-GenEigen} in this context, generalized by Haber and Vasy \cite{H-V}. The global version is actually sufficient for our purposes, since we will only need to localize to a component of the boundary, which is simple given the global result.

    Haber and Vasy \cite{H-V} also proved an additional module regularity result, of which the following is a special case.
    \begin{theorem}{\textbf{(Lagrangian regularity).}}
    \label{thm:mod-reg}
        Fix a point $\zeta\in R^{\pm}$ and $k\in \N$. For any neighborhood $U\subset\partial\tsccc$ of $\zeta$, define 
        \[\mathcal{M}_{U,\pm}=\Big\{A\in\psc{\infty,1}\ \mid\ \wf'(A)\subset U\text{ and }(x\sigma^{0,1}_{\mathrm{sc}}(A))|_{R^{\pm}}=0\Big\}.\]
        For any constant $r< -\frac{1}{2}$, if there exists a neighborhood $U'\subset \partial\tsccc$ of $\zeta$ such that $\mathcal{M}_{U',\pm}^k(Pu)\subset \hsc{\infty,r+1}$ and $\wf^{\infty,r+k}(u)\cap (U'\backslash R^{\pm})=\varnothing$, then there exists a neighborhood $U\subset U'$ of $\zeta$ such that $\mathcal{M}_{U,\pm}^k(u)\subset \hsc{\infty,r}$.
    \end{theorem}

\subsection{Distinguished inverses}
\label{sec:dist-inv-here}
    For the sake of uniformity of notation, let us write $R^{\pm\pm}=R^{\pm}$ and $R^{\pm\mp}=R^{\pm}_1\cup R^{\mp}_2$. We also use the shorthand $R^{-\alpha\beta}=R^{(-\alpha)(-\beta)}$, and similarly for other objects double-indexed by $\alpha,\beta\in \{+,-\}$. For any choice of $\alpha,\beta\in\{+,-\}$, any bicharacteristic is either a point in $R^{\alpha\beta}$, or a point in $R^{-\alpha\beta}$, or tends to $R^{\alpha\beta}$ in one direction and to $R^{-\alpha\beta}$ in the other.

    Let $r_{\pm\pm},r_{\pm\mp}\in C^{\infty}(\tsccc)$ be real-valued functions monotone along $H$ in each component of $\Sigma$ and such that
        \[
        \begin{cases}
        r_{++}|_{R^{++}}<-\frac{1}{2},\\
        r_{++}|_{R^{--}}>-\frac{1}{2};
        \end{cases}
        \hskip 30pt
        \begin{cases}
        r_{--}|_{R^{++}}>-\frac{1}{2},\\
        r_{--}|_{R^{--}}<-\frac{1}{2};
        \end{cases}
        \]
        \[
        \begin{cases}
        r_{+-}|_{R^{+-}}<-\frac{1}{2},\\
        r_{+-}|_{R^{-+}}>-\frac{1}{2};
        \end{cases}
        \hskip 30pt
        \begin{cases}
        r_{-+}|_{R^{+-}}>-\frac{1}{2},\\
        r_{-+}|_{R^{-+}}<-\frac{1}{2}.
        \end{cases}
        \]
    We refer to these assumptions as the \textbf{Feynman} ($\pm\pm$) and \textbf{causal} ($\pm\mp$) threshold conditions on $r_{\alpha\beta}$ by analogy with the case of the Klein--Gordon operator, for which the (anti-)Feynman and causal (retarded/advanced) propagators can be defined as inverses acting between spaces of functions with regularity satisfying analogous threshold conditions. For any $r\in \{r_{++},r_{--},r_{+-},r_{-+}\}$ and any $s\in \R$, define
        \[\Y_{s,r}=\hsc{s-2,r+1},
        \hskip 20pt
        \X_{s,r}=\{u\in \hsc{s,r}\ \mid\ Pu\in \Y_{s,r}\}.\]
    Define $P_{s,r}=P|_{\X_{s,r}}$. It follows from the propagation of singularities theorem that these are Fredholm maps. In fact, even more can be said.
    \begin{theorem}{\textbf{(Invertibility properties)}}
    \label{thm:invertibility}
        \begin{enumerate}
            \item $P_{s,r_{\pm\pm}}:\X_{s,r_{\pm\pm}}\to\Y_{s,r_{\pm\pm}}$ are invertible for all $s\in \R$ and $r_{\pm\pm}$ satisfying the Feynman threshold conditions, with inverses $P_{\pm\pm}^{-1}=(\Delta_{\g}-\lambda^2\pm i0)^{-1}$.
            \item In dimension $n\geqslant 3$, for any given $(M,\g)$ there exists $\lambda_0>0$ such that for any $\lambda\in (0,\lambda_0)$, $P_{s,r_{\pm\mp}}:\X_{s,r_{\pm\mp}}\to\Y_{s,r_{\pm\mp}}$ are invertible for all $s\in \R$ and $r_{\pm\mp}$ satisfying the causal threshold conditions.
        \end{enumerate}
    \end{theorem}
    \begin{proof}
        As already stated, for both (1) and (2) the Fredholm statement follows from propagation of singularities; see concretely \cite{Vasy-minicourse}. It remains to show the triviality of the nullspace of the operators and their adjoints, in the latter case on the dual function spaces.

        In case (1), this is the usual limiting absorption principle, but in a larger context it follows from the arguments of \cite[Section~4]{Vasy:Self-adjoint} which give membership of elements of the kernel (both for the operator and for its adjoint) in Schwartz functions, and then unique continuation at infinity (combined with the standard unique continuation) completes the proof.

        In case (2), this follows from either the arguments of \cite{Vasy:Zero-energy-lag} or those of \cite{Vasy:Zero-energy}. In both cases, while there the analogue of (1) is studied explicitly, (2) has completely similar behavior and the same arguments apply, mutatis mutandis. While \cite{Vasy:Zero-energy} is closer to our setting, in that it uses variable order spaces rather than second microlocalization, \cite{Vasy:Zero-energy-lag} is simpler to describe and as given the already discussed Fredholm property only the triviality of the nullspace and the operator matters, on relevant function spaces, regularity theory allows one to conclude that the triviality results of \cite{Vasy:Zero-energy-lag} imply those needed here.

        For the approach of \cite{Vasy:Zero-energy-lag} one conjugates the operator by the exponential encoding the asymptotics, which is $e^{\pm i\lambda/x}$. In terms of phase space, this moves the relevant radial set to the zero section, at which second microlocalization can be easily described in terms of the b-pseudodifferential operator algebra. Then one works with the resolved spaces in which $x=\lambda=0$ is blown up. On the front face one obtains rescaled models given by the Laplacian of an {\em exact} conic metric over $\partial M=S$ minus the spectral parameter $1$; note that this decouples the connected components of $S$ (in that they give rise to disjoint faces in the resolved space, and disjoint cones for the model), so on all of the connected components one solves either a Feynman or an anti-Feynman problem, all of which are invertible. The function spaces at the tip of the cone are independent of these choices, and correspondingly so is the energy zero problem, which gives the standard Laplacian, whose invertibility properties follow as in \cite{Vasy:Zero-energy-lag}. By the arguments of \cite{Vasy:Zero-energy-lag} this gives the desired invertibility for $\lambda\in (0,\lambda_0)$ for sufficiently small $\lambda_0$, albeit on slightly different function spaces. However, due to \cite[Proposition~4.10]{Vasy:Limiting-absorption-lag} any element of the kernel of the operator on our function space is necessarily in that of \cite{Vasy:Zero-energy-lag}, proving that the operator has trivial kernel on our function space. On the other hand, for the adjoint operator, the function space of \cite{Vasy:Zero-energy-lag} includes our function space, hence trivial kernel of the adjoint on that space implies the same for our function space.
    \end{proof}
    
    For QFT purposes, one could proceed by defining suitable quotient spaces on which $P_{s,r_{\pm\mp}}$ induces a bijection (as was done in \cite{V-W}), but for simplicity we will assume that we are in a setting where the map is indeed invertible for all choices of $s,r_{\pm\mp}$ as above. 
    
    The inverses $P_{\pm\pm}^{-1}$ are analogous to the Feynman and anti-Feynman propagators for wave equations in the Lorentzian setting, which are defined by the fact that they propagate singularities in the same direction (with respect to the Hamilton vector field) in every component of the characteristic set. Similarly, the inverses $P_{\pm\mp}^{-1}$ are analogous to the retarded and advanced propagators. A precise version of these statements is shown in Proposition~\ref{thm:PoS-G}.
    
    It is convenient to first define all our constructions in the most regular setting available. If we take the intersections of the spaces defined above for all $s\in\R$ and all $r_{\alpha\beta}$ satisfying a single set of threshold conditions (for a single choice of $\alpha,\beta\in \{+,-\}$), we get
        \[\Y_{\alpha\beta}=\{f\in \cap_{r<\frac{1}{2}}\hsc{\infty,r}\ \mid\ \wf(f)\subset R^{\alpha\beta}\},\]
        \[\X_{\alpha\beta}=\{u\in \cap_{r<-\frac{1}{2}}\hsc{\infty,r}\ \mid\ \wf(u)\subset R^{\alpha\beta}\text{ and }Pu\in\Y_{\alpha\beta}\},\]
    and linear maps $P_{\alpha\beta}:=P|_{\X_{\alpha\beta}}:\X_{\alpha\beta}\to\Y_{\alpha\beta}$. Note that for any $u\in C^{-\infty}(M)$, if $\wf(u)\cap R^{-\alpha\beta}=\varnothing$, then since $\wf(u)$ is closed, it is disjoint from a neighborhood of $R^{-\alpha\beta}$; then if we know $Pu\in\Y_{\alpha\beta}$, real principal type propagation from this neighborhood implies $\wf(u)\subset R^{\alpha\beta}$, and the below-threshold estimate further implies $u\in\bigcap_{r<-\frac{1}{2}}\hsc{\infty,r}$. Therefore, we can in fact equivalently write
        \[\X_{\alpha\beta}=\{u\in C^{-\infty}(M)\ \mid\ \wf(u)\cap R^{-\alpha\beta}=\varnothing\text{ and }Pu\in\Y_{\alpha\beta}\}.\]
    
    \begin{prop}
        For any choice $\alpha,\beta\in \{+,-\}$, $P_{\alpha\beta}:\X_{\alpha\beta}\to\Y_{\alpha\beta}$ is invertible.
    \end{prop}
    \begin{proof}
        Let $u\in\X_{\alpha\beta}$. Choose any particular $s\in\R$ and $r_{\alpha\beta}$ satisfying the relevant threshold conditions. Since $u\in\X_{s,r_{\alpha\beta}}$ and $P_{s,r_{\alpha\beta}}$ is invertible, $P_{\alpha\beta}u=P_{s,r_{\alpha\beta}}u\neq 0$ if $u\neq 0$. Therefore, $P_{\alpha\beta}$ is injective.

        To see surjectivity, let $f\in\Y_{\alpha\beta}$. Choose any particular $s\in\R$ and $r_{\alpha\beta}$ satisfying the relevant threshold conditions; we have $f\in\Y_{s,r_{\alpha\beta}}$. Let $u=P_{s,r_{\alpha\beta}}^{-1}f\in\X_{s,r_{\alpha\beta}}$. By the threshold condition for $r_{\alpha\beta}$, $\wf^{s,r}(u)\cap R^{-\alpha\beta}=\varnothing$ for some $r>-\frac{1}{2}$. Since $f\in\Y_{\alpha\beta}$, we also know that $\wf(f)\cap R^{-\alpha\beta}=\varnothing$. Then by the above-threshold estimate we get $\wf(u)\cap R^{-\alpha\beta}=\varnothing$. Since $Pu=f\in\Y_{\alpha\beta}$, this proves that $u\in \X_{\alpha\beta}$ and therefore the surjectivity of $P_{\alpha\beta}$.
    \end{proof}

\subsection{Boundary data map and Poisson operators}
\label{sec:bd-Poisson}
    We now consider spaces of solutions to the equation $Pu=0$ which are ``Schwartz away from the radial set''. We define
    \[\sol=\{u\in C^{-\infty}(M)\ \mid\ Pu=0\text{ and }\wf(u)\subset R\}\]
    and the space of ``approximate solutions''
    \[\sols=\{u\in C^{-\infty}(M)\ \mid\ Pu\in\dot{C}^{\infty}(M)\text{ and }\wf(u)\subset R\}.\]
    For any $f\in\dot{C}^{\infty}(M)$, we also write $\solf=\{u\in C^{-\infty}(M)\ \mid\ Pu=f\text{ and }\wf(u)\subset R\}$. The below-threshold radial point estimate implies that $\sols\subset \bigcap_{r<-\frac{1}{2}}\hsc{\infty,r}$.
    
    These spaces of solutions can be parametrized by incoming and outgoing data on the boundary of $M$. In this section, we collect important results about this parametrization, first obtained in \cite{Melrose-AES}.
    
    Let $x$ be a boundary-defining function on $M$ with respect to which the metric has the form Eq.~(\ref{eq:metric}).

    \begin{theorem}{\textbf{(Asymptotic boundary data for Schwartz right-hand side)}} 
    \label{thm:bd}
        \begin{enumerate}
                \item
                    For any $a^{\pm}\in C^{\infty}(S)$ there exists a unique $[a^{\pm}_M]\in C^{\infty}(M)/\dot{C}^{\infty}(M)$ such that for some (and therefore every) representative $a_M^{\pm}$ we have $a^{\pm}_M|_S=a^{\pm}$ and 
                    \[u_{\pm}:=x^{\frac{n-1}{2}}a^{\pm}_Me^{\mp\frac{i\lambda}{x}}\in \sols.\]
        
                \item
                    For any $u\in \sols$, there exist unique  $[a^+_M],[a^-_M]\in C^{\infty}(M)/\dot{C}^{\infty}(M)$ such that for some representatives $a^+_M$, $a^-_M$ we have
                        \[
                        u
                        =
                        x^{\frac{n-1}{2}}
                        \Big(
                        a_M^+e^{-\frac{i\lambda}{x}}
                        +
                        a_M^-e^{\frac{i\lambda}{x}}
                        \Big).
                        \]
                \item 
                    Fix $f\in\dot{C}^{\infty}(M)$. For any $a^{\pm}\in C^{\infty}(S)$ there exist unique $u\in \solf$ and $a^{\mp}\in C^{\infty}(S)$ for which there exist $a_M^+,a_M^-\in C^{\infty}(M)$ such that $a^+_M|_S=a^+$, $a^-_M|_S=a^-$ and
                        \[
                        u
                        =
                        x^{\frac{n-1}{2}}
                        \Big(
                        a^+_Me^{-\frac{i\lambda}{x}}
                        +
                        a^-_Me^{\frac{i\lambda}{x}}
                        \Big).
                        \]
        \end{enumerate}
    \end{theorem}

    Before proceeding to prove the theorem, we note that the wavefront set of a solution is encoded in the support of its boundary data.
    
    \begin{lemma}
    \label{thm:wf-bd}
        If $Pu\in \dot{C}^{\infty}(M)$ and $u=x^{\frac{n-1}{2}}\Big(a_M^+e^{-\frac{i\lambda}{x}}+a_M^-e^{\frac{i\lambda}{x}}\Big)$ for some $a_M^{+},a_M^{-}\in C^{\infty}(M)$, then
        \begin{equation}
        \wf(u)
        =
        \wf^{\infty,-\frac{1}{2}}(u)
        =
         \Big(\pi^{-1}(\supp(a^+))\cap R^+\Big)
         \cup
         \Big(\pi^{-1}(\supp(a^-))\cap R^-\Big),
        \label{eq:wf-bd}
        \end{equation}
        where $a^{\pm}=a_M^{\pm}|_S$ and $\pi:\tscc\to M$ is the bundle projection.
    \end{lemma}
    
    \begin{proof}
    Note that $x^{\frac{n}{2}+r}e^{\mp\frac{i\lambda}{x}}a\in \bigcap_{r'<r}\hsc{\infty,r'}$ automatically if $a\in C^{\infty}(M)$.
    
    Let $u_{\pm}=x^{\frac{n-1}{2}}a_M^{\pm}e^{\mp\frac{i\lambda}{x}}$.  Since $u_{\pm}\in C^{\infty}(M\backslash S)$, $\wf(u_{\pm})$ lies entirely over the boundary. 

    Consider any $q\in S$. Let $U\subset M$ be a neighborhood of $q$ on which we have coordinates $(x,y_1,\ldots,y_{n-1})$ as above. Choose $\varphi\in C_{\mathrm{c}}^{\infty}(M)$ such that $\supp(\varphi)\subset U$ and $\varphi(q)\neq 0$. Let $Q_{\pm}=\varphi\cdot\frac{1}{x}(x^2\partial_x\mp i\lambda)$. By direct computation,
    \[
    Q_{\pm}u_{\pm}
    =
    x^{\frac{n-1}{2}}e^{\mp\frac{i\lambda}{x}}\bigg(\frac{n-1}{2}\varphi a_M^{\pm}+x\varphi\cdot \partial_x a_M^{\pm}\bigg),
    \]
    and by iteration, for any $k\in\N$ there is $a_k^{\pm}\in C^{\infty}(M)$ such that $Q_{\pm}^k u_{\pm}=x^{\frac{n-1}{2}}e^{\mp\frac{i\lambda}{x}}a_k^{\pm}\in \bigcap_{r<-\frac{1}{2}}\hsc{\infty,r}$. Meanwhile, the characteristic set of $Q_{\pm}^k\in \psc{k,k}$ over $q$ is $\{\xi=\pm\lambda\}$, so by elliptic regularity $\wf(u_{\pm})\cap\pi^{-1}(q)\subset \{\xi=\pm\lambda\}$. Similarly, $\varphi\partial_{y_j}^ku_{\pm}\in \bigcap_{r<-\frac{1}{2}}\hsc{\infty,r}$ for all $k\in\N$. Since $\varphi\partial_{y_j}^k\in \psc{k,k}$ and its characteristic set over $q$ is $\{\eta_j=0\}$, we get $\wf(u_{\pm})\cap\pi^{-1}(q)\subset \{\eta_j=0\}$. 
    
    Combining the results, we get $\wf(u_{\pm})\subset R^{\pm}$. Since $u=u_++u_-$ and $Pu\in\dot{C}^{\infty}(M)$, the fact that $\wf(u_+)\cap\wf(u_-)=\varnothing$ implies $Pu_{\pm}\in\dot{C}^{\infty}(M)$ individually, which allows us to use propagation of singularities.

    Consider $q\in S\backslash\supp(a^{\pm})$, so we can choose $\varphi\in C_{\mathrm{c}}^{\infty}(M)$ which is nonzero at $q$ but for which $\varphi a_M^{\pm}|_S=0$. Then $\varphi a_M^{\pm}= xb$ for some $b\in C^{\infty}(M)$, and $\varphi u_{\pm}= x^{\frac{n+1}{2}}be^{\mp\frac{i\lambda}{x}} \in \bigcap_{r<\frac{1}{2}}\hsc{\infty,r}$, so, for example, $\pi^{-1}(q)\cap R^{\pm}\notin \wf^{\infty,0}(u_{\pm})$. Then the above-threshold estimate for $u_{\pm}$ implies $\pi^{-1}(q)\cap R^{\pm}\notin \wf(u_{\pm})$, so in fact by the previous result $\wf(u_{\pm})\cap \pi^{-1}(q)=\varnothing$.
    
    On the other hand, if $q\in\supp(a^{\pm})$, then for any neighborhood $U\subset M$ of $q$ there exists $\varphi\in C_{\mathrm{c}}^{\infty}(M)$ and a point $q'\in U\cap S$ such that $\supp(\varphi)\subset U$ but $\varphi(q')a_M^{\pm}(q')=C\neq 0$. Then $|\varphi u_{\pm}|\geqslant |C/2|x^{\frac{n-1}{2}}$ on a neighborhood of $q'$, so $\varphi u_{\pm}\notin \hsc{\infty,-\frac{1}{2}}$. Since $U$ can be taken arbitrarily small, this implies that $\wf^{\infty,-\frac{1}{2}}(u_{\pm})\cap\pi^{-1}(q)\neq\varnothing$.

    Thus we have shown $\wf(u_{\pm})=\wf^{\infty,-\frac{1}{2}}(u_{\pm})=\pi^{-1}(\supp(a^{\pm}))\cap R^{\pm}$, and the statement of the lemma follows.
    \end{proof}

    Now we consider Theorem \ref{thm:bd}. All components of this theorem are contained in \cite{Melrose-AES}, but we give a streamlined proof for completeness.
	\begin{proof}[Proof of Theorem~\ref{thm:bd}]
		Let $P_{\pm}^{-1}=P_{\pm\pm}^{-1}$ denote the Feynman inverses.

  \begin{enumerate}
      \item 
            Consider a general function on $M$ of the form $u_{\pm}=x^{\frac{n-1}{2}}e^{\mp i\frac{\lambda}{x}}a_M^{\pm}$ with $a_M^{\pm}\in C^{\infty}(M)$. Using Taylor's theorem in the $x$ variable for $a_M^{\pm}$, for any non-negative integer $k$ we can write
            \begin{equation}
            u_{\pm}=\sum_{j=0}^k x^{\frac{n-1}{2}+j}e^{\mp i\frac{\lambda}{x}}a_j+x^{\frac{n-1}{2}+k+1}e^{\mp i\frac{\lambda}{x}}b_k,
                \label{eq:Taylor}
            \end{equation}
            where $a_j,b_k\in C^{\infty}(M)$ and $a_j|_S=\partial_x^ja_M^{\pm}|_S$. Applying the coordinate expression Eq.~(\ref{eq:op}) for $P$ to Eq.~(\ref{eq:Taylor}), we can compute
            \begin{equation}
            \begin{split}
            Pu_{\pm}
            =
            x^{\frac{n-1}{2}}e^{\mp i\frac{\lambda}{x}}
            \Bigg(
            \sum_{j=2}^{k+1} \Big(P_{j-2}a_{j-2}\mp 2i\lambda(j-1)a_{j-1}\Big)x^j
            +\\
            +
            \Big(xP_{k+1}b_k\mp 2i\lambda(k+1)b_k\Big)x^{k+2}
            \Bigg),
            \end{split}
                \label{eq:Taylor-Pu}
            \end{equation}
            where $P_j\in\diff{}$, so in particular $P_ja_j,P_{k+1}b_k\in C^{\infty}(M)$.

            Now, with this in mind, let us choose functions $a_j\in C^{\infty}(M)$ by taking any $a_0$ such that $a_0|_S=a^{\pm}$ and setting $a_j=\pm\frac{P_{j-1}a_{j-1}}{2i\lambda j}$ for all $j>0$. By Borel's lemma, there exists $a_M^{\pm}\in C^{\infty}(M)$ such that $\partial_x^ja_M^{\pm}|_S=a_j|_S$. Then $u_{\pm}=x^{\frac{n-1}{2}}e^{\mp i\frac{\lambda}{x}}a_M^{\pm}$ satisfies Eq.~(\ref{eq:Taylor}), and the Taylor series of $x^{-\frac{n-1}{2}}e^{\pm i\frac{\lambda}{x}}Pu_{\pm}$ is zero by Eq.~(\ref{eq:Taylor-Pu}). Therefore, $Pu_{\pm}\in\dot{C}^{\infty}(M)$, which by Lemma~\ref{thm:wf-bd} implies $\wf(u_{\pm})\subset R$ and so $u_{\pm}\in\sols$.

            To see uniqueness of $a_M^{\pm}$ up to Schwartz terms, note that if $u,v\in\sols$ are of the form $u=x^{\frac{n-1}{2}}e^{\mp\frac{i\lambda}{x}}a$, $v=x^{\frac{n-1}{2}}e^{\mp\frac{i\lambda}{x}}b$ with $a,b\in C^{\infty}(M)$ and $a|_S=b|_S$, then $u-v\in\sols$ and $u-v=x^{\frac{n-1}{2}}e^{\mp\frac{i\lambda}{x}}(a-b)$, where $(a-b)|_S=0$. Then by Lemma~\ref{thm:wf-bd} we get $u-v\in\dot{C}^{\infty}(M)$, which implies $a-b\in\dot{C}^{\infty}(M)$.

        \item 
            Choose $Q_+,Q_-\in\psco$ such that $Q_++Q_-=I$ and $\wf'(Q_{\pm})\cap R^{\mp}=\varnothing$. Let $u_{\pm}=Q_{\pm}u$. Then we have $u=u_++u_-$, $\wf(u_{\pm})\subset R^{\pm}$, and since $\wf'([P,Q_{\pm}])\cap R=\varnothing$,
            \[Pu_{\pm}=PQ_{\pm}u=Q_{\pm}Pu +[P,Q_{\pm}]u\in \dot{C}^{\infty}(M).\]

            Define $a_M^{\pm}=x^{-\frac{n-1}{2}}e^{\pm i\frac{\lambda}{x}}u_{\pm}$ and conjugate $P$ to $\hat{P}_{\pm}=(x^{-\frac{n-1}{2}}e^{\pm\frac{ i\lambda}{x}})P(x^{\frac{n-1}{2}}e^{\mp\frac{ i\lambda}{x}})$. A priori we know that $a_M^{\pm}\in \bigcap_{r<-\frac{n}{2}}\hsc{\infty,r}$ and $\hat{P}_{\pm}a_M^{\pm}\in\dot{C}^{\infty}(M)$; we will show that in fact $a_M^{\pm}\in C^{\infty}(M)$. Since $a_M^{\pm}$ is smooth away from the boundary, let us work in a collar neighborhood of the boundary of the form $\{x<\varepsilon\}$ on which we have coordinates $(x,y_1,\ldots,y_{n-1})$ as above, and let us choose $\varphi\in C_{\mathrm{c}}^{\infty}(M)$ with support in this neighborhood and such that $\varphi=1$ in a neighborhood of $S$. Consider $\varphi a_M^{\pm}\in \bigcap_{r<-\frac{n}{2}}\hsc{\infty,r}$, which also satisfy $\hat{P}(\varphi a_M^{\pm})\in \dot{C}^{\infty}(M)$.

            It follows from Theorem \ref{thm:mod-reg} (where the conjugation of $P$ has moved the radial set to the zero section) that $Q(\varphi a_M^{\pm})\in \bigcap_{r<-\frac{n}{2}} H^{\infty,r}_{\mathrm{sc}}(M) \subset \bigcap_{r<0} x^rL^2_{\mathrm{b}}$ for all $Q\in\diff{b}$, so $\varphi a_M^{\pm}\in\bigcap_{r<0} H^{\infty,r}_{\mathrm{b}}(M)$. Meanwhile, in coordinates $\hat{P}_{\pm}=\pm 2i\lambda x(x\partial_x)+x^2Q$, where $Q\in \diff{b}$. Then 
            \[
            f_{\pm}
            :=
            x\partial_x (\varphi a_M^{\pm})
            =
            \mp\frac{i}{2\lambda}
            \bigg(\frac{1}{x}\hat{P}(\varphi a_M^{\pm})
            +xQ(\varphi a_M^{\pm})\bigg),
            \]
            and since $\hat{P}(\varphi a_M^{\pm})\in\dot{C}^{\infty}(M)$ and $Q(\varphi a_M^{\pm})\in \bigcap_{r<0} H^{\infty,r}_{\mathrm{b}}(M)$, we see that in fact $f_{\pm}\in \bigcap_{r<1}H^{\infty,r}_{\mathrm{b}}(M)$.

            Since $\varphi a^{\pm}_M$ and $f_{\pm}$ are supported in $\{x<\varepsilon\}$, we can consider their Mellin transforms in the $x$ variable, related by $\hat{f}_{\pm}(\sigma,y)=i\sigma\widehat{\varphi a}_M^{\pm}(\sigma,y)$ (see Section~\ref{sec:b-Mellin}). Then we can take advantage of the higher regularity of $f_{\pm}$ to shift the integration path in the Mellin inversion formula downwards and calculate, for any $r\in (0,1)$,
            \[
            \varphi a_M^{\pm}(x,y)
            =
            \frac{1}{2\pi}\int_{\R+ir}x^{i\sigma}\widehat{\varphi a}_M^{\pm}(\sigma,y)\ d\sigma
            =
            \frac{1}{2\pi i}\int_{\R+ir}x^{i\sigma}\frac{\hat{f}_{\pm}(\sigma,y)}{\sigma}\ d\sigma
            =
            \]
            \[
            =
            \hat{f}_{\pm}(0,y)
            +
            \frac{1}{2\pi i}\int_{\R+i(r-1)}x^{i\sigma}\frac{\hat{f}_{\pm}(\sigma,y)}{\sigma}\ d\sigma.
            \]
            The first term is in $C^{\infty}(M)$. The integral remainder, which we call $g$, has Mellin transform $\hat{g}(\sigma,y)=-i\hat{f}_{\pm}(\sigma,y)/\sigma$, so the properties of $\hat{f}_{\pm}$ imply that $(\sigma\mapsto\hat{g}(\sigma,\bullet))\in (1+|\sigma|^2)^{-k/2}L^2(\R+ir; L^2(S))$ for all $k\in\R$ and any $r\in (-1,0)$. This means $g\in \bigcap_{r<1} \hb{\infty,r}\subset\bigcap_{r<1-\frac{n}{2}}\hsc{\infty,r}$.

            By part 1, we know there exists $\Tilde{a}_M^{\pm}\in C^{\infty}(M)$ such that $\Tilde{a}_M^{\pm}(0,y)=\hat{f}_{\pm}(0,y)$ and $\Tilde{u}_{\pm}:=x^{\frac{n-1}{2}}e^{\mp \frac{i\lambda}{x}}\Tilde{a}_M^{\pm}\in\sols$. Then $\Tilde{a}_M^{\pm}(x,y)=\hat{f}_{\pm}(0,y)+xb(x,y)$ for some $b\in C^{\infty}(M)$, so
            \[
            u_{\pm}-\Tilde{u}_{\pm}
            =
            x^{\frac{n-1}{2}}e^{\mp\frac{i\lambda}{x}}(a_M^{\pm}-\Tilde{a}_M^{\pm})=x^{\frac{n-1}{2}}e^{\mp\frac{i\lambda}{x}}((1-\varphi)a_M^{\pm}+g-xb) 
            \]
            belongs to $\bigcap_{r<\frac{1}{2}}\hsc{\infty,r} \subset \hsc{\infty,-\frac{1}{2}}$. Then since $u_{\pm}-\Tilde{u}_{\pm}\in\sols$, Lemma~\ref{thm:wf-bd} shows that in fact $u_{\pm}-\Tilde{u}_{\pm}\in\dot{C}^{\infty}(M)$. Therefore, $a^{\pm}_M=\Tilde{a}_M^{\pm}+x^{-\frac{n-1}{2}}e^{\pm\frac{i\lambda}{x}}(u_{\pm}-\Tilde{u}_{\pm})\in C^{\infty}(M)$. This completes the existence proof.

            To see uniqueness of $a^+_M$ and $a^-_M$ up to Schwartz terms, let us assume that $u\in \text{Sol}_{P,f}$ and for some $a_M^+,a_M^-,b_M^+,b_M^-\in C^{\infty}(M)$, we have
            \[
            u
            =
            x^{\frac{n-1}{2}}\Big(a_M^+e^{-\frac{i\lambda}{x}}+a_M^-e^{\frac{i\lambda}{x}}\Big)
            =
            x^{\frac{n-1}{2}}\Big(b_M^+e^{-\frac{i\lambda}{x}}+b_M^-e^{\frac{i\lambda}{x}}\Big).
            \]
            Defining $c_M^{\pm}=a_M^{\pm}-b_M^{\pm}$, we get $c_M^+e^{-\frac{i\lambda}{x}}+c_M^-e^{\frac{i\lambda}{x}}=0$. In local coordinates $(x,y_1,\ldots,y_{n-1})$ on a collar neighborhood  of $S$, we can use Taylor's theorem in the $x$ variable to arbitrary order $k$ to write
            \[
            c_M^{\pm}(x,y)=c^{\pm}_0(y)+\sum_{j=1}^k c^{\pm}_k(y) x^k+o(x^k)
            \text{ as }x\to 0\text{ and }y=\mathrm{const}.
            \]

            Then $(c_0^+e^{-\frac{i\lambda}{x}}+c_0^-e^{\frac{i\lambda}{x}})\to 0$
            as $x\to 0$ along any curve of constant $y$; but because of the oscillatory factors this is only possible if $c_0^+(y)=c_0^-(y)=0$ there. By the same argument, all the higher-order coefficients must also be zero. This means that $c_M^{\pm}$ and all their $x$-derivatives vanish on $S$. Since $S$ itself is a surface of constant $x$, the fact that $c_M^{\pm}$ are constant there implies that the $y$-derivatives also vanish there. Similarly, the fact that any $x$-derivative is constant on $S$ implies that any mixed derivative vanishes on $S$. Therefore, $c_M^{\pm}\in \dot{C}^{\infty}(M)$, i.e. $a_M^{\pm}$ and $b_M^{\pm}$ are in the same class in $C^{\infty}(M)/\dot{C}^{\infty}(M)$.

        \item 
            Let us assume we are given $a^+\in C^{\infty}(S)$ (the case of $a^-$ is analogous) and $f\in\dot{C}^{\infty}(M)$. By part 1, there exists $u_+\in\sols$ of the form $u_+=x^{\frac{n-1}{2}}e^{-\frac{i\lambda}{x}}a_M^+$ for some $a_M^+\in C^{\infty}(M)$ with $a_M^+|_S=a^+$. Now let $u_-=P_-^{-1}(f-Pu_+)\in \sols$. By part 2, there exist $\tilde{a}_M^+,\tilde{a}_M^-\in C^{\infty}(M)$ such that $u_-=x^{\frac{n-1}{2}}(e^{-\frac{i\lambda}{x}}\tilde{a}_M^++e^{\frac{i\lambda}{x}}\tilde{a}_M^-)$; moreover, since $\wf(u_-)\subset R^-$, Eq.~(\ref{eq:wf-bd}) implies $\tilde{a}_M^+|_S=0$. Then $u=u_++u_-\in \solf$ has the desired form $u=x^{\frac{n-1}{2}}(e^{-\frac{i\lambda}{x}}a_M^++e^{\frac{i\lambda}{x}}a_M^-)$.

            To see uniqueness of $u$ and $a^-$, assume that $v=x^{\frac{n-1}{2}}\Big(b_M^+e^{-\frac{i\lambda}{x}}+b_M^-e^{\frac{i\lambda}{x}}\Big)$ is another element of $\solf$, where $b_M^+,b_M^-\in C^{\infty}(S)$ with $b_M^+|_S=a^+$. Then
            \[
            u-v
            =
            x^{\frac{n-1}{2}}\Big((a_M^+-b_M^+)e^{-\frac{i\lambda}{x}}+(a_M^--b_M^-)e^{\frac{i\lambda}{x}}\Big)
            \in\sol.
            \]
            Since $(a_M^+-b_M^+)|_S=0$, by Lemma \ref{thm:wf-bd} we have $\wf(u-v)\in R^-$, so $u-v\in\X_{--}$. Then since $P_{-}$ is invertible, $P(u-v)=0$ implies $u-v=0$. This proves uniqueness of $u$. Finally, given $u$, we already know from part 2 that $a_M^-$ is determined uniquely up to a Schwartz error, so in particular its restriction to the boundary $a^-$ is uniquely defined.
  \end{enumerate}
\end{proof}

    Theorem \ref{thm:bd} establishes the existence of
        \[\rho^{\pm}:u\mapsto a^{\pm}=a_M^{\pm}|_S\]
    as maps $\sols\to C^{\infty}(S)$ which restrict to invertible maps $\rho_f^{\pm}:\solf\to C^{\infty}(S)$ for any fixed $f\in \dot{C}^{\infty}(M)$. By part 1, they also induce bijective $\rho^{\pm}_{\infty}:P_{\pm\pm}^{-1}\dot{C}^{\infty}(M)/\dot{C}^{\infty}(M)\to C^{\infty}(S)$, where
    \[P_{\pm\pm}^{-1}\dot{C}^{\infty}(M)=\X_{\pm\pm}\cap\sols=\{u\in\sols\ \mid\ \rho^{\mp}u=0\}.\]
    We call $\rho^+$ and $\rho^-$ the incoming and outgoing Feynman boundary data maps respectively. The inverses $\mathcal{U}_f^{\pm}=(\rho^{\pm}_f)^{-1}$ are the incoming and outgoing exact Poisson operators for the equation $Pu=f$, while $\mathcal{U}_{\infty}^{\pm}=(\rho^{\pm}_{\infty})^{-1}$ are the purely-incoming and purely-outgoing approximate Poisson operators.
    
    We define the causal boundary data maps
    \begin{equation}
    \rho^{+-}=\pi_1\circ\rho^++\pi_2\circ\rho^-,
    \hskip 20pt
    \rho^{-+}=\pi_1\circ\rho^-+\pi_2\circ\rho^+,
        \label{eq:bd-causal}
    \end{equation}
    where $\pi_1,\pi_2:C^{\infty}(S_1)\oplus C^{\infty}(S_2)\to C^{\infty}(S_1)\oplus C^{\infty}(S_2)$ are projections onto the first and second component of the boundary respectively. For the sake of uniformity of notation, we will often also denote $\rho^{\pm}$ by $\rho^{\pm\pm}$. In the setting we are considering, where the causal Fredholm problems are in fact invertible, the causal boundary data maps have mapping properties analogous to those of the Feynman versions.
    \begin{prop}
    $\rho^{\pm\mp}:\sols\to C^{\infty}(S)$ restrict to bijective maps $\rho^{\pm\mp}:\solf\to C^{\infty}(S)$ for every $f\in\dot{C}^{\infty}(M)$.
    \end{prop}
    \begin{proof}\ 
                Assume $u,v\in\solf$ and $\rho^{\pm\mp}u=\rho^{\pm\mp}v$. Then $u-v\in\sol$ and $\rho^{\pm\mp}(u-v)=0$. By Eq.~(\ref{eq:wf-bd}), this means that $u-v\in\X_{\mp\pm}$, and since $P_{\mp\pm}$ is invertible, $P(u-v)=0$ implies $u-v=0$. This proves injectivity of $\rho^{\pm\mp}$.
    
                To prove surjectivity, we repeat the existence argument in the proof of Theorem~\ref{thm:bd}, parts 1 and 3 with appropriate modifications, now considering the incoming and outgoing terms simultaneously. Consider $a^{+-}\in C^{\infty}(S)$. Then we can choose any two functions $a_0^{\pm}\in C^{\infty}(M)$ such that $a_0^+|_{S_1}=a^{+-}|_{S_1}$ and $a_0^-|_{S_2}=a^{+-}|_{S_2}$. For $P_j$ as in Eq.~(\ref{eq:Taylor-Pu}), define $a_j^{\pm}=\pm\frac{P_{j-1}a_{j-1}^{\pm}}{2i\lambda j}$ for all $j\in\N$. By Borel's lemma, we can find two functions $a_M^{\pm}\in C^{\infty}(M)$ such that $\partial_x^ja_M^{\pm}|_S=a_j^{\pm}|_S$. Then if we define $u_{+-}=x^{\frac{n-1}{2}}(e^{-i\frac{\lambda}{x}}a_M^++e^{i\frac{\lambda}{x}}a_M^-)$, by construction we have $u_{+-}\in\sols$ and $\rho^{+-}u_{+-}=a^{+-}$. Finally, defining $u_{-+}=P_{-+}^{-1}(f-Pu_{+-})$ and $u=u_{+-}+u_{-+}$, we get $u\in\solf$ and $\rho^{+-}u=a^{+-}$.
    \end{proof}     

    The following pairing formula will be our main tool when relating properties of solutions to those of their boundary data.

        \begin{prop}{\textbf{(Boundary pairing formula; Proposition 13 in \cite{Melrose-AES})}\\}
    \label{thm:bpf}
        Let
        \[u=x^{\frac{n-1}{2}}\Big(a_M^+e^{-\frac{i\lambda}{x}}+a_M^-e^{\frac{i\lambda}{x}}\Big),
        \hskip 20pt
        v=x^{\frac{n-1}{2}}\Big(b_M^+e^{-\frac{i\lambda}{x}}+b_M^-e^{\frac{i\lambda}{x}}\Big)\]
        for some $a_M^{\pm},b_M^{\pm}\in C^{\infty}(M)$. Let $a^{\pm}=a_M^{\pm}|_S$, $b^{\pm}=b_M^{\pm}|_S$. If $Pu,Pv\in \dot{C}^{\infty}(M)$, then
        \[
        \langle u,Pv\rangle_{L^2_{\g}}-\langle Pu,v\rangle_{L^2_{\g}}
        =
        2i\lambda\Big(\langle a^+,b^+\rangle_{L^2_{\h}} -\langle a^-,b^-\rangle_{L^2_{\h}}\Big).\]
    \end{prop}
    \begin{proof}
        Consider a collar neighborhood of the form $U\simeq [0,\delta)_x\times S$ on which we have coordinates $(x,y_1,\ldots,y_{n-1})$ as above. For any $\varepsilon\in (0,\delta)$, define $M_{\varepsilon}=M\backslash (U\cap\{x<\varepsilon\})$ and $S_{\varepsilon}=\partial M_{\varepsilon}=U\cap\{x=\varepsilon\}$. Since $Pu,Pv\in\dot{C}^{\infty}(M)$, we have
        \[
        \langle u,Pv\rangle_{L^2_{\g}}-\langle Pu,v\rangle_{L^2_{\g}}
        =
        \lim_{\varepsilon\to 0^+}
        \int_{M_{\varepsilon}} (\bar{u}Pv-v P\bar{u})\ \dv
        =
        \lim_{\varepsilon\to 0^+}
        \int_{M_{\varepsilon}} (\bar{u}\Delta_{\g} v-v\Delta_{\g}\bar{u})\ \dv.
        \]
        We can rewrite the integral as
        \[
        \int_{M_{\varepsilon}} (\bar{u}\Delta_{\g} v-v \Delta_{\g}\bar{u})\ \dv
        =
        \int_{M_{\varepsilon}} \operatorname{div}( v\nabla \bar{u} - \bar{u} \nabla v)\ \dv
        =
        \int_{S_{\varepsilon}} \g( v\nabla \bar{u}-\bar{u} \nabla v,\hat{n})\ dA_{\g,\varepsilon},
        \]
        where the last equality is by the divergence theorem. Here $dA_{\g,\varepsilon}$ is the area density induced by $\g$ on $S_{\varepsilon}$ and $\hat{n}$ is the outward unit normal to $S_{\varepsilon}$ in $M$. In coordinates, for any $f\in C^{\infty}(M^{\circ})$ we have $\g(\nabla f,\hat{n})=-\frac{1}{\sqrt{g^{00}}}\Big(g^{00}\partial_x f+\sum_{j=1}^{n-1}g^{0j}\partial_{y_j} f\Big)$. By direct computation, $v\partial_x\bar{u}-\bar{u}\partial_x v=x^{n-3}((2i\lambda(\overline{a_M^-}b_M^--\overline{a_M^+}b_M^+)+\mathcal{O}(x^2))$ and $\bar{u}\partial_{y_j} v-v\partial_{y_j}\bar{u}=\mathcal{O}(x^{n-1})$. For a metric of our form, $g^{00}=x^4(1+\mathcal{O}(x^2))$ and $g^{0j}=\mathcal{O}(x^4)$, so we get
        \[
        \g(v\nabla \bar{u} - \bar{u} \nabla v,\hat{n})
        =x^{n-1}\Big(2i\lambda(\overline{a_M^+}b_M^+-\overline{a_M^-}b_M^-)\Big)+\mathcal{O}(x^{n+1}).
        \]
        Meanwhile, $x^{n-1}dA_{\g,\varepsilon}=dA_{\h,\varepsilon}$, where $dA_{\h,\varepsilon}$ is the area density induced on $S_{\varepsilon}$ by $h_{\varepsilon}$. Therefore,
        \[
        \int_{S_{\varepsilon}} \g( v\nabla \bar{u} - \bar{u} \nabla v,\hat{n})\ dA_{\g,\varepsilon}
        =
        2i\lambda\int_{S_{\varepsilon}}(\overline{a_M^+}b_M^+-\overline{a_M^-}b_M^-)\ dA_{\h,\varepsilon}+\mathcal{O}(\varepsilon^2).
        \]
        Finally, taking the limit $\varepsilon\to 0$ completes the proof.
    \end{proof}

    In terms of the boundary data maps, the formula says that for $u,v\in\sols$, we have
    \begin{equation}
        \langle u, Pv\rangle_{L^2_{\g}}-\langle Pu,v\rangle_{L^2_{\g}}
            =
            2i\lambda\Big(\big\langle \rho^{++}u,\rho^{++}v\big\rangle_{L^2_{\h}}-\big\langle \rho^{--}u,\rho^{--}v\big\rangle_{L^2_{\h}}\Big),
            \label{eq:bpf-feynman}
    \end{equation}
    \begin{equation}
        \langle u, Pv\rangle_{L^2_{\g}}-\langle Pu,v\rangle_{L^2_{\g}}
            =
            2i\lambda\Big(\big\langle (\pi_1-\pi_2)\rho^{+-}u,\rho^{+-}v\big\rangle_{L^2_{\h}}-\big\langle (\pi_1-\pi_2)\rho^{-+}u,\rho^{-+}v\big\rangle_{L^2_{\h}}\Big).
            \label{eq:bpf-causal}
    \end{equation}

\subsection{Continuity properties}
\label{sec:cont}
    Finally, we show that the boundary data maps and Poisson operators are continuous. The Schwartz kernel of the homogeneous problem's Poisson operator was constructed by Melrose and Zworski in \cite{M-Z}; however, for our purposes less explicit arguments based on the boundary pairing formula will be enough.

    For $m\in\N$, let $U_m\subset \tsccc$ be neighborhoods of $R$ such that $U_{m+1}\subset U_m$ and $\bigcap_{m\in\N}U_m=R$. Choose functions $r_m\in C^{\infty}(\tsccc)$ such that $r_m>-\frac{1}{2}-\frac{1}{m}$ everywhere, $r_m<-\frac{1}{2}$ at $R$, $r_m>m$ outside $U_m$, and $r_{m+1}>r_m$. Then $\sols=\bigcap_{m\in\N}\{u\in \hsc{m,r_m}\ \mid\ Pu\in \dot{C}^{\infty}(M)\}$. We can consider $\sols$ with the Fr\'echet-space topology generated by the family of norms $\lVert u\rVert _{\mathrm{Sol},m}^2=\lVert u\rVert _{m,r_m}^2+\lVert Pu\rVert _{m,m}^2$.

    \begin{prop}\
    \label{thm:cont-Schwartz}
        \begin{enumerate}
            \item 
                The maps $\rho^{\pm}:\sols\to C^{\infty}(S)$ are continuous when $\sols$ is considered with the topology defined above and $C^{\infty}(S)$ with the standard topology.
            \item 
                For any $f\in\sols$, the maps $\mathcal{U}_f^{\pm}:C^{\infty}(S)\to \solf$ are continuous with respect to the distributional topologies induced on $C^{\infty}(S)$ from $C^{-\infty}(S)$ and on $\solf$ from $C^{-\infty}(M)$.
        \end{enumerate}
    \end{prop}
    Note that continuity of the Poisson operators in the distributional topology on $C^{\infty}(S)$ also implies continuity in the standard topology.
    \begin{proof}\
        \begin{enumerate}
            \item 
                Let $Q_{\pm}\in \psco$ be such that $\wf'(Q_{\pm})\cap R^{\mp}=\varnothing$ and $\wf'(I-Q_{\pm})\cap R^{\pm}=\varnothing$. Then
                \[\rho^{\pm}u=\Big(x^{-\frac{n-1}{2}}e^{\pm i\frac{\lambda}{x}}\cdot Q_{\pm}u\Big)\Big|_{S}.\]
                $Q_{\pm}'=x^{-\frac{n-1}{2}}e^{\pm i\frac{\lambda}{x}}\cdot Q_{\pm}$ is a continuous operator on $C^{-\infty}(M)$. The fact that it maps $\sols$ into $C^{\infty}(M)$ has already been established. Continuity between these spaces can be seen using the closed graph theorem (for linear maps between Fr\'echet spaces) as follows. Assume that $u_n\to u$ in $\sols$ and $Q_{\pm}'u_n\to a$ in $C^{\infty}(M)$. By continuity in distributions, we must have $Q_{\pm}'u_n\to Q_{\pm}'u$ in $C^{-\infty}(M)$. But since the standard topology of $C^{\infty}(M)$ is stronger than the distributional topology, we must have $Q_{\pm}'u_n\to a$ in $C^{-\infty}(M)$. Therefore, $Q_{\pm}'u=a$. This proves that the graph of $Q_{\pm}'$ is closed, so $Q_{\pm}'$ is continuous.
    
                Finally, restriction to the boundary is a continuous map $C^{\infty}(M)\to C^{\infty}(S)$, so we see that $\rho^{\pm}:\sols\to C^{\infty}(S)$ is continuous as a composition of continuous maps.
    
            \item 
                The boundary pairing formula implies that for any $a\in C^{\infty}(S)$ and $f,g\in \dot{C}^{\infty}(M)$,
                \begin{equation}
                    \langle \mathcal{U}_f^{\pm}a, g\rangle_{L^2_{\g}} = \langle f,P_{\pm\pm}^{-1}g\rangle_{L^2_{\g}}
                            \pm 2i\lambda\langle a,\rho^{\pm}P_{\pm\pm}^{-1}g\rangle_{L^2_{\h}}.
                            \label{eq:poisson-inhom}
                \end{equation}
                This expression defines $\mathcal{U}_f^{\pm}a$ as a tempered distribution on $M$. The first term is independent of $a$, so to prove continuity of $\mathcal{U}_f^{\pm}$, it is enough to prove continuity of the homogeneous problem's Poisson operator $\mathcal{U}_0^{\pm}:C^{\infty}(S)\to\sol$, which is given by $\langle \mathcal{U}_0^{\pm}a, g\rangle_{L^2_{\g}} = \pm 2i\lambda\langle a,\rho^{\pm}P_{\pm\pm}^{-1}g\rangle_{L^2_{\h}}$. Taking absolute values yields a seminorm estimate proving continuity of $\mathcal{U}_0^{\pm}$ with respect to the distributional topologies.
        \end{enumerate}
    \end{proof}

    Continuity of the causal boundary data maps $\rho^{\pm\mp}:\sols\to C^{\infty}(S)$ follows directly from the definition Eq.~(\ref{eq:bd-causal}) and the continuity of $\rho^{\pm}$. The causal Poisson operators $\mathcal{U}_f^{\pm\mp}$ are also continuous in the distributional topologies, and the proof is analogous to the Feynman version. Eq.~(\ref{eq:poisson-inhom}) defining the output of the Poisson operator as a distribution is replaced by
    \begin{equation}
        \langle \mathcal{U}_f^{\pm\mp}a,g\rangle_{L^2_{\g}}=\langle f, P_{\pm\mp}^{-1}g\rangle_{L^2_{\g}} \pm 2i\lambda\langle a,(\pi_1-\pi_2)\rho^{\pm\mp}P_{\pm\mp}^{-1}g\rangle_{L^2_{\h}}.
        \label{eq:poisson-causal}
    \end{equation}

\section{Quantum fields and states}
\label{sec:fields-states}
	In QFT, the Hermitian form defining the (anti)commutator of fields arises from a difference of propagators, normally the causal (retarded/advanced) pair. We now use the tools of the previous section to define analogues of the constructions of Vasy and Wrochna \cite{V-W}. Namely, we show that the difference of either pair of complementary inverses of $P$ defines a Hermitian form on $\dot{C}^{\infty}(M)/P\dot{C}^{\infty}(M)$, allowing one to define a QFT model. The Hermitian form can be alternatively understood in terms of its action on the space of ``classical solutions'' $\sol$ or the space of boundary data $C^{\infty}(S)$. By decomposing the boundary data into its projections onto the two boundary components, we split the Hermitian form into two terms, thereby defining two-point functions of gauge-invariant quasi-free states of the model.

\subsection{Mapping properties of the difference of propagators}

    We can define $G_{++}=P_{++}^{-1}-P_{--}^{-1}$ as a map $\Y_{++}\cap\Y_{--}\to\X_{++}+\X_{--}$ and $G_{+-}=P_{+-}^{-1}-P_{-+}^{-1}$ as a map $\Y_{+-}\cap\Y_{-+}\to\X_{+-}+\X_{-+}$. In fact the domains and codomains are the same for either pair, since
        \[\Y_{\alpha\beta}\cap\Y_{-\alpha\beta}=\X_{\alpha\beta}\cap \X_{-\alpha\beta}=\dot{C}^{\infty}(M),\]
        \[\X_{\alpha\beta}+\X_{-\alpha\beta}=\{u\in C^{-\infty}(M)\ \mid\ \wf(u)\subset R\text{ and }Pu\in\cap_{r<\frac{1}{2}}\hsc{\infty,r}\},\]
    where once again $\X_{\alpha\beta}+\X_{-\alpha\beta}\subset \bigcap_{r<-\frac{1}{2}}\hsc{\infty,r}$ by the below-threshold estimate.
    
    \begin{prop}{\textbf{(Analogue of Proposition 4.2 in \cite{V-W})}}
        \begin{enumerate}
            \item 
                $\ker G_{++}=\ker G_{+-}=P\dot{C}^{\infty}(M)$,
            \item 
                $\ran G_{++}=\ran G_{+-}=\sol$.
        \end{enumerate}
    \end{prop}
    
    \begin{proof}
    Fix $(\alpha,\beta)=(+,+)$ or $(+,-)$. Denote $\X_{\pm}=\X_{\pm\alpha\beta}$, $\Y_{\pm}=\Y_{\pm\alpha\beta}$, $P_{\pm}=P_{\pm\alpha\beta}$, and $G=G_{\alpha\beta}=P_+^{-1}-P_-^{-1}$.
        \begin{enumerate}
            \item 
                
                     Assume $f=Pu$ for some $u\in \dot{C}^{\infty}(M)=\X_+\cap\X_-$. Then $P_+u=P_-u=f$, so $Gf=0$. This proves that $P\dot{C}^{\infty}(M)\subset\ker G$.
                        
                     On the other hand, assume $Gf=0$, i.e. $P_+^{-1}f=P_-^{-1}f=u$ for some $u$. Since $P_+^{-1}f\in\X_+$ and $P_-^{-1}f\in\X_-$, we get $u\in \X_+\cap \X_-=\dot{C}^{\infty}(M)$. Since $f=Pu$, this proves that $\ker G\subset P\dot{C}^{\infty}(M)$.
            \item
            
                        Since $G$ is the difference of two inverses of $P$, we always have $PGf=0$, so 
                        \[\ran G\subset\{u\in\X_++\X_-\ \mid\ Pu=0\}=\sol.\]
                    
                        For the other direction, let $u\in\sol$. Let us choose $Q\in\psco$ such that $\wf'(Q)\cap R^{-\alpha\beta}=\varnothing$ and $\wf'(I-Q)\cap R^{\alpha\beta}=\varnothing$. Since $\wf(u)\subset R$ and $\wf'(Q)\cap R^{-\alpha\beta}=\varnothing$, we get $\wf(Qu)\subset R^{\alpha\beta}$. Let us write
                            \[P(Qu)=QPu+[P,Q]u=[P,Q]u,\]
                        since $Pu=0$. Since $\wf(u)\subset R$ and $\wf'([P,Q])\cap R=\varnothing$, we get $[P,Q]u\in\dot{C}^{\infty}(M)$. This proves that $P(Qu)\in\Y_+$, so $Qu\in \X_+$. An analogous argument shows that $(I-Q)u\in \X_-$.
                        
                        Then, using the mapping properties of $P_{\pm}$ and the fact that $Pu=0$, we get
                            \[
                            u
                            =
                            Qu+(I-Q)u
                            =
                            P_+^{-1}PQu+P_-^{-1}P(I-Q)u
                            =
                            \]
                            \[
                            =
                            P_+^{-1}[P,Q]u+P_-^{-1}[P,I-Q]u
                            =
                            G[P,Q]u.
                            \]
                        Thus, $u$ is in the range of $G$. This proves that $\ran G=\sol$.
        \end{enumerate}
    \end{proof}
    
    To summarize, then, $G_{++}$ and $G_{+-}$ induce linear bijections
    \[[G_{++}],[G_{+-}]:\frac{\dot{C}^{\infty}(M)}{P\dot{C}^{\infty}(M)}\to\sol.\]
    Moreover, the proof shows that the inverse of $[G_{+\pm}]$ maps any $u\in\sol$ to the class of $[P,Q]u\in \dot{C}^{\infty}(M)$ for any $Q\in\psco$ such that $\wf'(Q)\cap R^{-\mp}=\varnothing$ and $\wf'(I-Q)\cap R^{+\pm}=\varnothing$.

    \begin{prop}
        $P_{\alpha\beta}^{-1}:\dot{C}^{\infty}(M)\to\sols$ and $G_{+\pm}:\dot{C}^{\infty}(M)\to\sol$ are continuous with respect to the topology on $\sols$ defined in Section~\ref{sec:cont}.
    \end{prop}
    \begin{proof}
        Let $r_m$ be the variable orders defined in Section~\ref{sec:cont}. For any $m\in\N$ and any choice of $(\alpha,\beta)\in\{+,-\}$, there exists $r_{\alpha\beta}>r_m$ satisfying the $(\alpha,\beta)$-threshold condition. Then 
        \[
        \lVert P_{\alpha\beta}^{-1}f\rVert _{\mathrm{Sol},m}^2
        =
        \lVert P_{\alpha\beta}^{-1}f\rVert _{m,r_m}^2 + \lVert f\rVert _{m,m}^2
        \leqslant
        \lVert P_{\alpha\beta}^{-1}f\rVert _{m,r_{\alpha\beta}}^2 + \lVert f\rVert _{m,m}^2
        \leqslant
        \]
        \[
        \leqslant 
        C\lVert f\rVert _{m-1,r_{\alpha\beta}+1}^2+\lVert f\rVert _{m,m}^2.
        \]
        The last inequality comes from the continuity of $P_{m,r_{\alpha\beta}}^{-1}:\Y_{m,r_{\alpha\beta}}\to \X_{m,r_{\alpha\beta}}$ (which follows from the Fredholm estimates). Since the topology of $\dot{C}^{\infty}(M)$ can be generated by the weighted Sobolev norms, this proves the continuity of $P_{\alpha\beta}^{-1}:\dot{C}^{\infty}(M)\to \sols$; and since $G_{+\pm}=P_{+\pm}^{-1}-P_{-\mp}^{-1}$, continuity of $G$ follows.
    \end{proof}

\subsection{Hermitian structure and positivity}
	Now we show that the difference of propagators endows the space $\dot{C}^{\infty}(M)/P\dot{C}^{\infty}(M)$, and therefore also the space $\sol$ of solutions to the homogeneous problem, with a Hermitian inner product structure which is expressed particularly clearly in terms of the boundary data of solutions.

    \begin{theorem}{\textbf{(Analogue of Theorem 4.3. in \cite{V-W})}}
        \begin{enumerate}
            \item For any $f,g\in \dot{C}^{\infty}(M)$ we have $\langle f,P_{\alpha\beta}^{-1}g\rangle_{L^2_{\g}}=\langle P_{-\alpha\beta}^{-1}f,g\rangle_{L^2_{\g}}$.
            \item 
                $D_{+\pm}([f],[g])=i\langle f,G_{+\pm}g\rangle_{L^2_{\g}}$ are nondegenerate Hermitian sesquilinear forms on $\frac{\dot{C}^{\infty}(M)}{P\dot{C}^{\infty}(M)}$. 
            \item (Positivity in the Feynman case) 
                For any $f\in \dot{C}^{\infty}(M)$ we have $i\langle f, G_{++}f\rangle_{L^2_{\g}}\geqslant 0$, with equality if and only if $f\in P\dot{C}^{\infty}(M)$.
        \end{enumerate}
    \end{theorem}
    
    \begin{proof}\ 
        \begin{enumerate}
            \item We compute
                \[
                \langle f,P_{\alpha\beta}^{-1}g\rangle
                =
                \langle PP_{-\alpha\beta}^{-1}f,P_{\alpha\beta}^{-1}g\rangle
                =
                \langle P_{-\alpha\beta}^{-1}f,PP_{\alpha\beta}^{-1}g\rangle
                =
                \langle P_{-\alpha\beta}^{-1}f,g\rangle,
                \]
                where the integration by parts is justified by Lemma~\ref{thm:int-by-parts}, since $\wf(P_{\pm\alpha\beta}^{-1}f)\subset R^{\pm\alpha\beta}$ and $R^{\alpha\beta}$ and $R^{-\alpha\beta}$ are always disjoint.
                
            \item 
                Fix $(\alpha,\beta)=(+,+)$ or $(+,-)$. Denote $P_{\pm}=P_{\pm\alpha\beta}$ and $G=G_{\alpha\beta}=P_+^{-1}-P_-^{-1}$, $D=D_{\alpha\beta}$.
            
                For $f,g\in \dot{C}^{\infty}(M)$, define the sesquilinear form $\Tilde{D}(f,g)=i\langle f, Gg\rangle$. Hermiticity follows from part 1:
                    \[
                    \Tilde{D}(f,g)
                    =
                    i\langle f, Gg\rangle
                    =
                    i\Big(\langle f, P_+^{-1}g\rangle - \langle f, P_-^{-1}g\rangle\Big)
                    =
                    \]
                    \[
                    =
                    i\Big(\langle P_-^{-1}f,g\rangle - \langle P_+^{-1}f, g\rangle\Big)
                    =
                    -i\langle Gf, g\rangle
                    =
                    -i\overline{\langle g, Gf\rangle}
                    =
                    \overline{i\langle g, Gf\rangle}
                    =
                    \overline{\Tilde{D}(g,f)}.
                    \]

                We have already seen that $\ker G=P\dot{C}^{\infty}(M)$, so for any $g\in P\dot{C}^{\infty}(M)$ we get $\Tilde{D}(f,g)=i\langle f,Gg\rangle=0$. By hermiticity we also get $\Tilde{D}(f,g)=0$ if $f\in P\dot{C}^{\infty}(M)$. Thus, $D([f],[g])=\Tilde{D}(f,g)$ is well-defined as a Hermitian form on $\dot{C}^{\infty}(M)/P\dot{C}^{\infty}(M)$.
                
                To see that $D$ is nondegenerate, assume there exists $g\in \dot{C}^{\infty}(M)$ such that for all $f\in\dot{C}^{\infty}(M)$ we have $\Tilde{D}(f,g)=0$. Since the $L^2$ inner product is nondegenerate on $\dot{C}^{\infty}(M)$, this implies $Gg=0$, and we have already seen that $\ker G=P\dot{C}^{\infty}(M)$. This proves nondegeneracy of $D$ on the quotient.        
                
            \item
                By part 1,
                    \[
                    \langle f, G_{++}f\rangle
                    =
                    \langle f, P_{++}^{-1}f\rangle - \langle f, P_{--}^{-1}f\rangle
                    =
                    \langle P_{--}^{-1}f,f\rangle -\langle f,P_{--}^{-1}f\rangle.
                    \]
                By the boundary pairing formula,
                    \[
                    \langle P_{--}^{-1}f,f\rangle_{L^2_{\g}} -\langle f,P_{--}^{-1}f\rangle_{L^2_{\g}}
                    =
                    -2i\lambda\langle \rho^-P_{--}^{-1}f,\rho^-P_{--}^{-1}f\rangle_{L^2_{\h}}.
                    \]
                Therefore,
                    \[i\langle f,G_{++}f\rangle_{L^2_{\g}}=2\lambda\lVert  \rho^-P_{--}^{-1}f\rVert_{L^2_{\h}} ^2\geqslant 0.\]
                Equality will be attained if and only if $\rho^-P_{--}^{-1}f=0$. Since $\wf(P_{--}^{-1}f)\subset R^-$, this condition is equivalent to $P_{--}^{-1}f\in \dot{C}^{\infty}(M)$ by Eq.~(\ref{eq:wf-bd}). This is in turn equivalent to $f\in P\dot{C}^{\infty}(M)$.
        \end{enumerate}
    \end{proof}

    This Hermitian structure can be transported onto $\sol$ using the isomorphisms defined by $G_{+\pm}$. We would also like to understand the Hermitian structure further induced on the space of boundary data by the isomorphisms $\rho^{\alpha\beta}$.
    
    Let $f,g\in \dot{C}^{\infty}(M)$. For $(\alpha,\beta)=(+,+)$ or $(+,-)$, define
    \begin{equation}
    a^{\pm\alpha\beta}=\rho^{\pm\alpha\beta}G_{\alpha\beta}f=\pm\rho^{\pm\alpha\beta}P_{\alpha\beta}^{-1}f,
    \hskip 20pt
    b^{\pm\alpha\beta}=\rho^{\pm\alpha\beta}G_{\alpha\beta}g=\pm\rho^{\pm\alpha\beta}P_{\alpha\beta}^{-1}g.
        \label{eq:bd-notation}
    \end{equation}
    Using the boundary pairing formula, we can compute
        \[
        D_{++}([f],[g])
        =
        \pm i\Big(\langle f,P_{\pm\pm}^{-1}g\rangle_{L^2_{\g}}-\langle P_{\pm\pm}^{-1}f, g\rangle_{L^2_{\g}}\Big)
        =
        2\lambda\langle a^{\pm\pm},b^{\pm\pm}\rangle_{L^2_{\h}},
        \]
        \[
        D_{+-}([f],[g])
        =
        \pm i\Big(\langle f,P_{\pm\mp}^{-1}g\rangle_{L^2_{\g}}-\langle P_{\pm\mp}^{-1}f, g\rangle_{L^2_{\g}}\Big)
        =
        2\lambda\Big\langle (\pi_1-\pi_2)a^{\pm\mp},b^{\pm\mp}\Big\rangle_{L^2_{\h}}.
        \]
    
    To summarize, then, we have the following isomorphisms of Hermitian inner product spaces:
   \[\begin{tikzcd}
   {\ } & {\ } & {\Big(C^{\infty}(S), D_{++}^{\mathrm{bd}}\Big)} \\
   	{\Big(\dot{C}^{\infty}(M)/P\dot{C}^{\infty}(M)}, D_{++}\Big) 
    & {\Big(\sol, D_{++}^{\mathrm{sol}}\Big)} 
    & {\ }\\
    {\ } & {\ } & {\Big(C^{\infty}(S), D_{++}^{\mathrm{bd}}\Big)} 
   	\arrow["{\rho^{++}}", from=2-2, to=1-3]
   	\arrow["G_{++}", from=2-1, to=2-2]
    \arrow["{\rho^{--}}", from=2-2, to=3-3]
   \end{tikzcd}\]
   
   \[\begin{tikzcd}
   {\ } & {\ } & {\Big(C^{\infty}(S), D_{+-}^{\mathrm{bd}}\Big)} \\
   	{\Big(\dot{C}^{\infty}(M)/P\dot{C}^{\infty}(M)}, D_{+-}\Big) 
    & {\Big(\sol, D_{+-}^{\mathrm{sol}}\Big)} 
    & {\ }\\
    {\ } & {\ } & {\Big(C^{\infty}(S), D_{+-}^{\mathrm{bd}}\Big)} 
   	\arrow["{\rho^{+-}}", from=2-2, to=1-3]
   	\arrow["G_{+-}", from=2-1, to=2-2]
    \arrow["{\rho^{-+}}", from=2-2, to=3-3]
   \end{tikzcd}\]
   
    The various realizations of the Hermitian forms are:
    \begin{itemize}
        \item 
            For $[f],[g]\in \frac{\dot{C}^{\infty}(M)}{P\dot{C}^{\infty}(M)}$, $D_{+\pm}([f],[g])=i\langle f, G_{+\pm}g\rangle_{L^2_{\g}}$;
        \item 
            For $u,v\in\sol$, $D_{+\pm}^{\mathrm{sol}}(u,v)=i\langle [P,Q]u,v\rangle_{L^2_{\g}}$ for any $Q\in\psco$ with $\wf'(Q)\cap R^{-\mp}=\varnothing$ and $\wf'(I-Q)\cap R^{+\pm}=\varnothing$;
        \item 
            For $a,b\in C^{\infty}(S)$, $D_{+\pm}^{\mathrm{bd}}(a,b)=2\lambda\Big(\langle \pi_1 a,\pi_1 b\rangle_{L^2_{\h}}\pm \langle \pi_2 a,\pi_2 b\rangle_{L^2_{\h}}\Big)$.
    \end{itemize}

\subsection{Two-point functions}
    Define $\Lambda_1^{\pm\pm},\Lambda_2^{\pm\pm}:\dot{C}^{\infty}(M)\to\sol$ by
        \begin{equation}
            \Lambda_1^{\pm\pm} f=+i\mathcal{U}_0^{\pm\pm}\pi_1\rho^{\pm\pm}G_{++}f,
            \hskip 50pt
            \Lambda_2^{\pm\pm} f=+i\mathcal{U}_0^{\pm\pm}\pi_2\rho^{\pm\pm}G_{++}f,
            \label{eq:2pt-feynman}
        \end{equation}
    and define $\Lambda_1^{\pm\mp},\Lambda_2^{\pm\mp}:\dot{C}^{\infty}(M)\to\sol$ by
        \begin{equation}
            \Lambda_1^{\pm\mp} f=+i\mathcal{U}_0^{\pm\mp}\pi_1\rho^{\pm\mp}G_{+-}f,
            \hskip 50pt
            \Lambda_2^{\pm\mp} f=-i\mathcal{U}_0^{\pm\mp}\pi_2\rho^{\pm\mp}G_{+-}f.
            \label{eq:2pt-causal}
        \end{equation}
    These operators satisfy $\Lambda_1^{\pm\pm}+\Lambda_2^{\pm\pm}=iG_{++}$ and $\Lambda_1^{\pm\mp}-\Lambda_2^{\pm\mp}=iG_{+-}$, as well as the following properties for any choice of $\alpha,\beta\in\{+,-\}$:
    
    \begin{itemize}
        \item 
            $\Lambda_{1}^{\alpha\beta}, \Lambda_{2}^{\alpha\beta}:\dot{C}^{\infty}(M)\to C^{-\infty}(M)$ are continuous. This follows from the continuity of every map in the diagram
            \[\begin{tikzcd}
   {\dot{C}^{\infty}(M)} & {\sol} & {C^{\infty}(S)} & {C^{\infty}(S)} & {C^{-\infty}(M).}
   	\arrow["G_{+\pm}", from=1-1, to=1-2]
   	\arrow["\rho^{\alpha\beta}", from=1-2, to=1-3]
    \arrow["\pi_j", from=1-3, to=1-4]
    \arrow["\mathcal{U}_0^{\alpha\beta}", from=1-4, to=1-5]
   \end{tikzcd}\]
        \item 
            Since they map into $\sol$, $P\Lambda_1^{\alpha\beta}=P\Lambda_2^{\alpha\beta}=0$.
        \item 
            Since $G_{+\pm}P=0$ on $\dot{C}^{\infty}(M)$, $\Lambda_1^{\alpha\beta}P=\Lambda_2^{\alpha\beta}P=0$. This means that $\Lambda_1^{\alpha\beta}$ and $\Lambda_2^{\alpha\beta}$ are well-defined on the quotient $\dot{C}^{\infty}(M)/P\dot{C}^{\infty}(M)$.
        \item 
            For any $f,g\in \dot{C}^{\infty}(M)$ and $j\in\{1,2\}$, we have $\langle f,\Lambda_j^{\alpha\beta} g\rangle_{L^2_{\g}}=\overline{\langle g,\Lambda_j^{\alpha\beta} f\rangle_{L^2_{\g}}}$. To see this, use the boundary pairing as follows. For $(\alpha,\beta)=(+,+)$ or $(+,-)$, define $a^{\pm\alpha\beta}, b^{\pm\alpha\beta}$ as in Eq.~(\ref{eq:bd-notation}).
            \begin{itemize}
                \item In the Feynman ($+,+$) case:
                \[
            \langle f,\Lambda_j^{\pm\pm}g\rangle
            =
            i\langle f,\mathcal{U}_0^{\pm\pm}\pi_j\rho^{\pm\pm}G_{++}g\rangle
            =
            \pm i\langle f,\mathcal{U}_0^{\pm\pm}\pi_j\rho^{\pm\pm}P_{\pm\pm}^{-1}g\rangle
            =
            \]
            \[
            =
            2\lambda\langle a^{\pm\pm},\pi_j b^{\pm\pm} \rangle
            =
            2\lambda\langle\pi_j a^{\pm\pm},\pi_j b^{\pm\pm}\rangle
            =
            \overline{2\lambda\langle\pi_j b^{\pm\pm},\pi_j a^{\pm\pm}\rangle}
            =
            \overline{\langle g,\Lambda_j^{\pm\pm} f\rangle}.
            \]

                \item In the causal ($+,-$) case:
            \[
            \langle f,\Lambda_j^{\pm\mp}g\rangle
            =
            i(\delta_{1j}-\delta_{2j})\langle f,\mathcal{U}_0^{\pm\mp}\pi_j\rho^{\pm\mp}G_{+-}g\rangle
            =
            \]
            \[
            =
            \pm i(\delta_{1j}-\delta_{2j})\langle f,\mathcal{U}_0^{\pm\mp}\pi_j\rho^{\pm\mp}P_{\pm\mp}^{-1}g\rangle
            =
            2\lambda (\delta_{1j}-\delta_{2j})\langle (\pi_1-\pi_2)a^{\pm\mp},\pi_j b^{\pm\mp}\rangle
            =
            \]
            \[
            =
            2\lambda \langle\pi_j a^{\pm\mp},\pi_j b^{\pm\mp}\rangle
            =
            \overline{2\lambda \langle \pi_jb^{\pm\mp},\pi_ja^{\pm\mp} \rangle}
            =
            \overline{\langle g,\Lambda_j^{\pm\mp} f\rangle}.
            \]
            \end{itemize}
            
        \item 
            For any $f\in \dot{C}^{\infty}(M)$ and $j\in\{1,2\}$, we have $\langle f,\Lambda_j^{\alpha\beta} f\rangle_{L^2_{\g}}\geqslant 0$. We can see this from the computation above:
            \[\langle f, \Lambda_j^{\alpha\beta}f\rangle_{L^2_{\g}}=2\lambda \lVert \pi_j\rho^{\alpha\beta}P_{\alpha\beta}^{-1}f\rVert ^2_{L^2_{\h}}\geqslant 0.\]
    \end{itemize}
    In particular, these properties mean that 
    \begin{itemize}
        \item 
            both $(\Lambda_1^{++},\Lambda_2^{++})$ and $(\Lambda_1^{--},\Lambda_2^{--})$ define valid pairs of \textbf{fermionic} two-point functions on the algebra $\text{CAR}^{\text{pol}}(\dot{C}^{\infty}(M),D_{++})$ by $\lambda_{1,2}^{\pm\pm}(f,g)=\langle f,\Lambda_{1,2}^{\pm\pm}g\rangle_{L^2_{\g}}$;
        \item 
            both $(\Lambda_1^{+-},\Lambda_2^{+-})$ and $(\Lambda_1^{-+},\Lambda_2^{-+})$ define valid pairs of \textbf{bosonic} two-point functions on the algebra $\text{CCR}^{\text{pol}}(\dot{C}^{\infty}(M),D_{+-})$ by $\lambda_{1,2}^{\pm\mp}(f,g)=\langle f,\Lambda_{1,2}^{\pm\mp}g\rangle_{L^2_{\g}}$.
    \end{itemize}

    On the level of the space of boundary data, the Hermitian form $D_{++}$ is (up to a constant factor) simply the standard $L^2$ inner product, and the two-point functions simply decompose it into two terms corresponding to the orthogonal decomposition $C^{\infty}(S)=C^{\infty}(S_1)\oplus C^{\infty}(S_2)$. The situation with $D_{+-}$ is similar, but it is not positive because it combines the $L^2$ inner products on the two components of the boundary with opposite signs, and correspondingly the decomposition is into a difference, not a sum, of positive two-point functions.

    The reason that in either case we arrive at two different pairs of two-point functions is the fact that in our construction, a pair of two-point functions is specified by prescribing the behavior of one piece of boundary data (incoming or outgoing) at each end. The other two pieces then depend on the scattering matrix. This means that after we specify if we will be using Feynman or causal data, there are two choices left of what to fix, which in practice depends on what process we want the state to model. This is similar to in/out vacuum states on asymptotically Minkowski spacetimes, which are defined to have properties similar to the Minkowski vacuum either in the infinite past or the infinite future and in general define different representations of the algebra of observables which are connected by a nontrivial S-matrix.

\section{Extending the two-point functions to distributions}
\label{sec:extension}

  In this section, we extend the domain of the two-point functions defined above and show that they satisfy a wavefront mapping property analogue of the Hadamard condition.
  
    \subsection{Extending $G$ to distributions}
    \label{sec:ext-G}
        In the context of Section \ref{sec:bd-Poisson}, instead of intersecting we can take the union over all $s\in\R$ and $r_{\alpha\beta}$ satisfying the threshold conditions. We get
            \[
            \Y_{\alpha\beta}^{-\infty}
            =
            \bigcup_{s,r_{\alpha\beta}}\Y_{s,r_{\alpha\beta}}
            =
            \bigcup_{r>\frac{1}{2}}\{f\in C^{-\infty}(M)\ \mid\ \wf^{\infty, r}(f)\cap R^{-\alpha\beta}=\varnothing\},
            \]
            \[
            \X_{\alpha\beta}^{-\infty}
            =
            \bigcup_{s,r_{\alpha\beta}}\X_{s,r_{\alpha\beta}}
            =
            \bigcup_{r>-\frac{1}{2}}\{u\in C^{-\infty}(M)\ \mid\ \big(\wf^{\infty,r}(u) \cup \wf^{\infty,r+1}(Pu)\big) \cap R^{-\alpha\beta}=\varnothing\},
            \]
        and $P$ restricts to invertible linear maps $P_{\alpha\beta}:\X_{\alpha\beta}^{-\infty}\to\Y_{\alpha\beta}^{-\infty}$. Since these maps are just extensions of the previously defined $P_{\alpha\beta}$ to a larger domain, we will use the same symbols as before for $P_{\alpha\beta}$ and the inverses $P_{\alpha\beta}^{-1}$.

        We can extend $G_{+\pm}=P_{+\pm}^{-1}-P_{-\mp}^{-1}$ to maps $(\Y_{+\pm}^{-\infty}\cap\Y_{-\mp}^{-\infty})\to (\X_{+\pm}^{-\infty}+\X_{-\mp}^{-\infty})$. Here we have
            \[
            \Y^{-\infty}
            :=
            \Y_{\alpha\beta}^{-\infty}\cap\Y_{-\alpha\beta}^{-\infty}
            =
            \bigcup_{r>\frac{1}{2}}\{f\in C^{-\infty}(M)\ \mid\ \wf^{\infty,r}(f)\cap R=\varnothing\},
            \]
            \[
            \X^{-\infty}
            :=
            \X_{\alpha\beta}^{-\infty}\cap\X_{-\alpha\beta}^{-\infty}
            =
            \bigcup_{r>-\frac{1}{2}}\{u\in C^{-\infty}(M)\ \mid\ \big(\wf^{\infty,r}(u) \cup \wf^{\infty,r+1}(Pu)\big) \cap R=\varnothing\},
            \]
            \[
            \X_{\alpha\beta}^{-\infty}+\X_{-\alpha\beta}^{-\infty}=C^{-\infty}(M).
            \]

		Using arguments analogous to the Schwartz case, we find that $\ker G=P\X^{-\infty}$ and $\ran G=\sol'$, where we define
		\[\sol'=\{u\in\X_{\alpha\beta}^{-\infty}+\X_{-\alpha\beta}^{-\infty}\ \mid\ Pu=0\}=\{u\in C^{-\infty}(M)\ \mid\ Pu=0\}.\]
        By analogy with the Schwartz case, let us also define for any $f\in\Y^{-\infty}$ the space
        \[\solf'=\{u\in\X_{\alpha\beta}^{-\infty}+\X_{-\alpha\beta}^{-\infty}\ \mid\ Pu=f\}=\{u\in C^{-\infty}(M)\ \mid\ Pu=f\}\]
        and $\soly'=\bigcup_{f\in\Y^{-\infty}}\solf'$.    

        We topologize $\Y^{-\infty}$ and $\soly'$ in the following way. Let $r_m$ be the variable orders defined in Section~\ref{sec:cont}. Then $\Y^{-\infty}=\bigcup_{m\in\N}\hsc{-m,-r_m}$ and $\soly'=\bigcup_{m\in\N}S_m$, where
        \[S_m=\{u\in \hsc{-m,-m}\ \mid\ Pu\in \hsc{-m,-r_m}\}\]
        with the Hilbert space norms $\lVert u\rVert _{S_m}^2=\lVert u\rVert _{-m,-m}^2+\lVert Pu\rVert _{-m,-r_m}^2$. We consider $\Y^{-\infty}$ and $\soly'$ to be equipped with the corresponding strict inductive limit topologies.

    \begin{prop}
        $P_{\alpha\beta}^{-1}:\Y^{-\infty}\to\soly'$ and $G_{+\pm}:\Y^{-\infty}\to \sol'$ are continuous with respect to these topologies.
    \end{prop}
    \begin{proof}
        It is enough to prove that for every $m\in\N$, there exists $m'\in\N$ such that $P_{\alpha\beta}^{-1}$ maps $\hsc{-m,-r_m}\to S_{m'}$ continuously. For every $m$ and any $\alpha,\beta\in\{+,-\}$, there exists $r_{\alpha\beta}<-r_m-1$ satisfying the $(\alpha,\beta)$-threshold conditions, so $\hsc{-m,-r_m}\subset\hsc{-m,r_{\alpha\beta}+1} =\Y_{-m+1,r_{\alpha\beta}}$ with continuous inclusion. Next, $P_{\alpha\beta}^{-1}:\Y_{-m+1,r_{\alpha\beta}}\to \X_{-m+1,r_{\alpha\beta}}$ is continuous. Finally, $\X_{-m+1,r_{\alpha\beta}}\subset S_{m'}$ with continuous inclusion for any $m'>\max(m-1,\sup (-r_{\alpha\beta}))$ such that $r_{m'}>-r_{\alpha\beta}-1$ (which is always satisfied for large enough $m'$). This proves continuity of $P_{\alpha\beta}^{-1}$, and the continuity of $G_{+\pm}$ follows since $G_{+\pm}=P_{+\pm}^{-1}-P_{-\mp}^{-1}$.
    \end{proof}

        Finally, the following proposition formalizes the notion that the distinguished inverses only propagate singularities in a certain direction along the Hamilton flow. This wavefront mapping property is analogous to that of the distinguished parametrices discussed in Section~\ref{sec:QFT}. A consequence is that except for possible creation of new singularities at the radial set, $G_{+\pm}$ propagates singularities only along bicharacteristics, even when propagation out of radial sets is considered.
        
        For any $U\subset\Sigma$, define $\lc(U)$ to be the union of bicharacteristics whose closures intersect $U$ and the limiting points of those bicharacteristics in $R$. We also define $\lc_{\alpha\beta}(U)$ to be the subset of $\lc(U)$ consisting of points which lie between a point of $U$ and a point of $R^{\alpha\beta}$ on some bicharacteristic.
        \begin{prop}
        \label{thm:PoS-G}
        Fix any $r\in\R$.
            \begin{enumerate}
                \item 
                    For any $f\in \Y^{-\infty}_{\alpha\beta}$, we have $\wf^{\infty,r}(P_{\alpha\beta}^{-1}f)\cap(\Sigma\backslash R^{\alpha\beta})\subset \lc_{\alpha\beta}(\wf^{\infty,r+1}(f)\cap\Sigma)$.
                \item 
                    For any $f\in\Y^{-\infty}$, we have $\wf^{\infty,r}(G_{+\pm}f)\backslash R \subset\lc(\wf^{\infty,r+1}(f)\cap\Sigma)$.
            \end{enumerate}
            The analogous results also hold for the absolute scattering wavefront set.
        \end{prop}
        \begin{proof}{\ }
            \begin{enumerate}
                \item 
                    Let $\zeta\in\Sigma\backslash R^{\alpha\beta}$, and let $\gamma$ denote the bicharacteristic through $\zeta$ (which consists only of $\zeta$ if $\zeta\in R^{-\alpha\beta}$). $\zeta\notin\lc_{\alpha\beta}(\wf^{\infty,r+1}(f)\cap\Sigma)$ means that the limiting point $\zeta_-\in R^{-\alpha\beta}$ of $\gamma$ is not in $\wf^{\infty,r+1}(f)$ and the segment of $\gamma$ between $\zeta_-$ and $\zeta$ does not intersect $\wf^{\infty,r+1}(f)$. Then there exists a neighborhood $U\subset\partial\tsccc$ of this segment (including both endpoints) such that $\wf^{\infty,r+1}(f)\cap U=\varnothing$.

                    Choose $s\in\R$ and $r_{\alpha\beta}$ satisfying the $(\alpha,\beta)$ threshold conditions such that $f\in \Y_{s,r_{\alpha\beta}}$. If $r>-\frac{1}{2}$, we can always choose $r'\in C^{\infty}(\tsccc)$ satisfying the $(\alpha,\beta)$ threshold conditions such that $r'\leqslant r_{\alpha\beta}$ outside of $U$, $r'\leqslant r$ in $U$, and $r'= r$ on a neighborhood of the segment of $\gamma$. Then $f\in \Y_{s,r'}$, and therefore $P_{\alpha\beta}^{-1}f\in\X_{s,r'}\subset \hsc{s,r'}$. This implies $\zeta\notin \wf^{\infty,r}(P_{\alpha\beta}^{-1}f)$.

                    If $r\leqslant -\frac{1}{2}$, we cannot find $r'$ with these properties that would satisfy the threshold condition near $\zeta_-$. However, since $P_{\alpha\beta}^{-1}f\in\X_{s,r_{\alpha\beta}}$, there automatically exists a neighborhood $U'$ of $R^{-\alpha\beta}$ which is disjoint from $\wf^{\infty,-\frac{1}{2}}(P_{\alpha\beta}^{-1}f)$, and if $\zeta \notin U'$, we can still conclude that $\zeta\notin\wf^{\infty,r}(P_{\alpha\beta}^{-1}f)$ by repeating the previous argument using $U\backslash U'$ in place of $U$.
                \item 
                    Since $PG_{+\pm}f=0$, we have $\wf(G_{+\pm}f)\subset\Sigma$ by microlocal elliptic regularity. In $\Sigma\backslash R$, the statement follows from part 1 since $G_{+\pm}=P_{+\pm}^{-1}-P_{-\mp}^{-1}$.
            \end{enumerate}
            The proof of the infinite-regularity version is similar, taking $r'$ to be arbitrarily large rather than equal to the desired order on the bicharacteristic segment.
        \end{proof}

\subsection{Distributional boundary data and generalized Poisson operators}

    To define the two-point functions for distributions, we need to first generalize the boundary data map and Poisson operators. We saw in Proposition~\ref{thm:cont-Schwartz} that the Poisson operators are in fact continuous with respect to distributional topologies, so they can be extended to continuous maps $C^{-\infty}(S)\to C^{-\infty}(M)$. We can also proceed in the opposite direction, directly extending the boundary data maps with the help of the boundary pairing formula.

    Consider $u\in \sols$ and $g\in \dot{C}^{\infty}(M)$. The boundary pairing formula tells us that
    \begin{equation}
    \langle \rho^{\pm}u, \rho^{\pm}P_{\pm\pm}^{-1}g\rangle_{L^2_{\h}}
    =
    \mp\frac{i}{2\lambda}\Big(\langle u,g\rangle_{L^2_{\g}}-\langle Pu, P_{\pm\pm}^{-1}g\rangle_{L^2_{\g}}\Big).
    \label{eq:bd-distr-1}
    \end{equation}
    We would like to re-interpret this formula as the definition of $\rho^{\pm}u$ for more general distributions $u$. For this, we need to describe the correspondence between $g\in\dot{C}^{\infty}(M)$ and $\rho^{\pm}P_{\pm\pm}^{-1}g\in C^{\infty}(S)$.

    The inverses $P_{\pm\pm}^{-1}:\Y_{\pm\pm}\to\X_{\pm\pm}$ restrict to bijective maps $\dot{C}^{\infty}(M)\to P_{\pm\pm}^{-1}\dot{C}^{\infty}(M)$. As we saw in Theorem \ref{thm:bd}, $\rho^{\pm}:P_{\pm\pm}^{-1}\dot{C}^{\infty}(M)\to C^{\infty}(S)$ is surjective with kernel $\dot{C}^{\infty}(M)$, so we can define a bijection $\rho^{\pm}_{\infty}P_{\pm\pm}^{-1}:\frac{\dot{C}^{\infty}(M)}{P\dot{C}^{\infty}(M)}\to C^{\infty}(S)$, where $\rho^{\pm}_{\infty}$ is defined in the discussion after Theorem \ref{thm:bd}. Note that for any $g\in P\dot{C}^{\infty}(M)$ and $u\in C^{-\infty}(M)$, we can integrate by parts on the right-hand side of Eq.~(\ref{eq:bd-distr-1}) to get $\langle u,g\rangle-\langle Pu, P_{\pm\pm}^{-1}g\rangle =0$. Therefore, the nontrivial kernel of the map $g\mapsto \rho^{\pm}P_{\pm\pm}^{-1}g$ is not an obstacle to using Eq.~(\ref{eq:bd-distr-1}) to extend $\rho^{\pm}$ to distributions.

    Recall $\mathcal{U}_{\infty}^{\pm}=(\rho^{\pm}_{\infty})^{-1}$. Then, based on the discussion above, we can write 
    \begin{equation}
    \overline{\rho^{\pm}u[\overline{a}]}
    =
    \langle \rho^{\pm}u,a\rangle_{L^2_{\h}}
    =
    \mp\frac{i}{2\lambda}\Big(\langle u,Pv^{\pm}\rangle_{L^2_{\g}}-\langle Pu,v^{\pm}\rangle_{L^2_{\g}}\Big)
    \label{eq:bd-distr-2}
    \end{equation}
    for any $u\in \sols$, $a\in C^{\infty}(S)$, and where $v^{\pm}$ is any representative of the class $\mathcal{U}_{\infty}^{\pm}a\in \frac{P_{\pm\pm}^{-1}\dot{C}^{\infty}(M)}{\dot{C}^{\infty}(M)}$. On the left-hand side we interpret $\rho^{\pm}u\in C^{\infty}(S)$ as an element of $C^{-\infty}(S)$.

    For any $u\in\soly'$, we have $Pu\in \hsc{-m,-r_m}$ for some $m$. On the other hand, for any $a\in C^{\infty}(S)$ we have $v^{\pm}\in\sols\subset \hsc{m,r_m}$ for any representative $v^{\pm}\in\mathcal{U}_{\infty}^{\pm}a$. Therefore, the right-hand side of Eq.~(\ref{eq:bd-distr-2}) is well-defined for any $u\in \soly'$. We then use Eq.~(\ref{eq:bd-distr-2}) to define $\rho^{\pm}u$ as a linear functional on $C^{\infty}(S)$ for any such $u$.

    \begin{theorem}{\textbf{(Extension of boundary data map and Poisson operators).}}
    \label{thm:ext}
        \begin{enumerate}
            \item 
                For any $u\in \soly'$, we have $\rho^{\pm}u\in C^{-\infty}(S)$.
            \item
                $\rho^{\pm}:\soly'\to C^{-\infty}(S)$ are continuous extensions (with respect to the topology on $\soly'$ defined in Section~\ref{sec:ext-G}) of $\rho^{\pm}:\sols\to C^{\infty}(S)$.
            \item
                For any fixed $f\in\Y^{-\infty}$, the maps $\rho^{\pm}:\solf'\to C^{-\infty}(S)$ are bijective.
            \item 
                The inverses $\mathcal{U}_f^{\pm}:C^{-\infty}(S)\to \solf'$ are continuous extensions (with respect to distributional topologies) of $\mathcal{U}_f^{\pm}:C^{\infty}(S)\to\solf$.
        \end{enumerate}
    \end{theorem}
    \begin{proof}{\ }
        We consider the case of $\rho^+,\mathcal{U}_f^+$. The case of $\rho^-,\mathcal{U}_f^-$ is analogous. We denote the distributional seminorms $\lVert a\rVert _{S,b}=|\langle a,b\rangle|$ for $a\in C^{-\infty}(S), b\in C^{\infty}(S)$.
        \begin{enumerate}
            \item 
                We need to prove that $\rho^+u$ is continuous as a linear functional on $C^{\infty}(S)$. Consider $a\in C^{\infty}(S)$.

                Let $v=\mathcal{U}_0^+a$, and take any $A\in\psco$ such that $\wf'(A)\cap R^-=\varnothing$ and $\wf'(I-A)\cap R^+=\varnothing$. Then $\wf(Av)\subset R$ and $P(Av)=[P,A]v \in \dot{C}^{\infty}(M)$, so $Av\in\sols$. Since $\wf(Av)\cap R^-=\varnothing$, we must have $\rho^-(Av)=0$. On the other hand, since $\wf((I-A)v)\cap R^+=\varnothing$, we must have $\rho^+((I-A)v)=0$, so $\rho^+(Av)=\rho^+(v)=a$. Thus, $Av$ is a representative of $\mathcal{U}_{\infty}^+a$.

                This means that for any choice of $A$ as above, we can define $\rho^+u$ by
                \begin{equation}				                    
                \langle \rho^+u, a\rangle_{L^2_{\h}}
                =
                -\frac{i}{2\lambda}\Big(\langle u, PA\mathcal{U}_0^+a\rangle_{L^2_{\g}}-\langle Pu,A\mathcal{U}_0^+a\rangle_{L^2_{\g}}\Big).
                \label{eq:bd-distr-3}					
                \end{equation}
                We now show that both terms on the right-hand side are continuous as a function of $a$.

                Consider the first term. It has already been established that $PA\mathcal{U}_0^+$ maps $C^{\infty}(S)\to\dot{C}^{\infty}(M)$. Continuity can be seen using the closed graph theorem (for linear maps between Fr\'echet spaces). Consider a sequence $a_n\in C^{\infty}(S)$ such that $a_n\to a$ in $C^{\infty}(S)$ and $PA\mathcal{U}_0^+a_n\to f$ in $\dot{C}^{\infty}(M)$. Since $\mathcal{U}_0^+:C^{\infty}(S)\to C^{-\infty}(M)$ and $PA:C^{-\infty}(M)\to C^{-\infty}(M)$ are continuous, we have $PA\mathcal{U}_0^+a_n\to PA\mathcal{U}_0^+a$ in $C^{-\infty}(M)$. At the same time, since the Schwartz topology of $\dot{C}^{\infty}(M)$ is stronger that the distributional one, we must have $PA\mathcal{U}_0^+a_n\to f$ in $C^{-\infty}(M)$. Therefore, $PA\mathcal{U}_0^+a=f$. This proves that the graph of $PA\mathcal{U}_0^+$ is closed, so we conclude continuity. Finally, since $u\in C^{-\infty}(M)$, the map $\langle u,\bullet\rangle:\dot{C}^{\infty}(M)\to\C$ is continuous, so the first term of Eq.~(\ref{eq:bd-distr-3}) is continuous in $a$.

                A similar argument applies to the second term. For any given $u$, we have $Pu\in \hsc{-m,-r_m}$ for some $m\in\N$. $A\mathcal{U}_0^+$ maps $C^{\infty}(S)\to \sols\subset \hsc{m,r_m}$. Continuity of $A\mathcal{U}_0^+\to\hsc{m,r_m}$ follows by the closed graph theorem in the same way. Finally, the Sobolev-space pairing $\langle Pu,\bullet\rangle:\hsc{m,r_m}\to\C$ is continuous.

            \item
                The new $\rho^+$ is an extension of the old one by construction, so we only need to show that $\rho^+ : \soly' \to C^{-\infty}(S)$ is continuous. It is enough to show continuity as a map $S_m\to C^{-\infty}(S)$ for any $m\in\N$. For any $u\in S_m$, $a\in C^{\infty}(S)$, and $v^+\in\mathcal{U}_{\infty}^+a$, we have directly from Eq.~(\ref{eq:bd-distr-2}):
                \[
                \lVert \rho^+u\rVert _{S,a}
                \leqslant 
                \frac{1}{2\lambda}\Big( |\langle u, Pv^+\rangle_{L^2_{\g}}|+|\langle Pu,v^+\rangle_{L^2_{\g}}|\Big)
                \leqslant
                \]
                \[
                \leqslant
                \frac{1}{2\lambda}\Big( \lVert u\rVert _{-m,-m}\lVert Pv^+\rVert _{m,m} +\lVert Pu\rVert _{-m,-r_m}\lVert v^+\rVert _{m,r_m}\Big)
                \leqslant
                C\lVert u\rVert _{S_m}.
                \]
                                
            \item
                Assume $u,v\in\solf'$ and $\rho^+u=\rho^+v$. Then $w:=u-v\in\sol'$ and $\rho^+w=0$. For any $g\in \dot{C}^{\infty}(M)$,
                \[
                \langle w,g\rangle_{L^2_{\g}}=\langle w,g\rangle_{L^2_{\g}}-\langle Pw,P_{++}^{-1}g\rangle_{L^2_{\g}}=2i\lambda\langle\rho^+w,\rho^+P_{++}^{-1}g\rangle_{L^2_{\h}}=0,\]
                i.e. $w=0$ as a distribution, so $u=v$. This proves injectivity.

                To prove surjectivity, consider an arbitrary $a\in C^{-\infty}(S)$. From the definition, we know that if there exists $u\in \solf'$ such that $\rho^+u=a$, then for any $g\in \dot{C}^{\infty}$ we have
                \[\langle u, g\rangle_{L^2_{\g}}=\langle f,P_{++}^{-1}g\rangle_{L^2_{\g}}+2i\lambda\langle a,\rho^+P_{++}^{-1}g\rangle_{L^2_{\h}}.\]
                
                Consider the linear functional $u$ on $\dot{C}^{\infty}(M)$ defined by this formula.      
                Then if we manage to prove $u\in C^{-\infty}(M)$ and $Pu=f$, we will have $\rho^+u=a$ by construction. 
                
                The statement that $u\in C^{-\infty}(M)$ is the continuity of this expression as a function of $g$, which follows from the continuity of $P_{++}^{-1}:\dot{C}^{\infty}\to \sols\subset \hsc{m,r_m}$ for any $m$, $\rho^+:\sols\to C^{\infty}(S)$, $\langle a,\bullet\rangle:C^{\infty}(S)\to\C$, and $\langle f,\bullet\rangle:\hsc{m,r_m}\to\C$ for $f\in\hsc{-m,-r_m}$.
                
                To see that $Pu=f$, consider any $v\in \dot{C}^{\infty}(M)$. Then $\rho^+v=\rho^-v=0$, so
                \[
                \langle Pu, v\rangle_{L^2_{\g}}
                =
                \langle u, Pv \rangle_{L^2_{\g}}
                =
                \langle f,P_{++}^{-1}Pv\rangle_{L^2_{\g}}+2i\lambda \langle a, \rho^+v\rangle_{L^2_{\h}} 
                =
                \langle f, v\rangle_{L^2_{\g}}.
                \]
                So indeed $Pu=f$ as a distribution. This proves surjectivity of $\rho^+:\solf'\to C^{-\infty}(S)$.

            \item
                From the previous part, we see that $\mathcal{U}^+_f:C^{-\infty}(S)\to \solf'$ is defined by
                \begin{equation}
                \langle \mathcal{U}_f^+a, g\rangle_{L^2_{\g}} = \langle f,P_{++}^{-1}g\rangle_{L^2_{\g}}+2i\lambda \langle a, \rho^+P_{++}^{-1}g\rangle_{L^2_{\h}}
                \label{eq:poisson-dist}
                \end{equation}
                for any $g\in\dot{C}^{\infty}(M)$. It is enough to show continuity of $\mathcal{U}_0^+$, defined by 
                \[\langle \mathcal{U}_0^+a, g\rangle_{L^2_{\g}} = 2i\lambda \langle a, \rho^+P_{++}^{-1}g\rangle_{L^2_{\h}},\]
                and we already saw in Proposition \ref{thm:cont-Schwartz} that this is continuous $C^{-\infty}(S)\to C^{-\infty}(M)$.
        \end{enumerate}
    \end{proof}

    We will denote the extended boundary data maps and Poisson operators by the same symbols as before, $\rho^{\pm}:\soly'\to C^{-\infty}(S)$ and $\mathcal{U}_f^{\pm}:C^{-\infty}(S)\to\solf'$, and occasionally $\rho^{\pm\pm}$ and $\mathcal{U}_f^{\pm\pm}$.

    The extended causal boundary data maps are $\rho^{\pm\mp}=\pi_1\circ\rho^{\pm}+\pi_2\circ\rho^{\mp}$, where now $\pi_1,\pi_2:C^{-\infty}(S_1)\oplus C^{-\infty}(S_2)\to C^{-\infty}(S_1)\oplus C^{-\infty}(S_2)$ are projections on the space of distributions. The explicit expression defining $\rho^{\pm\mp}u$ as a distribution is
    \begin{equation}
        \langle\rho^{\pm\mp}u,a\rangle_{L^2_{\h}}=\mp\frac{i}{2\lambda}\Big(\langle u, P(v_1^{\pm}-v_2^{\mp})\rangle_{L^2_{\g}}-\langle Pu, v_1^{\pm}-v_2^{\mp}\rangle_{L^2_{\g}}\Big)
        \label{eq:bd-distr-causal-1}
    \end{equation}
    for any $a\in C^{\infty}(S)$ and any representatives $v_j^{\pm}\in\mathcal{U}^{\pm}_{\infty}\pi_ja$. The analogue of Eq.~(\ref{eq:bd-distr-1}) is
    \begin{equation}
        \langle\rho^{\pm\mp}u,(\pi_1-\pi_2)\rho^{\pm\mp}P_{\pm\mp}^{-1}g\rangle_{L^2_{\h}} = \mp\frac{i}{2\lambda}\Big(\langle u,g\rangle_{L^2_{\g}}-\langle Pu,P_{\pm\mp}^{-1}g\rangle_{L^2_{\g}}\Big)
        \label{eq:bd-distr-causal-2}
    \end{equation}
    for any $g\in\dot{C}^{\infty}(M)$. The fact that $\rho^{\pm\mp}$ map $\sol'\to C^{-\infty}(S)$ continuously follows directly from the definition and the corresponding properties of $\rho^{\pm}$. They also restrict to invertible maps $\solf'\to C^{-\infty}(S)$ with continuous inverses $\mathcal{U}_0^{\pm\mp}:C^{-\infty}(S)\to\sol'$, the proofs being analogous to the Feynman case (Theorem~\ref{thm:ext}, parts 3 and 4) but using Eq.~(\ref{eq:bd-distr-causal-2}) instead of Eq.~(\ref{eq:bd-distr-1}).

    The following set of results connecting the wavefront set of solutions to properties of their boundary data is a weak extension of Lemma \ref{thm:wf-bd} to the distributional setting. Much finer results can be obtained using the explicit characterization of the Poisson operator by Melrose and Zworski \cite{M-Z}, see Appendix A of \cite{Vasy-3BS}.

   \begin{prop}\ 
   \label{thm:wf-transl}
       \begin{enumerate}
            \item 
                Let $u\in \soly'$, $q\in S$, and $\zeta_{\pm}=\pi^{-1}(q)\cap R^{\pm}$. If $\zeta_{\pm}\notin \wf(u)$, then $q\notin\supp(\rho^{\pm}u)$.
           \item 
                Let $u\in\sol'$, $j=1$ or $2$. Then $\wf(u)\cap\Sigma_j\subset R$ if and only if $\pi_j\rho^{+}u,\pi_j\rho^{-}u\in C^{\infty}(S)$.
            \item 
                Let $u\in\sol'$, $j=1$ or $2$. Then $\wf(u)\cap\Sigma_j\subset R^{\mp}$ if and only if $\pi_j\rho^{\pm}u=0$. 
       \end{enumerate}
   \end{prop}
   \begin{proof}\ 
       \begin{enumerate}
            \item 
                If $\zeta_{\pm}\notin \wf(u)$, then there exists a neighborhood $U\subset R^{\pm}$ of $\zeta_{\pm}$ such that $\wf(u)\cap U=\varnothing$. Define $V=\pi(U)$, which is a neighborhood of $q$ in $S$. Choose $\varphi\in C_{\mathrm{c}}^{\infty}(S)$ such that $\supp(\varphi)\subset V$. Let $v\in\sols$ be a representative of $\mathcal{U}^{\pm}_{\infty}\varphi$. Then $\wf(v)=\pi^{-1}(\supp(\varphi))\cap R^{\pm}\subset U$, so $\wf(u)\cap \wf(v)=\varnothing$. Then by Lemma~\ref{thm:int-by-parts}, we can integrate by parts to get
                \[
                \langle \rho^{\pm}u,\varphi\rangle_{L^2_{\h}}
                =
                \mp\frac{i}{2\lambda}\Big(\langle u, Pv\rangle_{L^2_{\g}}-\langle Pu,v\rangle_{L^2_{\g}}\Big)
                =
                0.
                \]
             Thus, there exists a neighborhood of $q$ such that pairing $\rho^{\pm}u$ with any smooth function supported in the neighborhood gives zero, which means $q\notin\supp(\rho^{\pm}u)$.

            \item 
                Choose $A\in\psco$ such that $\wf'(A)\cap ((\Sigma\backslash\Sigma_j)\cup R^{\mp})=\varnothing$ and $\wf'(I-A)\cap (R^{\pm}\cap\Sigma_j)=\varnothing$. One can check that for any $a\in C^{\infty}(S)$, $A\mathcal{U}_0^{\pm}(\pi_ja)$ is a representative of $\mathcal{U}_{\infty}^{\pm}(\pi_ja)$. Then for any $u\in\sol'$ and $a\in C^{\infty}(S)$, we can write
                \[
                \langle \pi_j\rho^{\pm}u,a\rangle_{L^2_{\h}}
                =
                \langle\rho^{\pm}u,\pi_ja\rangle_{L^2_{\h}}
                =
                \mp\frac{i}{2\lambda}\langle u,[P,A]\mathcal{U}_0^{\pm}(\pi_ja)\rangle_{L^2_{\g}}
                =
                \]
                \[
                =
                \pm\frac{i}{2\lambda}\langle [P,A^*]u,\mathcal{U}_0^{\pm}(\pi_ja)\rangle_{L^2_{\g}}.
                \]
                In the last step we used Lemma~\ref{thm:int-by-parts} to integrate by parts since $\wf'([P,A])\cap R=\varnothing$ and $\wf(\mathcal{U}_0^{\pm}(\pi_ja))\subset R$. Now if $\wf(u)\cap \Sigma_j\subset R$, then $[P,A^*]u\in\dot{C}^{\infty}(M)$, so we can use Eq.~(\ref{eq:poisson-dist}) to write
                \[
                \pm\frac{i}{2\lambda}\langle [P,A^*]u,\mathcal{U}_0^{\pm}(\pi_ja)\rangle
                =
                \langle \rho^{\pm}P_{\pm}^{-1}[P,A^*]u,\pi_ja\rangle
                =
                \langle \pi_j\rho^{\pm}P_{\pm}^{-1}[P,A^*]u,a\rangle.
                \]
                Thus, $\pi_j\rho^{\pm}u=\pi_j\rho^{\pm}P_{\pm}^{-1}[P,A^*]u\in C^{\infty}(S)$.

                For the other direction, take $\varphi\in C^{\infty}(M)$ such that $\varphi=1$ on a neighborhood of $S_1$ and $\varphi=0$ on a neighborhood of $S_2$. Let $u_1=\varphi u$ and $u_2=(1-\varphi)u$. Since $Pu=0$ and $\varphi$ is constant in a neighborhood of either component of $S$, $Pu_1=Pu_2=0$ in a neighborhood of $S$ as well, so in particular $Pu_1,Pu_2\in \dot{C}^{\infty}(M)$ and $\rho^{\pm}u_1,\rho^{\pm}u_2$ are well-defined and satisfy $\rho^{\pm}u=\rho^{\pm}u_1+\rho^{\pm}u_2$.
                
                By construction, $\wf(u_1)=\wf(u)\cap\Sigma_1$ and $\wf(u_2)=\wf(u)\cap\Sigma_2$. Then part 1 of the proposition implies $\supp(\rho^{\pm}u_1)\subset S_1$, $\supp(\rho^{\pm}u_2)\subset S_2$. Then we must have $\rho^{\pm}u_j=\pi_j\rho^{\pm}u$. 
                
                Thus, $\pi_j\rho^{\pm}u\in C^{\infty}(S)$ is equivalent to $\rho^{\pm}u_j\in C^{\infty}(S)$. Since we know that $Pu_j\in\dot{C}^{\infty}(M)$, this implies $u_j\in \sols$. Then by Eq.~(\ref{eq:wf-bd}) we get $\wf(u_j)\subset R$, so if $\pi_j\rho^+u$ and $\pi_j\rho^-u$ are both in $C^{\infty}(S)$, then $\wf(u)\cap\Sigma_j\subset R$.
       
           \item 
                $\wf(u)\cap\Sigma_j\subset R^{\mp}$ implies $\pi_j\rho^{\pm}u=0$ immediately by part 1. For the other direction, assume that $\pi_j\rho^{\pm}u=0$. Defining $u_1$ and $u_2$ as above, by the same argument we find that $u_j\in \sols$ and $\rho^{\pm}u_j=0$. Then Lemma~\ref{thm:wf-bd} implies $\wf(u)\cap\Sigma_j=\wf(u_j)\subset R^{\mp}$.
       \end{enumerate}
   \end{proof}

\subsection{Two-point functions and their wavefront mapping properties}
    Now that we have defined $\rho^{\alpha\beta}:\sol'\to C^{-\infty}(M)$, we can extend the two-point functions to operators $\Lambda_1^{\alpha\beta},\Lambda_2^{\alpha\beta}: \Y^{-\infty}\to \sol'$, still using the expressions 
    \begin{equation}
        \Lambda_1^{\pm\pm} f=+i\mathcal{U}_0^{\pm\pm}\pi_1\rho^{\pm\pm}G_{++}f,
        \hskip 50pt
        \Lambda_2^{\pm\pm} f=+i\mathcal{U}_0^{\pm\pm}\pi_2\rho^{\pm\pm}G_{++}f,
        \label{eq:2pt-feynman-dist}
    \end{equation}
    \begin{equation}
        \Lambda_1^{\pm\mp} f=+i\mathcal{U}_0^{\pm\mp}\pi_1\rho^{\pm\mp}G_{+-}f,
        \hskip 50pt
        \Lambda_2^{\pm\mp} f=-i\mathcal{U}_0^{\pm\mp}\pi_2\rho^{\pm\mp}G_{+-}f.
        \label{eq:2pt-causal-dist}
    \end{equation}
    The extensions are continuous as maps $\Y^{-\infty}\to C^{-\infty}(M)$, which follows from the continuity of every map in the diagram
            \[\begin{tikzcd}
   {\Y^{-\infty}} & {\sol'} & {C^{-\infty}(S)} & {C^{-\infty}(S)} & {C^{-\infty}(M).}
   	\arrow["G_{+\pm}", from=1-1, to=1-2]
   	\arrow["\rho^{\alpha\beta}", from=1-2, to=1-3]
    \arrow["\pi_j", from=1-3, to=1-4]
    \arrow["\mathcal{U}_0^{\alpha\beta}", from=1-4, to=1-5]
   \end{tikzcd}\]
   Note that the identity $\rho^{-\alpha\beta}P_{\alpha\beta}^{-1}f=0$ (which extends from $f\in\dot{C}^{\infty}(M)$ to $f\in\Y^{-\infty}$ by continuity due to the density of $\dot{C}^{\infty}(M)$ in every weighted Sobolev space) means that in fact in each of the expressions in Eqs.~(\ref{eq:2pt-feynman-dist})-(\ref{eq:2pt-causal-dist}) only one inverse contributes out of the two which compose $G_{+\pm}$.

    Propositions~\ref{thm:PoS-G} and~\ref{thm:wf-transl} provide information about the wavefront mapping properties of the two-point functions. First, since $\supp(\rho^{\alpha\beta}(\Lambda^{\alpha\beta}_jf))\subset S_j$, we have
    \[\wf(\Lambda^{\alpha\beta}_jf)\subset \Sigma_j\cup R^{\alpha\beta}.\]
    Furthermore, since $\pi_j\rho^{\alpha\beta}(\Lambda^{\alpha\beta}_jf\pm iP_{\alpha\beta}^{-1}f)=0$ (the choice of $\pm$ depending on the particular $\alpha,\beta,j$), in $\Sigma_j$ we have $\wf(\Lambda^{\alpha\beta}_jf\pm iP_{\alpha\beta}^{-1}f)\cap (\Sigma_j\backslash R^{\alpha\beta})=\varnothing$. In particular, this implies $\wf^{\infty,r}(\Lambda^{\alpha\beta}_jf)\cap (\Sigma_j\backslash R^{\alpha\beta}) = \wf^{\infty,r}(P_{\alpha\beta}^{-1}f)\cap (\Sigma_j\backslash R^{\alpha\beta})$ and therefore
    \[
    \wf^{\infty,r}(\Lambda^{\alpha\beta}_jf)\cap (\Sigma_j\backslash R^{\alpha\beta}) \subset \lc_{\alpha\beta}(\wf^{\infty,r+1}(f)\cap\Sigma_j)
    \]
    for any $r$, along with the analogous infinite-regularity statement. Put together, this becomes
    \begin{equation}
    \wf^{\infty,r}(\Lambda^{\alpha\beta}_jf)\backslash R^{\alpha\beta}\subset \lc_{\alpha\beta}(\wf^{\infty,r+1}(f)\cap \Sigma_j).
        \label{eq:Hadamard-sc}
    \end{equation}

    This is an analogue of Eq.~(\ref{eq:Hadamard-mapping}), the wavefront mapping property of Hadamard two-point functions. Unlike Eq.~(\ref{eq:Hadamard-mapping}), Eq.~(\ref{eq:Hadamard-sc}) provides information on behavior at infinity, except for the radial sets where singularities can appear. (On the other hand, since $P$ is elliptic over the interior of $M$, in our case all of the interesting behavior happens at infinity, which is not true for wave equations).

    We briefly comment on the reason why there is no simple analogue of the Hadamard condition Eq.~(\ref{eq:Hadamard-kernel}) on the kernels of the two-point functions in the setting of scattering theory. Operators on $M$ have kernels which are distributions on $M\times M$. When $M$ is a manifold without boundary, then so is $M\times M$, and consequently we can talk about wavefront sets of distributions on both $M$ and $M\times M$, and indeed the wavefront mapping properties of an operator are related to the wavefront set of its kernel. On the other hand, when $M$ is a manifold with boundary, $M\times M$ is instead a manifold with corners, which does not have a natural scattering calculus, so one cannot in general talk about the scattering wavefront set of the kernel. A condition like Eq.~(\ref{eq:Hadamard-kernel}) in this setting would therefore require the introduction of additional structure.

\subsection{Mock propagators}
\label{sec:mock-props}
    For the same reason, the definition of distinguished parametrices from Section~\ref{sec:QFT} in terms of kernels cannot be directly translated to the scattering setting. Nevertheless, as we have seen, the distinguished inverses $P_{\alpha\beta}^{-1}$ play a similar role in many ways. The most important properties that distinguish them are the wavefront mapping properties in Proposition~\ref{thm:PoS-G} and, on the level of boundary data, the identity $\rho^{-\alpha\beta}P_{\alpha\beta}^{-1}f=0$.

    Recall that in usual quantum field theories, given one pair of propagators and the two-point functions of a state, we can find the other pair of propagators using Eqs.~(\ref{eq:props-bose-c-to-f})-(\ref{eq:props-bose-f-to-c}), and if the state is a Hadamard state, then these propagators will be distinguished parametrices of the appropriate type. A version of this procedure can be carried out in our setting. If we consider $\mathrm{CCR}^{\mathrm{pol}}(\dot{C}^{\infty}(M),G_{+-})$, the inverses $P_{+-}^{-1}$ and $P_{-+}^{-1}$ play a role analogous to $G_{\mathrm{R}}$ and $G_{\mathrm{A}}$ respectively, and taking either of the pairs $(\Lambda_1^{+-},\Lambda_2^{+-})$ or $(\Lambda_1^{-+},\Lambda_2^{-+})$, we can define new operators that will be analogues of $G_{\mathrm{F}}$ and $G_{\bar{\mathrm{F}}}$. Similarly, if we consider $\mathrm{CAR}^{\mathrm{pol}}(\dot{C}^{\infty}(M),G_{++})$, the inverses $P_{++}^{-1}$ and $P_{--}^{-1}$ play a role analogous to $G_{\mathrm{F}}$ and $G_{\bar{\mathrm{F}}}$ respectively, and taking either $(\Lambda_1^{++},\Lambda_2^{++})$ or $(\Lambda_1^{--},\Lambda_2^{--})$, we can define analogues of $G_{\mathrm{A}}$ and $G_{\mathrm{R}}$. The pair of ``mock propagators'' thus defined shares many properties of the other pair of distinguished inverses.

    Thus, for $\alpha,\beta\in\{+,-\}$ and $j=1,2$, define operators $Q^{\alpha\beta}_j:\Y^{-\infty}\to \soly'$ by
    \begin{equation}
    Q^{++}_1=P_{+-}^{-1}-i\Lambda_2^{-+}=P_{-+}^{-1}-i\Lambda_1^{-+}
    \hskip 30pt
    Q^{++}_2=P_{+-}^{-1}-i\Lambda_2^{+-}=P_{-+}^{-1}-i\Lambda_1^{+-},
        \label{eq:mock-prop++}
    \end{equation}
    \begin{equation}
    Q^{--}_1=P_{+-}^{-1}+i\Lambda_1^{+-}=P_{-+}^{-1}+i\Lambda_2^{+-},
    \hskip 30pt
    Q^{--}_2=P_{+-}^{-1}+i\Lambda_1^{-+}=P_{-+}^{-1}+i\Lambda_2^{-+},
        \label{eq:mock-prop--}
    \end{equation}
    \begin{equation}
    Q^{+-}_1=P_{++}^{-1}+i\Lambda_2^{--}=P_{--}^{-1}-i\Lambda_1^{--},
    \hskip 30pt
    Q^{+-}_2=P_{++}^{-1}+i\Lambda_2^{++}=P_{--}^{-1}-i\Lambda_1^{++},
        \label{eq:mock-prop+-}
    \end{equation}
    \begin{equation}
    Q^{-+}_1=P_{++}^{-1}+i\Lambda_1^{++}=P_{--}^{-1}-i\Lambda_2^{++},
    \hskip 30pt
    Q^{-+}_2=P_{++}^{-1}+i\Lambda_1^{--}=P_{--}^{-1}-i\Lambda_2^{--}.
        \label{eq:mock-prop-+}
    \end{equation}
    By construction, $PQ^{\alpha\beta}_jf=f$ for any $f\in\Y^{-\infty}$, and the pairs $(Q^{\alpha\beta}_1,Q^{-\alpha\beta}_2)$ satisfy analogues of Eq.~(\ref{eq:props-bose}), namely
    \begin{equation}
    i(Q_1^{++}-Q_2^{--})=\Lambda_1^{-+}+\Lambda_2^{-+},
    \hskip 30pt
    i(Q_2^{++}-Q_1^{--})=\Lambda_1^{+-}+\Lambda_2^{+-},
        \label{eq:mock-G++}
    \end{equation}
    \begin{equation}
    i(Q_1^{+-}-Q_2^{-+})=\Lambda_1^{--}-\Lambda_2^{--},
    \hskip 30pt
    i(Q_2^{+-}-Q_1^{-+})=\Lambda_1^{++}-\Lambda_2^{++},
        \label{eq:mock-G+-}
    \end{equation}
    and
    \begin{equation}
    Q_1^{++}+Q_2^{--}=Q_2^{++}+Q_1^{--}=P_{+-}^{-1}+P_{-+}^{-1},
    \hskip 30pt
    Q_1^{+-}+Q_2^{-+}=Q_2^{+-}+Q_1^{-+}=P_{++}^{-1}+P_{--}^{-1}.
        \label{eq:mock-props-sum}
    \end{equation}
    Moreover, as we prove shortly, they propagate singularities in the same way as the corresponding distinguished inverses. However, on the level of boundary data, $Q^{\alpha\beta}_j$ only has the ``correct'' behavior on $S_j$; on the other component, the boundary data might not be zero like it is for $P_{\alpha\beta}^{-1}$, but it is smooth.
    
    \begin{table}[H]
        \centering
        \begin{tabular}{|c||c|c|c|c|}
        \hline
             & $\pi_1\rho^+$ & $\pi_1\rho^-$ & $\pi_2\rho^+$ & $\pi_2\rho^-$\\
             \hline\hline
            $P_{+-}^{-1}f$  & $a^{+-}_1$ & $0$ & $0$ & $a^{+-}_2$\\
            \hline
           $P_{-+}^{-1}f$  & $0$ & $-a^{-+}_1$ & $-a^{-+}_2$ & $0$\\
           \hline
           $-i\Lambda^{-+}_1f$  & $\pi_1\mathcal{S}^{+-}_{-+}(a^{-+}_1)$ & $a^{-+}_1$ & $0$ & $\pi_2\mathcal{S}^{+-}_{-+}(a^{-+}_1)$\\
           \hline
           $Q^{++}_1f$  & $\pi_1\mathcal{S}^{+-}_{-+}(a^{-+}_1)$ & $0$ & $-a^{-+}_2$ & $\pi_2\mathcal{S}^{+-}_{-+}(a^{-+}_1)$\\
           \hline
           $P_{++}^{-1}f$  & $a^{++}_1$ & $0$ & $a^{++}_2$ & $0$\\
           \hline
        \end{tabular}
        \caption{Components of boundary data of the images of $f\in\Y^{-\infty}$ under various operators, illustrating that the ``mock propagator'' $Q^{\alpha\beta}_j$ is constructed to match the boundary data of the global propagator $P_{\alpha\beta}^{-1}$ in one of the two pieces where it is zero, specifically the piece on $S_j$. The other such piece ($\pi_2\rho^-$ in this example) is smooth but typically nonzero. Here $\mathcal{S}^{-\alpha\beta}_{\alpha\beta}=\rho^{-\alpha\beta}\mathcal{U}_0^{\alpha\beta}$.}
        \label{tab:mock-prop}
    \end{table}

    \begin{prop}
    \label{thm:mock-props}
        For any $f\in \mathcal{Y}^{-\infty}$, we have $\rho^{-\alpha\beta}Q_j^{\alpha\beta}f\in C^{\infty}(S)$, $\pi_j\rho^{-\alpha\beta}Q^{\alpha\beta}_j f=0$, and $\wf((Q^{\alpha\beta}_j-P_{\alpha\beta}^{-1})f)\subset R^{\alpha\beta}$.
    \end{prop}
    \begin{proof} We consider $Q^{++}_1$; the other cases are all analogous. For $(\alpha,\beta)=(+,+)$ or $(+,-)$ and $j=1,2$, define $a^{\pm\alpha\beta}_j=\pi_j \rho^{\pm\alpha\beta}G_{\alpha\beta}f$. Then the boundary data components of images of $f$ under relevant operators, by properties of distinguished inverses and direct calculation from the definition, are as shown in Table~\ref{tab:mock-prop}. In particular, $\pi_1\rho^{-}Q^{++}_1=0$ is immediate. In the other component, $\pi_2\rho^+(-i\Lambda^{-+}_1f)=0$, so by Proposition~\ref{thm:wf-transl} we get $\wf(-i\Lambda^{-+}_1f)\cap\Sigma_2\subset R^-$ and therefore $\pi_2\rho^-(-i\Lambda^{-+}_1f)\in C^{\infty}(S)$. Thus, we get $\pi_2\rho^-(Q^{++}_1f)=\pi_2\rho^-(-i\Lambda^{-+}_1f)\in C^{\infty}(S)$ as well. This completes the proof of the boundary data statement.

    Since $\pi_1\rho^-$ of both terms is zero, we have $\pi_1\rho^-((Q_1^{++}-P_{++}^{-1})f)=0$. Then, since $(Q_1^{++}-P_{++}^{-1})f\in\sol'$, by Proposition~\ref{thm:wf-transl} we conclude that $\wf((Q_1^{++}-P_{++}^{-1})f)\cap \Sigma_1\subset R^+$. In the other component, on the other hand, we have $\wf(\Lambda_1^{-+}f)\cap(\Sigma_2\backslash R^+)=\varnothing$ by Eq.~(\ref{eq:Hadamard-sc}), so
    \[
    \wf((Q_1^{++}-P_{++}^{-1})f)\cap (\Sigma_2\backslash R^+)
    =
    \wf((P_{-+}^{-1}-P_{++}^{-1})f) \cap (\Sigma_2\backslash R^+)
    =
    \varnothing,
    \]
    where in the last step we used Proposition~\ref{thm:wf-transl} since $(P_{-+}^{-1}-P_{++}^{-1})f\in\sol'$ and satisfies $\pi_2\rho^-((P_{-+}^{-1}-P_{++}^{-1})f)=0$. Thus, $\wf((Q_1^{++}-P_{++}^{-1})f)\cap\Sigma_2\subset R^+$ as well.
    \end{proof}

    As we see, the $(\alpha,\beta)$ ``mock propagators'' thus constructed do not coincide with the $(\alpha,\beta)$ distinguished inverses, but their differences are completely regularizing away from $R^{\alpha\beta}$. In particular, while the inverses themselves do not satisfy equations like Eq.~(\ref{eq:mock-G++})-(\ref{eq:mock-props-sum}) exactly, the errors are completely regularizing away from $R$. Demanding that $\rho^{-\alpha\beta}P_{\alpha\beta}^{-1}f=0$ distinguishes the true inverse among the operators with the correct wavefront mapping properties. (The fact that this singles out a unique operator follows from the bijectivity of $\mathcal{U}_f^{-\alpha\beta}: C^{-\infty}(S)\to \solf$). This is similar to the situation with distinguished parametrices, which are unique modulo smoothing operators.

    The reason for the discrepancy can be understood as follows. Our two-point functions describe states based on part of the asymptotic behavior of the fields: at each end, a two-point function has either known incoming or known outgoing behavior, but not both (in the sense that if we express the boundary data in terms of $a^{\pm\alpha\beta}_j$, the other components will depend on the scattering matrix). Therefore, by combining an inverse and a two-point function, we can ``engineer'' the desired behavior at each end for the kind of data that the state refers to, but what happens to the other pieces of data is not controlled.
    
    This is similar to the situation with in/out vacuum states of QFT on asymptotically Minkowski spacetimes: we can construct a Feynman propagator for the ``in'' vacuum state, and if we only consider the ``in'' representation, it will have many of the properties of the Minkowski Feynman propagator. However, the ``in'' states are connected to the ``out'' states by a nontrivial S-matrix, so the Feynman propagator of the ``in'' vacuum will not have the desired properties from the perspective of the ``out'' representation. The statement that $\rho^{-\alpha\beta}P_{\alpha\beta}^{-1}f$ is smooth even though nonzero is a statement about the regularity of the S-matrix.

    A natural further question is: can a pair of ``mock propagators'' be used to define a Hermitian form on $\dot{C}^{\infty}(M)/P\dot{C}^{\infty}(M)$ similarly to the corresponding pair of inverses, and if so, how different are these Hermitian forms from each other? Eqs.~(\ref{eq:mock-G++})-(\ref{eq:mock-G+-}) along with the properties of the two-point functions imply that $i(Q^{\alpha\beta}_1-Q^{-\alpha\beta}_2)$ and $i(Q^{\alpha\beta}_2-Q^{-\alpha\beta}_1)$ indeed define nondegenerate Hermitian forms on $\dot{C}^{\infty}(M)/P\dot{C}^{\infty}(M)$ for either choice of $(\alpha,\beta)=(+,+)$ or $(+,-)$, and in the $(+,+)$ cases the forms are positive.

    Let $f,g\in\dot{C}^{\infty}(M)$, and define $a_j^{\pm\alpha\beta}=\pi_ja^{\pm\alpha\beta}$, $b_j^{\pm\alpha\beta}=\pi_j b^{\pm\alpha\beta}$, where $a^{\pm\alpha\beta},b^{\pm\alpha\beta}$ are as in Eq.~(\ref{eq:bd-notation}). Then using Eqs.~(\ref{eq:mock-G++})-(\ref{eq:mock-G+-}) we get
    \begin{equation}
    \label{eq:mock-G-bd-1}
    i\langle f, (Q^{+\pm}_1-Q^{-\mp}_2)g\rangle_{L^2_{\g}}
    =
    2\lambda\Big(\langle a^{-\pm}_1, b^{-\pm}_1\rangle_{L^2_{\h}} \pm \langle a^{-\pm}_2, b^{-\pm}_2\rangle_{L^2_{\h}}\Big),
    \end{equation}
    \begin{equation}
    \label{eq:mock-G-bd-2}
    i\langle f, (Q^{+\pm}_2-Q^{-\mp}_1)g\rangle_{L^2_{\g}}
    =
    2\lambda\Big(\langle a^{+\mp}_1, b^{+\mp}_1\rangle_{L^2_{\h}} \pm \langle a^{+\mp}_2, b^{+\mp}_2\rangle_{L^2_{\h}}\Big),
    \end{equation}
    compared to 
    \begin{equation}
    \label{eq:global-G-bd}
    i\langle f,G_{+\pm}g\rangle_{L^2_{\g}}
    =
    2\lambda\Big(\langle  a^{+\pm}_1, b^{+\pm}_1\rangle_{L^2_{\h}}\pm \langle a^{+\pm}_2, b^{+\pm}_2\rangle_{L^2_{\h}}\Big)
    =
    2\lambda\Big(\langle  a^{-\mp}_1, b^{-\mp}_1\rangle_{L^2_{\h}}\pm \langle a^{-\mp}_2, b^{-\mp}_2\rangle_{L^2_{\h}}\Big).
    \end{equation}

    \begin{table}[H]
        \centering
        \begin{tabular}{|c||c|c|c|c|}
        \hline
             & $\pi_1\rho^+$ & $\pi_1\rho^-$ & $\pi_2\rho^+$ & $\pi_2\rho^-$\\
             \hline\hline
            $P_{++}^{-1}f$  & $a^{++}_1$ & $0$ & $a^{++}_2$ & $0$\\
            \hline
           $P_{+-}^{-1}f$  & $a^{+-}_1$ & $0$ & $0$ & $a^{+-}_2$\\
           \hline
           $(P_{++}^{-1}-P_{+-}^{-1})f$  & $a^{++}_1-a^{+-}_1$ & $0$ & $a^{++}_2$ & $-a^{+-}_2$\\
           \hline
        \end{tabular}
        \caption{Components of boundary data of the images of $f\in\dot{C}^{\infty}(M)$ under inverses which have the same propagation properties in one component of $\Sigma$.}
        \label{tab:G-diff}
    \end{table}
    
    To characterize the discrepancy, let us consider separately scattering within a single component of the boundary and cross-component scattering. Let $\mathcal{S}^{-\alpha\beta}_{\alpha\beta}=\rho^{-\alpha\beta}\mathcal{U}^{\alpha\beta}_0$. For simplicity, we consider a particular example, the difference
    \begin{equation}
    \begin{split}
    \label{eq:G-diff}
        i\Big(\langle f,G_{++}g\rangle_{L^2_{\g}}-\langle f, (Q^{++}_2-Q^{--}_1)g\rangle_{L^2_{\g}}\Big)
        = \\
        2\lambda\Big( 
        \big(\langle a_1^{++},b_1^{++}\rangle_{L^2_{\h}} -\langle a_1^{+-},b_1^{+-}\rangle_{L^2_{\h}}\big)
        +
        \big(\langle a_2^{++},b_2^{++}\rangle_{L^2_{\h}} -\langle a_2^{+-},b_2^{+-}\rangle_{L^2_{\h}}\big)
        \Big).
    \end{split}
    \end{equation}
    The inverses $P_{++}^{-1}$ and $P_{+-}^{-1}$ differ by an operator which is completely regularizing on $\Sigma_1\backslash R^+$. This is reflected on the level of boundary data by the fact that $\pi_1\rho^-(P_{++}^{-1}-P_{+-}^{-1})f=0$ (see Table~\ref{tab:G-diff}). Therefore, $(P_{++}^{-1}-P_{+-}^{-1})f=-\mathcal{U}_0^{--}(a^{+-}_2)$. Then we can write
    \[a_1^{++}-a_1^{+-}=\pi_1\mathcal{S}^{++}_{--}(a_2^{+-}),\]
    so $a_1^{++}-a_1^{+-}$ is nonzero purely due to the presence of cross-component scattering, and the same can be said about the first term on the right-hand side of Eq.~(\ref{eq:G-diff}) (since the argument applies to $b_1^{+-}$ as well). We remark that this cross-component scattering amplitude $\pi_1\mathcal{S}^{++}_{--}(a_2^{+-})$ is {\em always} nonzero when $a_2^{+-}$ is such, for if it vanished, $\mathcal{U}_0^{--}(a^{+-}_2)$ would have both asymptotic data pieces vanish at $S_1$, and thus in fact would be rapidly decaying there, but then unique continuation at infinity and then standard unique continuation would in fact show that $\mathcal{U}_0^{--}(a^{+-}_2)$ vanished identically. This is in contrast with the Lorentzian setting, in which timelike Killing vector fields allow trivial cross-component scattering.
    
    In the other component, we get
    \[a_2^{+-}=-\mathcal{S}^{--}_{++}(a_2^{++})-\mathcal{S}^{--}_{++}(a_1^{++}-a_1^{+-}),\]
    and similarly for $b_2^{+-}$. The second term is the result of cross-component scattering; in its absence, we would have $\langle a_2^{+-},b_2^{+-}\rangle=\langle \mathcal{S}^{--}_{++}(a_2^{++}),\mathcal{S}^{--}_{++}(b_2^{++})\rangle=\langle a_2^{++}, b_2^{++}\rangle$ due to the unitarity of the scattering matrix. Thus, the second term on the right-hand side of Eq.~(\ref{eq:G-diff}) would also be zero in the absence of cross-component scattering.
    
    A similar analysis applies to the other three cases. Thus we see that the difference between the Hermitian forms in Eqs.~(\ref{eq:mock-G-bd-1})-(\ref{eq:global-G-bd}) is due to the presence of nontrivial cross-component scattering.

\printbibliography
\end{document}